\documentclass[10pt,letterpaper]{article}

\usepackage[utf8]{inputenc}
\usepackage{fullpage}
\usepackage{authblk}
\usepackage{ulem}

\usepackage[bookmarks]{hyperref}
\hypersetup{colorlinks=true,citecolor=blue,linkcolor=blue,filecolor=blue,urlcolor=blue}

\usepackage{amsmath,amsthm,amsfonts,amssymb,mathtools}

\usepackage{graphicx}
\usepackage{caption}
\usepackage{tikz}
\usetikzlibrary{quantikz}
\usetikzlibrary{external}
\usetikzlibrary{shapes,decorations,arrows,calc,arrows.meta,fit,positioning,external}
\tikzset{
    terminal/.style={},
    event/.style={draw, circle, fill = black, minimum size = 0.12cm, inner sep=0pt},
    source/.style ={ellipse, draw, minimum width = 0.5 cm, color=classical_gray, fill=classical_gray!10, text=black},
    graph_node/.style ={circle, draw, minimum width = 0.5 cm, color=quantum_purple, fill=quantum_purple!20, text=black},
    qsource/.style ={ellipse, draw, minimum width = 0.5 cm, color=prep_green, fill=prep_green!20, text=black},
    dev/.style={rectangle, rounded corners, draw, minimum width = 0.7 cm, color=classical_gray, fill=classical_gray!10, text=black},
    prep_dev/.style={dev, color=prep_green, fill=prep_green!20, text=black},
    proc_dev/.style={dev, color=proc_red, fill=proc_red!20, text=black},
    meas_dev/.style={dev, color=meas_blue, fill=meas_blue!20, text=black},
    el/.style = {align=left},
    meas_gate/.style={color=meas_blue, fill=meas_blue!20},
    prep_gate/.style={color=prep_green, fill=prep_green!20},
    proc_gate/.style={color=proc_red, fill=proc_red!20},
    meas_gate_group/.style={dashed,rounded corners,color=meas_blue},
    proc_gate_group/.style={dashed,rounded corners,color=proc_red},
    prep_gate_group/.style={dashed,rounded corners,color=prep_green},
}

\newcommand{\cedge}{edge[double, line width=1pt, double distance=1.5pt, arrows = {-Latex[length=0.5pt 2.5 0]}]}
\newcommand{\qedge}{edge[line width=2pt, arrows = {-Latex[length=6pt 1.5 0]}]}

\definecolor{prep_green}{HTML}{669933}
\definecolor{proc_red}{HTML}{CC3333}
\definecolor{meas_blue}{HTML}{3366CC}
\definecolor{quantum_purple}{HTML}{663366}
\definecolor{classical_gray}{HTML}{666666}

\usepackage{tabularx}
\usepackage{hhline}

\usepackage{verbatim}
\usepackage{dsfont}

\usepackage{enumitem}

\usepackage{algorithm}

\makeatletter
\newenvironment{breakablealgorithm2}[1][htb]
{% \begin{breakablealgorithm}
	\begin{flushleft}
		\refstepcounter{protocol}% New algorithm
		\hrule height.8pt depth0pt \kern2pt% \@fs@pre for \@fs@ruled
		\renewcommand{\caption}[2][\relax]{% Make a new \caption
			{\raggedright\textbf{\fname@algorithm~\thealgorithm} ##2\par}%
			\ifx\relax##1\relax % #1 is \relax
			\addcontentsline{loa}{algorithm}{\protect\numberline{\thealgorithm}##2}%
			\else % #1 is not \relax
			\addcontentsline{loa}{algorithm}{\protect\numberline{\thealgorithm}##1}%
			\fi
			\kern2pt\hrule\kern2pt
		}
	}{% \end{breakablealgorithm}
		\kern2pt\hrule\relax% \@fs@post for \@fs@ruled
	\end{flushleft}
}
\makeatother

\newcounter{protocol}
\makeatletter
\newenvironment{protocoldesc}{%
	\renewcommand{\ALG@name}{Protocol}% Update algorithm name
		\let\c@algorithm\c@protocol
	\begin{breakablealgorithm2}%
	}{\end{breakablealgorithm2}
}

\usepackage{physics}

\usepackage{tikz}
\usetikzlibrary{quantikz}
\usetikzlibrary{external}

\usepackage[most]{tcolorbox}

\usepackage[sort&compress,numbers]{natbib}
\newcommand{\cC}{\mathcal{C}}
\newcommand{\cD}{\mathcal{D}}
\newcommand{\cE}{\mathcal{E}}
\newcommand{\cF}{\mathcal{F}}
\newcommand{\cG}{\mathcal{G}}

\newcommand{\cK}{\mathcal{K}}
\newcommand{\cL}{\mathcal{L}}
\newcommand{\cM}{\mathcal{M}}
\newcommand{\cN}{\mathcal{N}}

\newcommand{\cP}{\mathcal{P}}
\newcommand{\cQ}{\mathcal{Q}}
\newcommand{\cR}{\mathcal{R}}

\newcommand{\cT}{\mathcal{T}}
\newcommand{\cU}{\mathcal{U}}
\newcommand{\cV}{\mathcal{V}}
\newcommand{\cW}{\mathcal{W}}
\newcommand{\cX}{\mathcal{X}}
\newcommand{\cY}{\mathcal{Y}}
\newcommand{\cZ}{\mathcal{Z}}

\theoremstyle{definition}
\newtheorem{theorem}{Theorem}
\newtheorem{observation}{Observation}
\newtheorem{lemma}{Lemma}

\newtheorem{definition}[lemma]{Definition}

\newtheorem{corollary}[lemma]{Corollary}

\newtheorem{example}[lemma]{Example}
\newtheorem{remark}[lemma]{Remark}
\newtheorem{proposition}[lemma]{Proposition}

\newcommand{\msf}{\mathsf}
\newcommand{\mbf}{\mathbf}
\newcommand{\mbb}{\mathbb}
\newcommand{\wt}{\widetilde}
\newcommand{\ol}{\overline}

\newcommand{\ve}{\varepsilon}
\newcommand{\net}{\mathrm{net}}

\newcommand{\Lin}{\cL}

\newcommand{\Pos}{\mathrm{Pos}}

\newcommand{\Channel}{\mathrm{C}}
\newcommand{\Unitary}{\mathrm{U}}
\newcommand{\Density}{\mathrm{D}}
\newcommand{\id}{\mathrm{id}}

\newcommand{\BSM}{\mathrm{BSM}}

\begin{document}

\title{Communication Advantages %with Quantum Multiaccess Networks
from Quantum Dense Network Coding}

\author[a]{Ian George}
\author[b]{Brian Doolittle}

\affil[a]{Centre for Quantum Technologies, National University of Singapore, Singapore 117543, Singapore}
\affil[b]{Aliro Technologies, Inc., Brighton, Massachusetts, 02135, USA}

\date{\today}

\maketitle

\begin{abstract}
    A central problem in quantum information theory is understanding how quantum resources can be used to communicate information more efficiently than classical resources. We introduce quantum dense network coding--- a protocol that transmits the output of a non-Boolean function to a receiver using provably
    half as many qubits as bits for each sender by not transmitting the entirety of the function inputs. We show this advantage requires both shared entanglement and quantum communication, is robust to noise, and the gap in success probability between quantum and classical communication can be amplified exponentially in the number of senders.
    Finally, we show that dense network coding gives rise to a novel, information-theoretically secure, quantum cryptographic protocol, which we call measurement-device-independent quantum key growing.
\end{abstract}

%Some requirements:
% - Up to 70 references in main text
% - No more than 10 display items total
% - Total main text at most 5000 words

\section{Introduction}
%Why we are doing this
% A central question in quantum information theory is when can information be communicated using less quantum resources than would be necessary using classical resources.
% A central problem in quantum information theory is understanding how quantum resources can be used to communicate information more efficiently than classical resources.
%A main area of study in quantum information theory is the fundamental limit on the transmission of information from a sender to a receiver using quantum resources.
% Such cases are known to arise by using entanglement and quantum communication.
Quantum communication resources can offer an advantage over classical communication in which information can be transmitted using fewer qubits than bits.
For example, when a pair of maximally entangled qubits (ebit) is shared between the sender and receiver of a quantum channel, a protocol known as superdense coding can be used to transmit two bits of information from sender to receiver using one qubit of communication \cite{bennet1992_dense_coding}. 
% For example, in superdense coding \cite{bennet1992_dense_coding}, the sender and receiver share a pair of maximally entangled qubits sometimes known as an ebit.
% By the sender acting on their half of the ebit and then sending it to the receiver over a quantum channel, the sender can transmit 2 bits of information using 1 qubit of communication.
Superdense coding demonstrates that that one ebit combined with one qubit of communication is more powerful than two bits of communication \cite{Devetak-2008a}. Motivated by this communication advantage, strict limits on communication with quantum resources have been established in the point-to-point communication setting \cite{Holevo1973BoundsFT, EA-capacity, Cubitt2011_NS-capacity, Frenkel2015_classical_information_n-level_quantum_system,dallarno2017_no_hypersignaling,doolittle_2021_certifying_classical_simulation_cost, chitambar2023_communication_value}.

Much less is known about the fundamental limits on communication that arise in a \textit{network} when quantum resources are available. Nevertheless, prominent examples of communication advantages have been found in communication networks with multiple senders and a single receiver \cite{Cleve1999_ea_inner_product,Bowles2015_nonclassicality_communication_networks,Zhang2022_single_particle_mac}, which we refer to as multiaccess networks (MNs) (see Fig.~\ref{fig:MN_dags}).
Notably, Leditzky et al. \cite{Leditzky2020_mac_games_ea_cmac} studied the MN in which entanglement is used to preprocess two independent inputs to the receiver's classical multiple access channel (MAC). The authors show that if receiver's MAC penalizes incorrect answers to a non-local game, then when entanglement shared between senders improves the ability to win the non-local game, the amount of information that can be transmitted to the receiver using the specified MAC can be improved. Separately, Buhrman et al.~\cite{buhrman2001quantum} studied the MN in which two senders can transmit qubits to the receiver. The authors found that the senders need to transmit exponentially fewer qubits than bits to the receiver for the receiver to determine with high probability whether or not the two senders hold equivalent bit strings.
These results suggest a rich theory of communication advantages of networks using quantum resources over their classical counterparts.
% However, much less is known about the fundamental limits on communication that arise in a \textit{network} when quantum resources are available, i.e.~much less is known about quantum network information theory.
%Background 

Recently, Doolittle et al. \cite{doolittle2024operational_nonclassicality} developed a framework for studying communication advantages in quantum networks, and numerically surveyed communication advantages over a broad range of communication network topologies and quantum resource configurations. The survey found that the strongest communication advantages occurred in networks that have many senders and a single receiver. 
Specifically, the authors provide an example in which the receiver computes the bitwise XOR between each sender's two-bit inputs where the two senders are each allowed one qubit of communication to the receiver and may use a shared ebit to assist with their joint encoding. Remarkably, these quantum resources allow this computation to be performed without error, which would classically require two-bits of communication from each sender.

% \bdedit{Our goal is to identify notable examples of these communication advantages that can be achieved by entanglement-assisted quantum communication networks.}
In this work, we generalize the bitwise XOR protocol from reference \cite{doolittle2024operational_nonclassicality} to develop dense network coding, a protocol in which certain non-Boolean functions are computed over an entanglement-assisted quantum MN (see Fig.~\ref{fig:MN_dags}.d) without needing the senders to share their entire input with the receiver. 
Dense network coding offers greater communication efficiency than classical network coding \cite{ahlswede_2000_network_coding,ho2008_network_coding} because entanglement between the senders enables each sender to communicate half as many qubits as the number of bits classical schemes would require while still computing the function correctly.
%a quadratic reduction in the number of communicated qubits, as compared to bits, that each sender needs to transmit for the function to be correctly computed.
Although similar  advantages can be achieved by each sender utilizing superdense coding with the receiver, dense network coding is distinct because the receiver does not share any entanglement with the senders. Remarkably, we prove that the dense network coding requires entanglement, and that the communication advantage disappears if either quantum communication or shared entanglement is unavailable, leading to a large error in the receiver's output. We use this to further show the communication advantage is noise-robust and could be applied in early quantum communication networks. Finally, we show quantum dense network coding gives rise to the security of what we call measurement-device-independent quantum key growing, thereby establishing its relation to quantum cryptography.

\section{Results}
%Divided by subheadings less than 60 characters (incl spaces) that do not contain punctuation
\subsection{Multiaccess Networks and Signaling Dimension}
Before stating the results, we briefly review multiaccess networks (MNs) and multiple access channels (MACs) where formal definitions are found in Appendix~\ref{app:preliminaries}.
As depicted in Fig.~\ref{fig:MN_dags}, a two-sender MN consists of two senders and a receiver that have a fixed set of communication resources available. We consider the following configurations of communication resources:
\begin{itemize}
    \item \textbf{Classical MN:} Sender 1 (resp.~2) is connected to the receiver by a noiseless $n_{1}$-bit (resp.~$n_{2}$-bit) classical channel. The receiver then processes the senders' signals. We denote the set of classical MNs $\mbf{C}(2^{n_{1}},2^{n_{2}})$ (see Fig.~\ref{fig:MN_dags}.a).
    \item \textbf{Quantum MN:} Sender 1 (resp.~2) is connected to the receiver by a noiseless $n_{1}$-qubit (resp.~$n_{2}$-qubit) quantum channel. The set of quantum MNs is denoted $\mbf{Q}(2^{n_{1}},2^{n_{2}})$ (see Fig.~\ref{fig:MN_dags}.b).
    \item \textbf{Entanglement-Assisted (resp.~Nonsignaling-Assisted)  Classical MN:} The two senders of a classical MN use a shared entangled quantum state (resp. bipartite non-signaling box \cite{popescu1994quantum_pr_box, on-quantum-ns-boxes}) to encode their transmitted classical messages. We denote the set of entanglement-assisted MNs by $\mbf{C}^{E}(2^{n_{1}},2^{n_{2}})$ (resp.~nonsignaling-assisted MNs by~$\mbf{C}^{N}(2^{n_{1}},2^{n_{2}})$) (see Fig.~\ref{fig:MN_dags}.c).
    \item \textbf{Entanglement-Assisted Quantum MN:} The two senders of a quantum MN use a shared entangled quantum state to encode their transmitted quantum states. We denote the set of entanglement-assisted quantum MNs as $\mbf{Q}^{E}(2^{n_{1}},2^{n_{2}})$ (see Fig.~\ref{fig:MN_dags}.d).
\end{itemize}

\begin{figure}
    \centering
    \small
    \begin{tabular}{l c l}
        {\normalsize (a)}  &  & {\normalsize (b) } \\
        \begin{tikzpicture} %CMN
            \node[terminal] (x) at (-1.25,1) {$\cX \ni x$};
            \node[terminal] (y) at (-1.25,-1) {$\cY \ni y$};    
            \node[dev] (A) at (0.75,1) {$P_{X' \vert X}$};
            \node[dev] (B) at (0.75,-1) {$Q_{Y' \vert Y}$};
            \node[dev, minimum height = 1cm, minimum width = 1cm] (R) at (2.75, 0) {$R_{Z \vert X' Y'}$};
            \node[terminal] (z) at (4.25, 0) {$z$};
    
            \path (x) \cedge (A);
            \path (y) \cedge (B);
            \path (A) \cedge node[below,pos=0.25] {$d_{1}$} (R);
            \path (B) \cedge node[above,pos=0.25] {$d_{2}$} (R);
            \path (R) \cedge (z);
        \end{tikzpicture} & & 
        \begin{tikzpicture} %QMN
            % x node set with absolute coordinates
            \node[terminal] (x) at (-1,1) {$\cX \ni x$};
            \node[terminal] (y) at (-1,-1) {$\cY \ni y$};    
            \node[qsource] (A) at (1.1,1) {$\rho^{x}_{A'}$};
            \node[qsource] (B) at (1.1,-1) {$\sigma^{y}_{B'}$};
            \node[meas_dev, minimum height = 1cm] (R) at (2.75, 0) {$\{\Pi_{z}\}_{z}$};
            \node[terminal] (z) at (4, 0) {$z$};
        
            \path (x) \cedge (A);
            \path (y) \cedge (B);
            \path (A) \qedge node[below,pos=0.25] {$d_{1}$} (R);
            \path (B) \qedge node[above,pos=0.25] {$d_{2}$} (R);
            \path (R) \cedge (z);
        \end{tikzpicture} \\
        \hfill \\
        {\normalsize (c)} & & {\normalsize (d)} \\
        \begin{tikzpicture} %NSFCMN
            \node[terminal] (x) at (-1,0.8) {$\cX \ni x$};
            \node[terminal] (y) at (-1,-0.8) {$\cY \ni y$};    
            \node[prep_dev, minimum height = 2.5cm] (AB) at (1.25,0) {$P_{X'Y' \vert XY}$}; 
            \node[dev, minimum height = 1cm] (R) at (3.5, 0) {$R_{Z \vert X'Y'}$};
            \node[terminal] (z) at (5, 0) {$z$};

            \path (x) \cedge (0.45,0.8) ;
            \path (y) \cedge (0.45,-0.8) ;
            \path (2.05,0.8) \cedge node[below,pos=0.25] {$d_{1}$} (R);
            \path (2.05,-0.8) \cedge node[above,pos=0.25] {$d_{2}$} (R);
            \path (R) \cedge (z);
        \end{tikzpicture} & &   \begin{tikzpicture}
            % x node set with absolute coordinates
            \node[terminal] (x) at (-1,1) {$\cX \ni x$};
            \node[qsource] (rho) at (0,0) {$\rho_{AB}$};
            \node[terminal] (y) at (-1,-1) {$\cY \ni y$};    
            \node[proc_dev] (A) at (1.1,1) {$\cE_{A \to A'}^{x}$};
            \node[proc_dev] (B) at (1.1,-1) {$\cF_{B \to B'}$};
            \node[meas_dev, minimum height = 1cm] (R) at (2.75, 0) {$\{\Pi_{z}\}_{z}$};
            \node[terminal] (z) at (4, 0) {$z$};
        
            \path (rho) \qedge (A);
            \path (rho) \qedge (B);
            \path (x) \cedge (A);
            \path (y) \cedge (B);
            \path (A) \qedge node[below,pos=0.25] {$d_{1}$} (R);
            \path (B) \qedge node[above,pos=0.25] {$d_{2}$} (R);
            \path (R) \cedge (z);
        \end{tikzpicture}\\
    \end{tabular}
    \caption{Examples of Multiaccess Networks with Classical Inputs and Outputs. Double-lined arrows denote classical communication and single-lined arrows denote quantum communication. Communication limited by signaling dimension $d_{i}=2^{n_i}$ is specified by the arrow being labeled by $d_{i}$. (a) Classical multiaccess network. (b) Quantum multiaccess network. (c) Nonsignaling-assisted classical multiaccess network. (d) Entanglement-assisted Quantum multiaccess network.
    }
    \label{fig:MN_dags}
\end{figure}
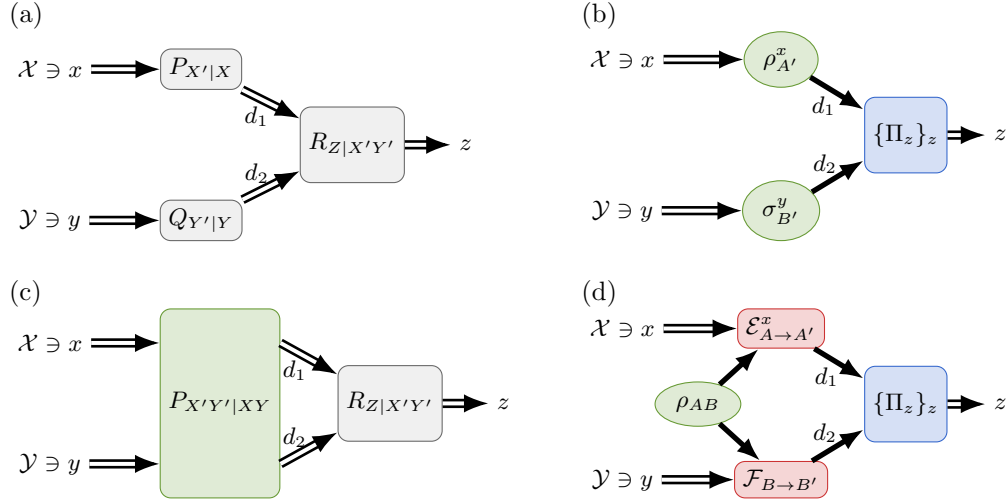

We parameterize each family of MNs by the signaling dimension pair $(2^{n_{1}},2^{n_{2}})$ where $n_1$ and $n_2$ specify the number of bits/qubits transmitted to the receiver from each respective sender. The signaling dimension of a channel is a resource agnostic quantity that specifies the amount of transmitted classical information. More precisely, given a quantum or classical channel with classical inputs and outputs, the signaling dimension of the channel is the minimum amount of classical communication needed to simulate the behavior of the channel for a uniformly random input where it is assumed that classical randomness is shared by the sender and receiver \cite{dallarno2017_no_hypersignaling,doolittle_2021_certifying_classical_simulation_cost}. When quantum communication networks have classical inputs and outputs, the signaling dimension serves as a natural way to establish constraints on the amount of communication, allowing quantum communication advantage to be witnessed as a violation of the resulting classical bounds \cite{Bowles2015_nonclassicality_communication_networks,doolittle2024operational_nonclassicality}. For intuition, the channel induced by superdense coding allows two-bits of classical information to be transmitted, using only a single qubit of communication from sender to receiver. Thus, when a maximally entangled pair assists a qubit channel with signaling dimension $d=2$, the signaling dimension of the induced channel increases to $d=2^2$. Alternatively, if a classical channel with signaling dimension $d=2$ attempted to transmit two-bits of information, it would do so with a probability of success of one-half. Therefore, the signaling dimension and success probability each play an important role in quantifying communication advantage.

% This hardware agnostic approach is used in prior works and often quantified by the \textit{signaling dimension} of the channel where an $n$-bit/qubit channel has a signaling dimension $d=2^n$ \cite{dallarno2017_no_hypersignaling,doolittle_2021_certifying_classical_simulation_cost,doolittle2024operational_nonclassicality}.
% For intuition, entanglement is known to increase the signaling dimension of channel the signaling dimension is a straightforward network generalization of how one compares the number of qubits sent with entanglement-assistance to the number of bits sent when considering superdense coding.

To identify what resources and signaling dimension of a MN are necessary to implement certain functions in a distributed manner, we use the formalism of MACs. A two-sender MAC with input alphabets $\cX$ and $\cY$ and output alphabet $\cZ$ is represented by a conditional distribution $P_{Z \vert XY}$. Computing a function over two senders and a receiver is thus a special case of a MAC. Concretely, given a function $f: \cX \times \cY \to \cZ$, we define the MAC $P^{f}_{Z \vert XY}(z \vert x,y) = \delta_{z,f(x,y)}$ for all $x,y,z$. Similarly, for a two-sender MN where Senders 1 and 2 receive inputs $\cX$ and $\cY$ and the receiver outputs $\cZ$, the total implementation of the MN induces a MAC represented by a conditional distribution $Q_{Z|XY}$. For uniform inputs, the probability that a MN successfully implements a function $f: \cX \times \cY \to \cZ$ is $P^{f}_{S}(Q_{Z \vert XY}) \coloneq \frac{1}{\vert \cX \vert \vert \cY \vert}\sum_{x,y} \delta_{z,f(x,y)}Q_{Z \vert XY}(z \vert x,y)$ where $0 \leq P^{f}_{S}(Q_{Z \vert XY})\leq 1$ and this value is one if and only if $Q_{Z \vert XY} = P^{f}_{Z \vert XY}$ \cite{doolittle2024operational_nonclassicality}. Thus $P^{f}_{S}(Q_{Z \vert XY})=1$ means the MN perfectly implements a distributed computation of the function $f$. By optimizing over all MNs of a specified resource configuration and signaling dimension, we can determine how well on average a given choice of resources and signaling dimension can compute a function $f$ in a distributed manner. For example,
\begin{equation}\label{eq:main-text-channel-sim-measure}
    P_{S}^{f}(\mbf{Q}(2^{n_{1}},2^{n_{2}})) \coloneq \sup_{Q_{Z|XY}\in \mbf{Q}(n_{1},n_{2})} P_{S}^{f}(Q_{Z\vert XY})
\end{equation}
determines how well a quantum MN with signaling dimensions $(2^{n_{1}},2^{n_{2}})$ can implement a distributed computation of $f$. 

\subsection{Quantum Dense Network Coding}\label{sec:main-text-dense-network-coding}
Our first result is the protocol for dense network coding (DNC). While DNC supports various functions, in the main text
we present DNC for computing
addition modulo $2^{n}$ for two pairs of integers in $(q,r),(s,t)\in\mbb{Z}^{\times2}_{2^{n}}$. For integers $x,y\in\mbb{Z}_{2^{n}}$, we denote ditwise addition modulo $2^{n}$ by
the function
\begin{align}
    \oplus_{2^{n}} : \mbb{Z}_{2^{n}} \times \mbb{Z}_{2^{n}} \to \mbb{Z}_{2^{n}} \quad x \oplus_{2^{n}} y = x + y \mod 2^{n} \ , 
\end{align}
and we denote the application of the function to $k$ pairs of integers by $\oplus^{k}_{2^{n}}$. For example, $(q,r) \oplus^{2}_{2^{n}} (s,t) = (q \oplus_{2^{n}} s, r \oplus_{2^{n}} t)$ for $q,r,s,t \in \mbb{Z}_{2^{n}}$.

Protocol \ref{prot:main-text-dense-network-coding-example} describes DNC of $\oplus_{2^{n}}^{2}$ for $n \in \mbb{N}$. This protocol makes use of the set of discrete Weyl operators, a.k.a.~generalized Pauli operators, for integer $d$:
\begin{align}\label{eq:discrete-weyl-operators}
        \{W_{q,r} \coloneq X^{q}V^{r} \}_{(q,r) \in \mbb{Z}_{d} \times \mbb{Z}_{d}} \ , 
\end{align}
which are defined using the shift and phase operators:\begin{align}\label{eq:phase-and-shift-operators}
        X \coloneq \sum_{q \in \mbb{Z}_{d}} \ket{q \oplus_{d} 1}\bra{q} \quad Z \coloneq \sum_{q \in \mbb{Z}_{d}} \zeta^{q} \dyad{q} \ ,  
\end{align}
where $\zeta \coloneq \exp(2\pi i/d)$. Protocol \ref{prot:main-text-dense-network-coding-example} also makes use of the generalized Bell measurement, which is defined as the projectors onto the basis
\begin{align}
    \left\{\ket{\Phi_{q,r}} \coloneq W_{q,r} \otimes I_{d}\ket{\Phi}\right\}_{(q,r) \in \mbb{Z}_{d} \times \mbb{Z}_{d}} \ , 
\end{align}
where $\ket{\Phi} = \frac{1}{\sqrt{d}} \sum_{i \in \mbb{Z}_{d}} \ket{i}\ket{i}$ is the maximally entangled state on $\mbb{C}^{d} \otimes \mbb{C}^{d}$.

\begin{protocoldesc}[H]
\caption{Dense Network Coding for Ditwise Addition Modulo $2^{n}$}\label{prot:main-text-dense-network-coding-example}
		\textbf{Inputs:} \\
		\hspace{0.5cm}
		\begin{tabular}{ l l}
			$(q,r) \in \mbb{Z}_{2^{n}}^{\times 2}$ & Sender 1's Inputs  \\[2mm]
            $(s,t) \in \mbb{Z}_{2^{n}}^{\times 2}$ & Sender 2's Inputs \\[2mm]
            $\ket{\Phi} \in \mbb{C}^{2^{n}} \otimes \mbb{C}^{2^{n}}$ & Maximally Entangled State Shared by Senders
		\end{tabular}
		
		\vspace{0.5cm}
		
		\textbf{Protocol:} 
		\begin{enumerate}
            \item On input $(q,r)$, Sender 1 applies $W_{q,r}$ to their local system of $\ket{\Phi}$.
            \item On input $(s,t)$, Sender 2 applies $W_{s,t}^{T}$, to their local system of $\ket{\Phi}$.
            \item Senders 1 and 2 forward their systems to the receiver.
            \item The receiver applies the generalized Bell measurement to its received systems to obtain outcome $(x,y)$, which it outputs as the answer.
        \end{enumerate}
\end{protocoldesc}
\begin{theorem}\label{thm:main-text-dense-coding}
    Protocol \ref{prot:main-text-dense-network-coding-example} implements $\oplus^{2}_{2^{n}}$ perfectly, i.e. $P_{S}^{\oplus^{2}_{2^{n}}}(\mbf{Q}^{E}(2^n,2^n)) = 1$. Moreover, if either party communicates strictly less than $n$ qubits, the success probability is bounded away from 1.
\end{theorem}
\begin{proof} 
As Protocol \ref{prot:main-text-dense-network-coding-example} is manifestly a strategy using entanglement-assistance and quantum communication, it suffices to show that the receiver always outputs the correct value. We do this by calculating the probability of the correct output. We first determine the state transmitted to the receiver when Senders 1 and 2 are given inputs $(q,r)$ and $(s,t)$ respectively. As follows from Items 1 through 3 of Protocol \ref{prot:main-text-dense-network-coding-example}, the transmitted state is
\begin{align}
    W_{q,r} \otimes W_{s,t}^{T}\ket{\Phi} = W_{q,r}W_{s,t} \otimes I_{B} \ket{\Phi} = \zeta^{rs}W_{q+s,r+t} \otimes I_{B} \ket{\Phi} \ , 
\end{align}
where the first equality follows from the transpose trick and the second equality follows from implicitly letting addition be modulo $2^n$ and applying \cite[Eq.~4.75]{WatrousBook}. In Item 4 of Protocol \ref{prot:main-text-dense-network-coding-example}, the receiver applies the generalized Bell state measurement. We can therefore calculate the probability of each outcome of the receiver on the state it receives when the input is $(q,r,s,t)$:
\begin{align}
    \Pr[x,y \vert q,r,s,t] &= \vert \bra{\Phi_{x,y}} (W_{q,r} \otimes W_{s,t}^{T})\ket{\Phi} \vert^{2} \\
    &= \vert \bra{\Phi} W^{\dagger}_{x,y}W_{q+s,r+t} \otimes I_{B} \ket{\Phi} \vert^{2} \\
    &= \frac{1}{2^{2n}} \vert \Tr[W^{\dagger}_{x,y}W_{q+s,r+t}] \vert^{2} \\
    &= \begin{cases}
        1 & (x,y) = (q+s,r+t) \\
        0 & \text{otherwise}
    \end{cases}
\end{align}
where the third equality is the fact $\bra{\Phi}X \otimes I \ket{\Phi} = \frac{1}{2^n}\Tr[X]$ as may be verified by direct calculation and the fourth uses (c.f.~\cite[Eq. 4.78]{WatrousBook}):
\begin{align}
    \Tr[W_{q,r}^{\dagger} W_{s,t}] = \begin{cases}
        2^n & (q,r) = (s,t) \\
        0 & \text{otherwise} \ .
    \end{cases}
\end{align}
Thus, the receiver's outcome is $(q+s, r+t) = (q,r) \oplus^{2}_{2^n} (s,t)$ always, so the function is computed. The proof that $\oplus^{2}_{2^{n}}$ cannot be computed if either party sends strictly less than $n$ qubits is shown in Appendix \ref{app:dense-network-coding}.
\end{proof}

\subsubsection{Dense Network Coding Other Functions and Algebraic Structure} % This paragraph basically just summarizes the rest of the insights of the appendix on dense network coding
Theorem \ref{thm:main-text-dense-coding} shows the ability to compute addition modulo $2^{n}$ for two pairs of integers in $\mbb{Z}_{2^{n}}$. Interestingly, with appropriate pre- and post-processing, Protocol \ref{prot:main-text-dense-network-coding-example} is used to achieve the rate 1 quantum private information retrieval protocol of Song and Hayashi \cite[Section III.B]{Song-QPIR-2021a}. This already shows DNC has further applications. However, as the proof highlights, the ability to compute this function stems from the relation between the discrete Weyl operators acting on vector space $\mbb{C}^{2^{n}}$ and the function $\oplus^{2}_{2^{n}}$. Mathematically, the relation being utilized is the discrete Weyl operators acting on $\mbb{C}^{d}$ are a projective unitary representation of the product group $\mbb{Z}_{d} \times \mbb{Z}_{d}$. We thus show in Appendix \ref{app:dense-network-coding} that one can dense network code the multiplication of group elements of any group $(G,\cdot)$ of order $d^{2}$ that admits a projective unitary representation acting on vector space $\mbb{C}^{d}$. We call such a group a `tightly network codeable' group,\footnote{The modifier `tightly' follows from sending less qubits would lead to error (Theorem \ref{thm:main-text-dense-coding}) as well as the observation that demanding the group be order $d^{2}$ and the vector space be of dimension $d$ is not necessary for dense network coding--- it merely allows us to enjoy a particularly large communication advantage between the quantum and classical network.} and we show another example of such a group corresponds to computing the bitwise XOR of bitstrings of even length. The algebraic identification of tightly network codeable groups allows us to separate dense network coding from point-to-point dense coding as Werner identified weaker necessary and sufficient algebraic conditions for the point-to-point case \cite{werner2001-dense-coding}. Interestingly, we also show that tightly network codeable groups are in exact correspondence with what Knill called nice unitary error bases \cite{knill1996group,klappenecker2002beyond} and thus have some relation to quantum error correcting codes through algebraic considerations.

\subsubsection{Alternative Topology of Shared Entanglement and Compression of Information} An alternative network setup to compute $f$ exists other than DNC: if each sender shares a maximally entangled state with the \textit{receiver}, then each sender could use superdense coding to send their respective input to the receiver. The receiver then knows all the inputs and can compute the target function. This alternative protocol differs significantly from DNC. First, it requires twice as much shared entanglement as DNC. Thus, it is more resource intensive to achieve the same goal. Second, we show in Appendix \ref{app:dense-network-coding} that when the inputs to the protocol are uniformly random, the receiver in dense network coding cannot guess the value of the inputs with high probability, so the inputs cannot have been transmitted. In contrast, if each sender does superdense coding with the receiver, the receiver learns all of the inputs. This difference in information that the receiver holds may then be seen as an explanation for the difference in entanglement: dense network coding communicates just enough information to output the correct value of the function rather than all information in the inputs.

\subsection{Communication Advantage of Dense Network Coding}\label{sec:comm-adv-in-dist-comp}
Theorem \ref{thm:main-text-dense-coding} gives a strategy for computing the function $\oplus^{2}_{2^{n}}$ in a distributed manner using entanglement-assistance and quantum communication. To truly be  `dense,' it ought to use less communication than would be needed using classical resources. Unlike superdense coding where the communication advantage over classical communication follows immediately from a counting argument, proving an advantage for dense network coding requires establishing new limits of information flowing through a network. By developing such tools (see Section \ref{sec:methods}), we establish not only that dense network coding uses less communication than classical resources, but that this advantage requires both shared entanglement and quantum communication.

\begin{theorem}[Communication Advantage]\label{thm:main-text-communication-advantage}
Let $\mbf{S} \in \{\mbf{C},\mbf{C}^{E},\mbf{C}^{N},\mbf{Q}\}$. For all $n \in \mbb{N}$,
    \begin{equation}\label{eq:comm-adv}
    \begin{aligned}
        P^{\oplus_{2^{n}}^{2}}_{S}(\mbf{S}(2^n,2^n))) & \leq \frac{1}{2^{n}} <1 = P_{S}^{\oplus_{2^{n}}^{2}}(\mbf{Q}^{E}(2^n,2^n))  \ .
    \end{aligned}
    \end{equation}
    That is, there is a there is a quadratic advantage in signaling dimension or, equivalently, a linear advantage in the number of qubits to communicate versus the number of bits for each sender. Furthermore, if $n$ is even, the first inequality in \eqref{eq:comm-adv} can be replaced with an equality, and the above claims do not change if the senders and receiver possess arbitrary shared randomness.
\end{theorem}

We remark that the inequality in \eqref{eq:comm-adv} being an equality for even $n$ shows our methods are tight and that both entanglement-assistance and quantum communication are necessary to gain a communication advantage over classical MNs under these signaling dimension constraints. We also highlight \eqref{eq:comm-adv} shows there is an exponential advantage in the success probability of implementing the function with quantum resources. 

\subsection{Noise Robustness of Communication Advantage}
Given the 
%quadratic 
advantage in computing $\oplus_{2^{n}}^{2}$ using an entanglement-assisted quantum MN in Theorem \ref{thm:main-text-communication-advantage}, computing $\oplus_{2^{n}}^{2}$ in a distributed manner is a natural benchmark for successfully constructing nodes of a quantum network. Because of the exponential gap in success probability as a function of $n$ in Theorem \ref{thm:main-text-communication-advantage}, this could in principle provide a lot of leeway to outperform a weaker resource configuration with noisy quantum devices so long as the noise does not degrade advantage too quickly. The following theorem shows the advantage is indeed generally robust to noise.
\begin{theorem}[Noise-Robustness, Informal] \label{thm:main-text-noise-robustness} Let $\ve_{st}, \ve_{enc,i}, \ve_{tr,i}, \ve_{dec} \in [0,1]$ denote the deviation of the shared entangled state, encoding map of sender $i$, transmission map of sender $i$, and decoding map of the receiver from the corresponding ideal component respectively in Protocol \ref{prot:main-text-dense-network-coding-example}. Here deviation is measured according to trace or diamond distance and each component is assumed to be independent of the others. Then the probability of guessing the output of the function correctly using the implementation, $P_{G}$, satisfies 
 \begin{align}\label{eq:main-text-scaling}
        P_{G} \geq \max\{1 - (\ve_{st} + \ve_{enc,1} + \ve_{enc,2} + \ve_{tr,1} + \ve_{tr,2} + \ve_{dec}),0\} \ . 
    \end{align}
\end{theorem}
The appeal of this result is its generality--- it makes no claim about the noise model beyond its deviation from the ideal resources under standard norms. Physically, the generality of this result stems from the idea that computing a function over a quantum network should degrade smoothly as the noise grows, and this is captured by the fact the success probability degrades linearly in each source of error. To visualize this result, in Fig.~\ref{fig:main-text-noise-robustness} we plot the bounds that follow from Theorem \ref{thm:main-text-noise-robustness} for a concrete noise model and compare it to the success probabilities for weaker resources according to Theorem \ref{thm:main-text-communication-advantage}. The noise model considered is when all components are ideal except each sender is connected to the receiver with an erasing depolarizing channel $\cD^{p,e}$ of noise parameter $p$ and erasure parameter $e$. This models the setting where a sender's signal is replaced with uniform randomness with probability $p$, outputs heralded loss with probability $e$, and must still guess a value of $z$. This is a special case of a more general noise model we consider in Appendix \ref{app:noise-robustness} where we derive our results on noise robustness. We note that if one has a concrete noise model, they need not use Theorem \ref{thm:main-text-noise-robustness} as they can simply compute the success probabilities.
\begin{figure}
    \centering
    \includegraphics[width=0.5\linewidth]{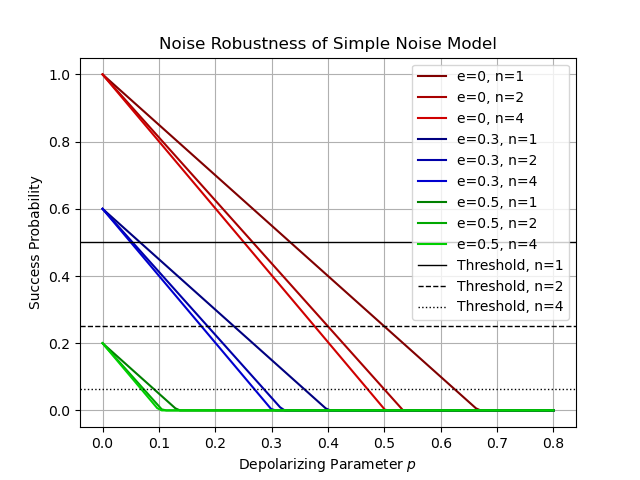}
    \caption{Noise Robustness of Communication Advantage with Heralded Loss and Depolarizing Noise. Bounds on the success probability of the receiver correctly guessing the value of $\oplus^{2}_{2^{n}}$ are plotted.  The colored lines show lower bounds on the success when each sender is allowed to send $n$ qubits over a lossy, depolarizing channel of loss parameter $e$ and depolarizing parameter $p$. These bounds follow from Theorem \ref{thm:main-text-noise-robustness}. The horizontal black lines show upper bounds on the probability of the receiver correctly guessing without both shared entanglement and quantum communication according to Theorem \ref{thm:main-text-communication-advantage}. For a chosen $n$, there is a quantum communication advantage so long as the colored line is above the corresponding black line. \sloppy The lower bound on the noisy case is $P_{S}^{f}(n,e,p) \coloneq \max\{1 - 2(e + p\frac{2^{2n}-1}{2^{2n}}), 0\}$. As $n$ grows, the lower bound approaches $1-2(e+p)$ exponentially fast in $n$, which is why the curves for fixed $e$ become increasingly similar as $n$ grows. For large $n$, a communication advantage can be achieved whenever  $(e + p) < \frac{1}{2}$ because the classical bound scales as $\frac{1}{2^n}\to 0$.}
    \label{fig:main-text-noise-robustness}
\end{figure}

\subsection{Exponential Amplification in the Number of Senders}
So far we have focused on two-sender MNs. We now show how the advantage of dense network coding is amplified over networks with more senders by introducing three additional classes of $2k$-sender MNs (formal details and definitions are found in Appendix \ref{app:amplification-of-advantage}):
\begin{itemize}
    \item \textbf{Pairwise-Entanglement-Assisted (resp.~Pairwise-Nonsignalling-Assisted) Classical MNs:} For every $i \in \{1,...,k\}$, Sender $2i-1$ and $2i$ are connected to the receiver by a noiseless $n_{2i-1}$-bit (resp.~$n_{2i}$-bit) classical channel and share arbitrary entanglement (resp.~an arbitrary fully classical non-signaling box) with each other. No other resources are available. We denote the set of such MNs $\mbf{C}^{P-E}(\vec{n})$ (resp.~$\mbf{C}^{P-N}(\vec{n})$) where the vector $\vec{n} = (2^{n}, \dots, 2^n)$ indexes the signaling dimension of the channel connecting each sender to the receiver.
    \item \textbf{Pairwise-Entanglement-Assisted Quantum MNs:} For every $i \in \{1,...,k\}$, Sender $2i-1$ and $2i$ are connected to the receiver by a noiseless $n_{2i-1}$-qubit (resp.~$n_{2i}$-qubit) quantum channel and share arbitrary entanglement with each other. No other resources are available. We denote the set of such MNs $\mbf{Q}^{P-E}(\vec{n})$.
\end{itemize}

In Appendix \ref{app:amplification-of-advantage}, we show that for the above MNs and appropriate product functions, the optimal strategy is for each function to be implemented between the corresponding senders and the receiver independently. It follows the total success probability of guessing the output is the product of the success probabilities, resulting in the success probability being exponentially suppressed in the number of senders. The following is a concrete example.
\begin{theorem}\label{thm:main-text-amplification-theorem}
    For $n,k \in \mbb{N}$, define the product function $f_{n,k} = \bigtimes_{i \in \{1,...,k\}} f_{i}$ where $f_{i} \coloneq \oplus^{2}_{2^{n}}$ for each $i$. Then,
    \begin{equation}
    \begin{aligned}
        P_{S}^{f_{n,k}}(\mbf{Q}^{P-E}(\vec{n}')) = 1 &>  \left(\frac{1}{2^{n}} \right)^{k} \\ &\geq \max\{P_{S}^{f_{n,k}}(\mbf{Q}(\vec{n}')),P_{S}^{f_{n,k}}(\mbf{C}^{P-N}(\vec{n}'))\} \geq P_{S}^{f_{n,k}}(\mbf{C}^{P-E}(\vec{n}')) \geq P_{S}^{f_{n,k}}(\mbf{C}(\vec{n}')) \ , 
    \end{aligned}
    \end{equation}
    where $\vec{n}' \coloneq (2^{n},...,2^{n})$ is the $2k$-long vector where every entry is $2^{n}$.
\end{theorem}

\subsection{Measurement-Device-Independent Quantum Key Growing}\label{sec:MDI-QKG}
%Brian's motivation for privacy
As seen in Section \ref{sec:main-text-dense-network-coding}, DNC involves Senders 1 and 2 sharing the maximally entangled state. Since the maximally entangled state is pure, any other quantum system would have to be independent and therefore separable from external systems.
As such, if the shared entangled state is ideal, the receiver is trusted, and the quantum communication channels are secure, DNC provides a method for distributed computation that is private from eavesdroppers for functions that can be dense network coded.
Under these assumptions, no information is leaked during the distributed computation because  the initial state, and its subsequent measurement outcome, remain independent of all external quantum systems.
\begin{observation}\label{obs:private-distributed-computation}
    An entanglement-assisted quantum MN with trusted ideal resources and operations can implement a private distributed computation of $\oplus_{2^{n}}^{2}$.
\end{observation}
\noindent For clarity, this observation is formalized in Appendix \ref{app:MDI-QKG}. 

%Limitations of private distributed computation
The limitation of this private distributed computation is the number of independent parties that need to be trusted for it to be securely implemented. First, if the state is not sufficiently entangled, then an eavesdropper could be entangled with the shared state and extract information about the computation. 
If the senders distrust their shared entangled state, they would need to generate a trusted maximally entangled state by utilizing entanglement distillation and purification, which requires classical communication and many entangled pairs.
Second, the quantum channels between the senders and the receiver must be secure against tampering, as otherwise bad actors could simply replace the encoded state. 
Finally, the receiver has to be trusted as otherwise they can simply output the wrong answer.
DNC is also susceptible to attacks that make use of superdense coding.
For example, since the encoding unitaries match those of superdense coding, an eavesdropper with control of the source and measurement device could provide the two senders each with an entanglement-assisted quantum channel, causing the two senders to unknowingly transmit their entire input to the receiver. 

% if the senders distrust their shared entangled state, they would need to be able to exchange arbitrary classical communication and share enough entanglement to verify or distill a trusted maximally entangled state.
% This is because if the state were not sufficiently entangled, then a third party could be entangled with the shared state and extract information about the computation. 
% For example, since the encoding unitaries match those of superdense coding, an eavesdropper with control of the source and measurement device could provide the two senders each with an entanglement-assisted quantum channel, causing the two senders to unknowingly transmit their entire input to the receiver.  
% Second, the quantum channels between the senders and the receiver must be secure against tampering as otherwise bad actors could simply replace the encoded state. Finally, the receiver has to be trusted as otherwise they can simply output the wrong answer. This last point is arguably particularly limiting as trust in the receiver can vary across situations. 

%Why we shouldn't give up and context
While the demands on the number of trusted components make the communication advantage of DNC less relevant for privacy guarantees, we still might expect the new quantum information-theoretic principle to give rise to some cryptographic protocol. This is because new (quantum-information-theoretic) physical principles generally imply new quantum cryptographic protocols. The natural examples of this are the forms of quantum key distribution and their corresponding physical principles, which we provide in Table \ref{tab:QKD-and-Physical-Principle} for motivation. 

\begin{table}[H]
    \centering
    \begin{tabular}{c|c}
        \textbf{Type of QKD} & \textbf{Physical Principle}  \\ \hline
         Prepare-and-Measure \cite{Bennett_2014} & Complementarity \cite{Bohr1928-nf, Bennett_92} \\ \hline
         Entanglement-Based \cite{ekert1991quantum} & Entanglement \cite{schrodinger1935discussion} \\ \hline
         Measurement-Device-Independent \cite{lo2012measurement} & Entangling Measurements \cite{lo2012measurement}
         \\ \hline
         Device-Independent \cite{Vazirani_2014} & Bell Nonlocality \cite{bell1964einstein} \\ \hline

    \end{tabular}
    \caption{Different forms of QKD and their corresponding physical principle. This motivates a QKD-adjacent protocol arising from dense network coding.}
    \label{tab:QKD-and-Physical-Principle}
\end{table}

%Our Result
Building on this history of new types of QKD protocols from different physical principles, we introduce measurement-device-independent (MDI) quantum key growing (QKG) in Protocol \ref{prot:main-text-MDI-Q-Key-Scheme}. Protocol \ref{prot:main-text-MDI-Q-Key-Scheme} presumes an authenticated classical channel between Alice and Bob, but otherwise uses insecure channels. Here we state the asymptotic rate of MDI QKG informally and refer the reader to Appendix \ref{app:MDI-QKG} for further technical details and the one-shot analysis.
\begin{theorem}\label{thm:main-text-MDI-QKG}[Asymptotic Rate, Informal]
    There exists a sequence of MDI QKG schemes as described in Protocol \ref{prot:main-text-MDI-Q-Key-Scheme} such that on input $\rho_{AB}^{\otimes n}$, the number of bits of key extracted per copy of $\rho_{AB}$ as $n$ goes to infinity satisfies
    \begin{align}\label{eq:main-text-MDI-QKG-rate}
        R \geq H(X \vert ABE)_{(\cE \otimes \id_{E})(\psi)} - f_{\text{ec}}H(W\vert Y)_{q} \ , 
    \end{align}
    where $H(\cdot \vert \cdot)$ is the conditional entropy, $\cE$ is Alice and Bob's encoding map, $f_{\text{ec}} \geq 1$ models the asymptotic inefficiency of the error correcting code, $\psi$ is any purification of $\rho_{AB}$, and $q_{WY}$ is the joint distribution induced by the protocol when the receiver is honest (i.e. applies an expected, possibly noisy, measurement $\cM_{AB \to Z}$ each round).
\end{theorem}
%Discussion on result
We remark that MDI QKG as presented in Protocol \ref{prot:main-text-MDI-Q-Key-Scheme} assumes an initial input state and thus is a key \textit{distillation} protocol like that of Devetak and Winter \cite{devetak2005distillation} rather than a QKD protocol. However, similar to how a QKD protocol is a key distillation protocol where the input state is tested in an online manner to determine how much key can be securely distilled, MDI QKG could be implemented with an untrusted input by having Alice and Bob test the input state. In this regard, the protocol is not particularly distinct from QKD. However, we identify this task as `key growing' rather than a new variant of `key distribution' for the following reason. A QKD scheme accounts for the distribution of the quantum state from which to distill key.\footnote{We feel obligated however to note that this is not why QKD is named such.} In contrast, the start of Protocol \ref{prot:main-text-MDI-Q-Key-Scheme} assumes that a (promised) entangled state is shared between Alice and Bob \textit{a priori}. As such, Alice and Bob may already distill a secret key. For example, if they share many copies of the maximally entangled state $\rho^{0}_{A^{n}B^{n}} = \dyad{\Phi^{+}}^{\otimes n}$, then they could already measure each copy in the computational basis to extract a secret key of length $n$. However, on input $\dyad{\Phi^{+}}^{\otimes n}$, Protocol \ref{prot:main-text-MDI-Q-Key-Scheme} converts the shared state to $2n$ bits of key, doubling the key length achievable from distillation (See Fig.~\ref{fig:KeyGrowingVsKeyDistillation}). Since Alice and Bob start with an entangled state from which they can distill key, but instead use Protocol~\ref{prot:main-text-MDI-Q-Key-Scheme} to distill more key, we view this protocol as `growing' the key. It is for this reason we call it `key growing.' \footnote{We acknowledge this may be confusing as sometimes QKD is called quantum key growing because it requires an initial short key to authenticate the classical channel. Here we are naming our scheme key growing because of its relation to key distillation.} 

\begin{protocoldesc}[H]
\caption{MDI Quantum Key Growing Scheme}\label{prot:main-text-MDI-Q-Key-Scheme}
		\textbf{Inputs:} \\
		\hspace{0.5cm}
		\begin{tabular}{ l l}
			$n \in \mathbb{N}$ & Number of rounds  \\
            $\rho^{0}_{A^{n}B^{n}}$ & Initial Quantum State \\
            $r$ & Bits of the error correction (EC) syndrome \\
            $t$ & Bits of EC verification 2-universal hash
		\end{tabular}
		
		\vspace{0.5cm}
		
		\textbf{Protocol:} 
		\begin{enumerate}
            \item \textbf{Growing Phase:}
            \begin{enumerate}[label=\alph*.]
			     \item For $i \in \{1,...,n\}$, Alice draws uniformly random bits $(x_{2i-1},x_{2i}) \leftarrow \{0,1\}^{\times 2}$, applies qubit discrete Weyl operator $W_{x_{2i-1},x_{2i}}$ to the $A_{i}$ system, and sends it to the receiver.
                \item For $i \in \{1,...,n\}$, Bob draws uniformly random bits $(y_{2i-1},y_{2i}) \leftarrow \{0,1\}^{\times 2}$, applies qubit discrete Weyl operator $W_{y_{2i-1},y_{2i}}$ to the $B_{i}$ system, and sends it to the receiver.
                \item Alice receives $\hat{\mbf{z}} \in \{0,1\}^{\times 2n}$ and computes the bitwise XOR $\mbf{w} := \hat{\mbf{z}} \oplus^{2n}_{n} \mbf{x}$.
            \end{enumerate}
            \item \textbf{Post-Processing Phase:}
            \begin{enumerate}[label=\alph*.]
                \item Alice and Bob apply some error correcting code so that \textit{Bob} now holds $\hat{\mbf{w}} \in \{0,1\}^{2n}$.
                \item Alice and Bob apply error verification for $\mbf{w},\hat{\mbf{w}}$ via a hash function and abort if it fails.
                \item Alice and Bob perform privacy amplification so that they hold their respective shares of the key $K_{A}$ and $K_{B}$.
            \end{enumerate}
		\end{enumerate}
	\end{protocoldesc}

While MDI QKG is appealing conceptually, its advantage over key distillation can be limited by the physical reality that it requires further transmission of the state, namely over the insecure quantum channels to the receiver. As the noise over the quantum channels increases, even the honest receiver's announcement becomes noisier, and thus the amount of error correction between Alice and Bob can increase. It follows that for sufficiently high noise in the transmissions to the receiver, the advantage of MDI QKG over standard key distillation disappears. To clarify both the advantage of MDI QKG and its limitation with respect to standard key distillation, in Fig.~\ref{fig:KeyGrowingVsKeyDistillation} we compare the asymptotic rates of MDI QKG and key distillation when the shared state is many copies of a Werner state and the transmission lines to the receiver in MDI QKG are independent depolarizing channels. We refer the reader to Appendix \ref{app:MDI-QKG} for details on how this is calculated.

\begin{figure}[H]
    \centering
    \includegraphics[width=0.7\columnwidth]{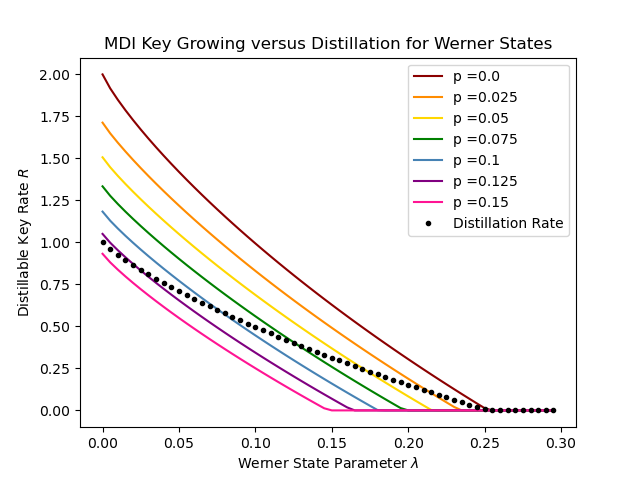}
    \caption{The asymptotic key rate of MDI QKG for Werner states $\rho_{\lambda} = (1-\lambda)\dyad{\Phi^{+}} + \frac{\lambda}{4}I$ using \eqref{eq:main-text-MDI-QKG-rate} when $f_{ec} = 1$ and the state the receiver measures is $(\cD_{p} \otimes \cD_{p})(\rho_{\lambda})$ where $\cD_{p}(X) = (1-p)X + p\Tr[X]\pi$. We compare this to optimally distilling key from the shared Werner states prior to growing using optimal ($f_{ec} = 1$) one-way error correction. As is intuitive, as the noise to the receiver increases, the cost of error correction removes the advantage of key growing over distillation.}
    \label{fig:KeyGrowingVsKeyDistillation}
\end{figure}

\section{Discussion}
In this work we introduce a protocol for quantum dense network coding of information using bipartite entanglement and quantum communication. Dense network coding provides a communication advantage because it enables a class of functions to be computed using one qubit for each two bits of classical communication. %Although parallel superdense coding channels from each sender to the receiver can achieve the same computations as dense network coding, approaches based on superdense coding require twice the number of bipartite entangled pairs to achieve the same performance. 
Although the main text focuses on ditwise modular addition as the key example of dense network coding, we consider generalizations that preserve the quadratic gap in signaling dimension between classical and quantum communication in the appendices. We expect that more functions will be able to be densely computed  by using states beyond Bell states and encoding operations beyond the Discrete-Weyl operators. For example, analogous to Theorem \ref{thm:main-text-communication-advantage}, one may be able to identify simple coding strategies for many-input linear computations over quantum MNs and establish error bounds for MNs with weaker resource configurations. Such results would strengthen recent work on linear computations over networks with entanglement-assistance (see \cite{Allaix-2023a,Yao-2024a,Yao-2025a,Meng-2026a}), making these computational tasks more practical for near-term quantum networks. Indeed, our introduced dense network coding protocol offers a basic primitive from which a broader concept of dense multiparty information processing emerges. In this paradigm, entanglement-assisted quantum communication leads to distributed information processing systems that offer greater efficiency and more privacy than any system lacking these resources. 

Beyond the development of a richer theory of functions that can be dense network coded, an important next step will be experimental demonstration of dense network coding. It should be noted that in practice as the size of the inputs to the function and needed quantum communication resources increase, there will likely be increased noise in the physical system \cite{Polozova2016_highdim_bell_inequalities}. Even with the noise-robustness of dense network coding, the increased noise from using higher-signaling dimensions or number of parties may prevent the advantage from being realized. Nonetheless, there has been very recent progress in the experimental demonstration of the quantum communication advantage of quantum fingerprinting \cite{yacoub2026experimentaldemonstrationquantumadvantage}, and thus we believe dense network coding may serve as a natural next benchmark for implementing quantum networks.

\section{Methods}\label{sec:methods}
The dense network coding protocol was presented in Section \ref{sec:main-text-dense-network-coding}. This section elucidates key key details of our techniques. At a high-level, we convert the problem of simulating the function to characterizing the min-entropy that the output of the function must have when conditioned on the output of a family of MNs. We then use physical principles of the network to control this min-entropy. We sketch this below and refer the reader to the appendices for further details.

\subsection{From Channel Simulation to Guessing Probability} The first step is to reduce our measure of channel simulation given in \eqref{eq:main-text-channel-sim-measure} to an optimization of the guessing probability over all joint states of the output function value and the senders' encoding. This reduction is performed in Appendix \ref{app:chan-sim-to-guessing-prob}. Here we state an informal version of the two-sender MN case here for clarity.
\begin{lemma}[Reduction to Guessing Probability, Informal \& Simplified]\label{lem:main-text-reduc-to-guessing-prob}
    Let $f: \cX \times \cY \to \cZ$ be a function. Consider a set $\mbf{S}$ of two-sender MNs that are always of the form of an `encoding-and-transmission channel' $\cE^{\net}_{XY \to C}$ where the $C$ register is the total message the receiver receives. Then,
    \begin{align}
        P_{S}^{f}(\mbf{S}) = \sup_{\cE^{\net}} p_{g}(Z \vert C)_{(\id_{Z} \otimes \cE^{\net}_{XY \to C})(\rho_{ZXY})} \ , 
    \end{align}
    where the maximization is over encoding-and-transmission channels possible according to the set $\mbf{S}$, $\rho_{ZXY} = \frac{1}{\vert X \vert \vert Y \vert} \sum_{x,y} \dyad{f(x,y)} \otimes \dyad{x}_{X} \otimes \dyad{y}_{Y}$, and the guessing probability is:
    \begin{align}
        p_{g}(Z \vert C)_{\rho} = \max_{\{\Gamma_{z}\}_{z} \text{is POVM}} \sum_{z} p_{Z}(z)\Tr[\Gamma_{z}\rho^{z}_{C}] \ ,
    \end{align}
    where $\rho_{ZC} = \sum_{z} p_{Z}(z) \dyad{z} \otimes \rho^{z}_{C}$ for some distribution over $\cZ$, $p_{Z}$, and set of quantum states $\{\rho^{z}_{C}\}_{z \in \cZ}$.
\end{lemma}
\noindent In \cite{konig2009operational}, it was shown that the guessing probability is proportional to the conditional min-entropy of the joint quantum state, i.e.~
\begin{align}\label{eq:main-text-guessing-prob-to-min-ent}
    p_{g}(X \vert C)_{\rho} = \exp(-H_{\min}(X \vert C)_{\rho}) \ ,
\end{align}
and thus Lemma \ref{lem:main-text-reduc-to-guessing-prob} is enough to need to control the conditional min-entropy given the structure of the MNs. Moreover, the guessing probability is multiplicative over tensor products, which formalizes independence of different systems. As formalized in Appendix \ref{app:amplification-of-advantage}, it follows that if one selects for networks that preserve independence in the encoding of different systems, then using Lemma \ref{lem:main-text-reduc-to-guessing-prob}, one can  obtain multiplicative bounds on computing functions in parallel over the network. Combining this with the subsequent bounds allows us to establish Theorem \ref{thm:main-text-amplification-theorem}.

\subsection{Controlling Guessing Probability of Output by the Inputs} As shown in Appendix \ref{app:preliminaries}, a straightforward implication of \cite{george2022finite} is that given a multivariate function $f:\cX \times \cY \to \cZ$, it is generally easier to guess the output of the function than its inputs even with quantum side-information, i.e. for the joint classical-quantum state 
\begin{align}
    \rho_{ZXYQ} = \sum_{x,y} p_{XY}(x,y)\dyad{f(x,y)}_{Z} \otimes \dyad{x}_{X} \otimes \dyad{y}_{Y} \otimes \rho_{Q}^{x,y} \ ,
\end{align}
it is the case $p_{g}(Z \vert YQ)_{\rho} \geq p_{g}(X \vert YQ)_{\rho}$. This makes controlling the guessing probability of the output of a function in terms of the randomness in the inputs difficult. To resolve this difficulty, we introduce conditionally bijective functions.
\begin{definition}\label{def:main-text-DCB}
    Let $\cX,\cY,\cZ$ be finite alphabets. A function $f: \cX \times \cY \to \cZ$ is $Y$-conditionally bijective if for each $y \in \cY$, the function $f_{y}:\cX \to \cZ$ defined by $f_{y}(x) = f(x,y)$ is a bijection. The function is doubly conditionally bijective if it is both $\cX$- and $Y$-conditionally bijective.
\end{definition}
\noindent Intuitively, $Y$-conditionally bijective functions are appealing because, conditioned on the value of $y \in \cY$, guessing the value of $z \in \cZ$ is the same as guessing the value of $x \in \cX$. This is because one can simply apply inverse of $f_{y}$ to the guessed value of $z$ to obtain a guessed value of $x$. Formalizing this idea, in Appendix \ref{app:entropy-chain-rule}, we establish that for a $Y$-conditionally bijective function, guessing the output is as hard as guessing the input:
\begin{align}
    p_{g}(Z \vert YQ)_{\rho} = p_{g}(X \vert YQ)_{\rho} \ . 
\end{align}
This equality allows us to control the guessing probability using the randomness in the inputs. Moreover, as we show this equality in fact holds for a large family of entropic quantities and includes quantum side-information, we are able to use it to establish the security of our MDI QKG protocol in Appendix \ref{app:MDI-QKG}.

\subsection{Limits of Information Flow through Networks}
The ability to guess a function of inputs at the output of a network should be limited by the ability for the relevant information to propagate through the network to the output. In Appendix \ref{app:bounds-on-guessing-cond-bij-over-MNs}, we formalize this physical intuition for $Y$-conditionally bijective functions. In particular, we limit two cases. The first case is for two-sender quantum MNs. A quantum MN is point-to-point quantum channels run in parallel. By considering a $Y$-conditionally bijective function and giving the receiver access to the variable $Y$, we limit the receiver's ability to guess the output of the function by how well the receiver can learn the variable $X$, i.e.~we reduce the network problem to point-to-point information transmission over a quantum channel. We then use the celebrated Frenkel-Weiner theorem \cite{Frenkel2015_classical_information_n-level_quantum_system} to limit the guessing probability by the size of the input and the signaling dimension. The second case is when two senders are connected to the receiver by a fully classical bipartite channel $\cN_{XY \to X'Y'}$ that does not signal from $X$ to $Y$. A fundamental fact is that such a channel may be decomposed into the composition of two channels such that information only flows from $Y$ to $X$ \cite{eggeling2002semicausal}. Using this decomposition in conjunction with the channel being fully classical, we show that for a $Y$-conditionally bijective function guessing the output of the function is limited by the randomness of the input $X$ and the dimension of the output register $X'$. The limitations of these two networks together give rise to the following lemma.
\begin{lemma}\label{lem:main-text-limits-of-guessing}
    Let $f:\cX \times \cY \to \cZ$ be a doubly conditionally bijective function. Then
    \begin{align}
        \max\{P_{S}^{f}(\mbf{Q}(2^n,2^m)), P_{S}^{f}(\mbf{C}^{N}(2^n,2^m))\} \leq \frac{\min\{2^{n},2^{m}\}}{\vert \cZ \vert} \ . 
    \end{align}
    Moreover, these bounds are tight for $f = \oplus_{2^{k}}$ when $2^{n}$ and $2^{m}$ divide $2^{k}$.
\end{lemma}
\noindent By using that the multiplication operation of any group is a doubly conditionally bijective function, Lemma \ref{lem:main-text-limits-of-guessing} allows us to place limits on computing the functions considered in dense network coding for other network structures, thus allowing us to conclude Theorem \ref{thm:main-text-communication-advantage}.

\paragraph*{Data Availability} No data sets were generated during this study.

\paragraph*{Code Availability} MATLAB and Python code used to generate Figures \ref{fig:main-text-noise-robustness} and \ref{fig:KeyGrowingVsKeyDistillation} are publicly accessible at \href{https://github.com/qit-george/QuantumDenseNetworkCoding}{this GitHub repository.}

\bibliography{references.bib}

@article{doolittle2024operational_nonclassicality,
  title={Operational Nonclassicality in Quantum Communication Networks},
  author={Doolittle, Brian and Leditzky, Felix and Chitambar, Eric},
  journal={arXiv preprint arXiv:2403.02988},
  url= {https://arxiv.org/abs/2403.02988},
  year={2024}
}

@Article{bennet1992_dense_coding,
  author        = {Bennett, Charles H. and Wiesner, Stephen J.},
  title         = {Communication via one- and two-particle operators on Einstein-Podolsky-Rosen states},
  journal       = {Phys. Rev. Lett.},
  year          = {1992},
  volume        = {69},
  pages         = {2881--2884},
  month         = {Nov},
  archiveprefix = {arXiv},
  doi           = {10.1103/PhysRevLett.69.2881},
  issue         = {20},
  numpages      = {0},
  publisher     = {American Physical Society},
  url           = {https://link.aps.org/doi/10.1103/PhysRevLett.69.2881},
}

@article{Holevo1973BoundsFT,
  title={Bounds for the quantity of information transmitted by a quantum communication channel},
  author={Alexander S. Holevo},
  year={1973},
  journal = {Problemy Peredachi Informatsii},
  volume = {9},
  url={https://api.semanticscholar.org/CorpusID:118312737}
}

@Article{Frenkel2015_classical_information_n-level_quantum_system,
  author        = {P{\'{e}}ter E. Frenkel and Mih{\'{a}}ly Weiner},
  title         = {Classical Information Storage in an n-Level Quantum System},
  journal       = {Communications in Mathematical Physics},
  year          = {2015},
  volume        = {340},
  number        = {2},
  pages         = {563--574},
  month         = sep,
  archiveprefix = {arXiv},
  eprint		= {1304.5723},
  primaryclass	= {quant-ph},
  doi           = {10.1007/s00220-015-2463-0},
  publisher     = {Springer Science and Business Media {LLC}},
  url           = {https://doi.org/10.1007/s00220-015-2463-0},
}

@article{dallarno2017_no_hypersignaling,
  title = {No-Hypersignaling Principle},
  author = {Dall'Arno, Michele and Brandsen, Sarah and Tosini, Alessandro and Buscemi, Francesco and Vedral, Vlatko},
  journal = {Phys. Rev. Lett.},
  volume = {119},
  issue = {2},
  pages = {020401},
  numpages = {7},
  year = {2017},
  month = {Jul},
  publisher = {American Physical Society},
  doi = {10.1103/PhysRevLett.119.020401},
  url = {https://link.aps.org/doi/10.1103/PhysRevLett.119.020401},
  archiveprefix = {arXiv},
  primaryclass = {quant-ph},
  eprint = {1609.09237}
}

@article{doolittle_2021_certifying_classical_simulation_cost,
  title = {Certifying the classical simulation cost of a quantum channel},
  author = {Doolittle, Brian and Chitambar, Eric},
  journal = {Phys. Rev. Res.},
  volume = {3},
  issue = {4},
  pages = {043073},
  numpages = {25},
  year = {2021},
  month = {Oct},
  publisher = {American Physical Society},
  doi = {10.1103/PhysRevResearch.3.043073},
  url = {https://link.aps.org/doi/10.1103/PhysRevResearch.3.043073}
}

@ARTICLE{chitambar2023_communication_value,
  author={Chitambar, Eric and George, Ian and Doolittle, Brian and Junge, Marius},
  journal={IEEE Transactions on Information Theory}, 
  title={The Communication Value of a Quantum Channel}, 
  year={2023},
  volume={69},
  number={3},
  pages={1660-1679},
  doi={10.1109/TIT.2022.3218540},
  eprint = {2109.11144},
  archiveprefix = {arXiv},
  primaryclass = {quant-ph},
}

@article{Polozova2016_highdim_bell_inequalities,
  title = {Higher-dimensional Bell inequalities with noisy qudits},
  volume = {93},
  ISSN = {2469-9934},
  url = {http://dx.doi.org/10.1103/PhysRevA.93.032130},
  DOI = {10.1103/physreva.93.032130},
  number = {3},
  journal = {Physical Review A},
  publisher = {American Physical Society (APS)},
  author = {Polozova,  Elena and Strauch,  Frederick W.},
  year = {2016},
  month = mar 
}

@article{popescu1994quantum_pr_box,
  title={Quantum nonlocality as an axiom},
  author={Popescu, Sandu and Rohrlich, Daniel},
  journal={Foundations of Physics},
  volume={24},
  number={3},
  pages={379--385},
  year={1994},
  publisher={Springer}
}

@article{Leditzky2020_mac_games_ea_cmac,
  title = {Playing games with multiple access channels},
  volume = {11},
  ISSN = {2041-1723},
  url = {http://dx.doi.org/10.1038/s41467-020-15240-w},
  DOI = {10.1038/s41467-020-15240-w},
  number = {1},
  journal = {Nature Communications},
  publisher = {Springer Science and Business Media LLC},
  author = {Leditzky,  Felix and Alhejji,  Mohammad A. and Levin,  Joshua and Smith,  Graeme},
  year = {2020},
  month = mar 
}

@article{Zhang2022_single_particle_mac,
  title = {Building Multiple Access Channels with a Single Particle},
  volume = {6},
  ISSN = {2521-327X},
  url = {http://dx.doi.org/10.22331/q-2022-02-16-653},
  DOI = {10.22331/q-2022-02-16-653},
  journal = {Quantum},
  publisher = {Verein zur Forderung des Open Access Publizierens in den Quantenwissenschaften},
  author = {Zhang,  Yujie and Chen,  Xinan and Chitambar,  Eric},
  year = {2022},
  month = feb,
  pages = {653},
  eprint = {2006.12475},
  primaryClass = {quant-ph},
  archivePrefix = {arXiv},
}

@inbook{Cleve1999_ea_inner_product,
  title = {Quantum Entanglement and the Communication Complexity of the Inner Product Function},
  ISBN = {9783540492085},
  ISSN = {0302-9743},
  url = {http://dx.doi.org/10.1007/3-540-49208-9_4},
  DOI = {10.1007/3-540-49208-9_4},
  booktitle = {Quantum Computing and Quantum Communications},
  publisher = {Springer Berlin Heidelberg},
  author = {Cleve,  Richard and van Dam,  Wim and Nielsen,  Michael and Tapp,  Alain},
  year = {1999},
  pages = {61–74}
}

@Article{Bowles2015_nonclassicality_communication_networks,
  author        = {Bowles, Joseph and Brunner, Nicolas and Paw\l{}owski, Marcin},
  title         = {Testing dimension and nonclassicality in communication networks},
  journal       = {Phys. Rev. A},
  year          = {2015},
  volume        = {92},
  pages         = {022351},
  month         = {Aug},
  archiveprefix = {arXiv},
  eprint		= {1505.01736},
  primaryclass	= {quant-ph},
  doi           = {10.1103/PhysRevA.92.022351},
  issue         = {2},
  numpages      = {10},
  publisher     = {American Physical Society},
  url           = {https://link.aps.org/doi/10.1103/PhysRevA.92.022351},
}

@book{Wilde-Book,
  title={Quantum information theory},
  author={Wilde, Mark M},
  year={2013},
  DOI = {10.1017/CBO9781139525343},
  publisher={Cambridge University Press}
}

@book{Tomamichel-Book,
  title={Quantum information processing with finite resources: mathematical foundations},
  author={Tomamichel, Marco},
  volume={5},
  year={2015},
  publisher={Springer},
  doi = {10.1007/978-3-319-21891-5},
  note = {All references are to version 5 on the arXiv: https://arxiv.org/abs/1504.00233v5}
}

@phdthesis{tomamichel-thesis,
      title={A Framework for Non-Asymptotic Quantum Information Theory}, 
      author={Marco Tomamichel},
      year={2013},
      eprint={1203.2142},
      archivePrefix={arXiv},
      primaryClass={quant-ph},
      school = {ETH Zurich},
      url={https://arxiv.org/abs/1203.2142}, 
}

@article{Tomamichel-2017a,
  title={A largely self-contained and complete security proof for quantum key distribution},
  author={Tomamichel, Marco and Leverrier, Anthony},
  journal={Quantum},
  volume={1},
  pages={14},
  year={2017},
  publisher={Verein zur F{\"o}rderung des Open Access Publizierens in den Quantenwissenschaften}
}

@article{slepian1973noiseless,
  title={Noiseless coding of correlated information sources},
  author={Slepian, David and Wolf, Jack},
  journal={IEEE Transactions on information Theory},
  volume={19},
  number={4},
  pages={471--480},
  year={1973},
  publisher={IEEE}
}

@article{george2022finite,
  title={Finite-Key Analysis of Quantum Key Distribution with Characterized Devices Using Entropy Accumulation},
   volume={9},
   ISSN={2521-327X},
   url={http://dx.doi.org/10.22331/q-2025-12-12-1941},
   DOI={10.22331/q-2025-12-12-1941},
   journal={Quantum},
   publisher={Verein zur Forderung des Open Access Publizierens in den Quantenwissenschaften},
   author={George, Ian and Lin, Jie and van Himbeeck, Thomas and Fang, Kun and Lütkenhaus, Norbert},
   year={2025},
   month=Dec, pages={1941} }

@BOOK{Horst1994,
  title     = "Handbook of global optimization",
  editor    = "Horst, Reiner and Pardalos, Panos M",
  publisher = "Springer",
  series    = "Nonconvex Optimization and Its Applications",
  edition   =  1995,
  month     =  nov,
  year      =  1994,
  address   = "Dordrecht, Netherlands",
  language  = "en"
}

@article{Marwah_2022,
   title={Uniform continuity bound for sandwiched Rényi conditional entropy},
   volume={63},
   ISSN={1089-7658},
   url={http://dx.doi.org/10.1063/5.0088507},
   DOI={10.1063/5.0088507},
   number={5},
   journal={Journal of Mathematical Physics},
   publisher={AIP Publishing},
   author={Marwah, Ashutosh and Dupuis, Frédéric},
   year={2022},
   month=May }

@article{Bluhm_2024,
   title={Unified Framework for Continuity of Sandwiched Rényi Divergences},
   volume={27},
   ISSN={1424-0661},
   url={http://dx.doi.org/10.1007/s00023-024-01519-x},
   DOI={10.1007/s00023-024-01519-x},
   number={1},
   journal={Annales Henri Poincaré},
   publisher={Springer Science and Business Media LLC},
   author={Bluhm, Andreas and Capel, Ángela and Gondolf, Paul and Möbus, Tim},
   year={2024},
   month=Dec, pages={1–50} }

@misc{khatri-book,
      title={Principles of Quantum Communication Theory: A Modern Approach}, 
      author={Sumeet Khatri and Mark M. Wilde},
      year={2024},
      eprint={2011.04672},
      archivePrefix={arXiv},
      primaryClass={quant-ph},
      url={https://arxiv.org/abs/2011.04672}, 
}

@inproceedings{klappenecker2003unitary,
  title={Unitary error bases: Constructions, equivalence, and applications},
  author={Klappenecker, Andreas and R{\"o}tteler, Martin},
  booktitle={International Symposium on Applied Algebra, Algebraic Algorithms, and Error-Correcting Codes},
  pages={139--149},
  year={2003},
  organization={Springer}
}

@book{dinitz2007handbook,
  title={Handbook of combinatorial designs},
  author={Dinitz, Jeffrey H},
  year={2007},
  publisher={Chapman \& Hall/CRC}
}

@article{werner2001-dense-coding,
  title={All teleportation and dense coding schemes},
  author={Werner, Reinhard F},
  journal={Journal of Physics A: Mathematical and General},
  volume={34},
  number={35},
  pages={7081--7094},
  year={2001}
}

@article{knill1996group,
  title={Group representations, error bases and quantum codes},
  author={Knill, Emanuel},
  journal={arXiv preprint quant-ph/9608049},
  year={1996}
}

@article{klappenecker2002beyond,
  title={Beyond stabilizer codes. I. Nice error bases},
  author={Klappenecker, Andreas A and Rotteler, M},
  journal={IEEE Transactions on Information Theory},
  volume={48},
  number={8},
  pages={2392--2395},
  year={2002},
  publisher={IEEE}
}

@article{eggeling2002semicausal,
  title={Semicausal operations are semilocalizable},
  author={Eggeling, Tilo and Schlingemann, Dirk and Werner, Reinhard F},
  journal={EPL (Europhysics Letters)},
  volume={57},
  number={6},
  pages={782--788},
  year={2002}
}

@article{on-quantum-ns-boxes,
  title = {Properties of quantum nonsignaling boxes},
  author = {Piani, M. and Horodecki, M. and Horodecki, P. and Horodecki, R.},
  journal = {Phys. Rev. A},
  volume = {74},
  issue = {1},
  pages = {012305},
  numpages = {13},
  year = {2006},
  month = {Jul},
  publisher = {American Physical Society},
  doi = {10.1103/PhysRevA.74.012305},
  url = {https://link.aps.org/doi/10.1103/PhysRevA.74.012305}
}

@article{chiribella2009theoretical,
  title={Theoretical framework for quantum networks},
  author={Chiribella, Giulio and D’Ariano, Giacomo Mauro and Perinotti, Paolo},
  journal={Physical Review A—Atomic, Molecular, and Optical Physics},
  volume={80},
  number={2},
  pages={022339},
  year={2009},
  publisher={APS}
}

@article{konig2009operational,
  title={The operational meaning of min-and max-entropy},
  author={Konig, Robert and Renner, Renato and Schaffner, Christian},
  journal={IEEE Transactions on Information theory},
  volume={55},
  number={9},
  pages={4337--4347},
  year={2009},
  publisher={IEEE}
}

@article{Portmann_2022,
   title={Security in quantum cryptography},
   volume={94},
   ISSN={1539-0756},
   url={http://dx.doi.org/10.1103/RevModPhys.94.025008},
   DOI={10.1103/revmodphys.94.025008},
   number={2},
   journal={Reviews of Modern Physics},
   publisher={American Physical Society (APS)},
   author={Portmann, Christopher and Renner, Renato},
   year={2022},
   month={June} }

@book{WatrousBook,
  title={The Theory of Quantum Information},
  author={Watrous, John},
  year={2018},
  publisher={Cambridge University Press},
  eprint = {https://cs.uwaterloo.ca/~watrous/TQI/},
  doi = {10.1017/9781316848142}
}

@article{Bennett_2014,
   title={Quantum cryptography: Public key distribution and coin tossing},
   volume={560},
   ISSN={0304-3975},
   url={http://dx.doi.org/10.1016/j.tcs.2014.05.025},
   DOI={10.1016/j.tcs.2014.05.025},
   journal={Theoretical Computer Science},
   publisher={Elsevier BV},
   author={Bennett, Charles H. and Brassard, Gilles},
   year={2014},
   month=Dec, pages={7–11} }

@article{Tomamichel-2009a,
  title={A fully quantum asymptotic equipartition property},
  author={Tomamichel, Marco and Colbeck, Roger and Renner, Renato},
  journal={IEEE Transactions on information theory},
  volume={55},
  number={12},
  pages={5840--5847},
  year={2009},
  publisher={IEEE}
}

@article{devetak2005distillation,
  title={Distillation of secret key and entanglement from quantum states},
  author={Devetak, Igor and Winter, Andreas},
  journal={Proceedings of the Royal Society A: Mathematical, Physical and engineering sciences},
  volume={461},
  number={2053},
  pages={207--235},
  year={2005},
  publisher={The Royal Society}
}

@article{ekert1991quantum,
  title={Quantum cryptography based on Bell’s theorem},
  author={Ekert, Artur K},
  journal={Physical review letters},
  volume={67},
  number={6},
  pages={661},
  year={1991},
  publisher={APS}
}

@article{lo2012measurement,
  title={Measurement-device-independent quantum key distribution},
  author={Lo, Hoi-Kwong and Curty, Marcos and Qi, Bing},
  journal={Physical review letters},
  volume={108},
  number={13},
  pages={130503},
  year={2012},
  publisher={APS}
}

@article{Vazirani_2014,
   title={Fully Device-Independent Quantum Key Distribution},
   volume={113},
   ISSN={1079-7114},
   url={http://dx.doi.org/10.1103/PhysRevLett.113.140501},
   DOI={10.1103/physrevlett.113.140501},
   number={14},
   journal={Physical Review Letters},
   publisher={American Physical Society (APS)},
   author={Vazirani, Umesh and Vidick, Thomas},
   year={2014},
   month={Sept} }

@ARTICLE{Bohr1928-nf,
  title     = "The quantum postulate and the recent development of atomic
               Theory1",
  author    = "Bohr, Niels",
  journal   = "Nature",
  publisher = "Springer Science and Business Media LLC",
  volume    =  121,
  number    =  3050,
  pages     = "580--590",
  month     =  apr,
  year      =  1928,
  language  = "en"
}

@article{Bennett_92,
  title = {Quantum cryptography using any two nonorthogonal states},
  author = {Bennett, Charles H.},
  journal = {Phys. Rev. Lett.},
  volume = {68},
  issue = {21},
  pages = {3121--3124},
  numpages = {0},
  year = {1992},
  month = {May},
  publisher = {American Physical Society},
  doi = {10.1103/PhysRevLett.68.3121},
  url = {https://link.aps.org/doi/10.1103/PhysRevLett.68.3121}
}

@inproceedings{schrodinger1935discussion,
  title={Discussion of probability relations between separated systems},
  author={Schr{\"o}dinger, Erwin},
  booktitle={Mathematical proceedings of the cambridge philosophical society},
  volume={31},
  pages={555--563},
  year={1935},
  organization={Cambridge University Press}
}

@article{bell1964einstein,
  title={On the einstein podolsky rosen paradox},
  author={Bell, John S},
  journal={Physics Physique Fizika},
  volume={1},
  number={3},
  pages={195},
  year={1964},
  publisher={APS}
}

@article{Bae_2015,
   title={Quantum state discrimination and its applications},
   volume={48},
   ISSN={1751-8121},
   url={http://dx.doi.org/10.1088/1751-8113/48/8/083001},
   DOI={10.1088/1751-8113/48/8/083001},
   number={8},
   journal={Journal of Physics A: Mathematical and Theoretical},
   publisher={IOP Publishing},
   author={Bae, Joonwoo and Kwek, Leong-Chuan},
   year={2015},
   month=Jan, pages={083001} }

@article{Kitaev_1997,
doi = {10.1070/RM1997v052n06ABEH002155},
url = {https://doi.org/10.1070/RM1997v052n06ABEH002155},
year = {1997},
month = {dec},
publisher = {},
volume = {52},
number = {6},
pages = {1191},
author = {A Yu Kitaev},
title = {Quantum computations: algorithms and error correction},
journal = {Russian Mathematical Surveys}
}

@ARTICLE{Devetak-2008a,
  author={Devetak, Igor and Harrow, Aram W. and Winter, Andreas J.},
  journal={IEEE Transactions on Information Theory}, 
  title={A Resource Framework for Quantum Shannon Theory}, 
  year={2008},
  volume={54},
  number={10},
  pages={4587-4618},
  doi={10.1109/TIT.2008.928980}}

@misc{tan2024prospectsdeviceindependentquantumkey,
      title={Prospects for device-independent quantum key distribution}, 
      author={Ernest Y. -Z. Tan},
      year={2024},
      eprint={2111.11769},
      archivePrefix={arXiv},
      primaryClass={quant-ph},
      url={https://arxiv.org/abs/2111.11769}, 
}

@ARTICLE{EA-capacity,
  author={Bennett, C.H. and Shor, P.W. and Smolin, J.A. and Thapliyal, A.V.},
  journal={IEEE Transactions on Information Theory}, 
  title={Entanglement-assisted capacity of a quantum channel and the reverse Shannon theorem}, 
  year={2002},
  volume={48},
  number={10},
  pages={2637-2655},
  doi={10.1109/TIT.2002.802612}}

@ARTICLE{Cubitt2011_NS-capacity,
  author={Cubitt, Toby S. and Leung, Debbie and Matthews, William and Winter, Andreas},
  journal={IEEE Transactions on Information Theory}, 
  title={Zero-Error Channel Capacity and Simulation Assisted by Non-Local Correlations}, 
  year={2011},
  volume={57},
  number={8},
  pages={5509-5523},
  doi={10.1109/TIT.2011.2159047}}

@inproceedings{yao1979_communication_complexity,
    author = {Yao, Andrew Chi-Chih},
    title = {Some complexity questions related to distributive computing(Preliminary Report)},
    year = {1979},
    isbn = {9781450374385},
    publisher = {Association for Computing Machinery},
    address = {New York, NY, USA},
    url = {https://doi.org/10.1145/800135.804414},
    doi = {10.1145/800135.804414},
    booktitle = {Proceedings of the Eleventh Annual ACM Symposium on Theory of Computing},
    pages = {209–213},
    numpages = {5},
    location = {Atlanta, Georgia, USA},
    series = {STOC '79}
}

@article{buhrman2010_nonlocality_communication_complexity,
  title = {Nonlocality and communication complexity},
  author = {Buhrman, Harry and Cleve, Richard and Massar, Serge and de Wolf, Ronald},
  journal = {Rev. Mod. Phys.},
  volume = {82},
  issue = {1},
  pages = {665--698},
  numpages = {0},
  year = {2010},
  month = {Mar},
  publisher = {American Physical Society},
  doi = {10.1103/RevModPhys.82.665},
  url = {https://link.aps.org/doi/10.1103/RevModPhys.82.665}
}

@article{buhrman2001quantum,
  title={Quantum fingerprinting},
  author={Buhrman, Harry and Cleve, Richard and Watrous, John and De Wolf, Ronald},
  journal={Physical review letters},
  volume={87},
  number={16},
  pages={167902},
  year={2001},
  publisher={APS}
}

@misc{yacoub2026experimentaldemonstrationquantumadvantage,
      title={Experimental demonstration of quantum advantage in communication complexity for Euclidean distance problem}, 
      author={Verena Yacoub and Niraj Kumar and Iordanis Kerenidis and Eleni Diamanti},
      year={2026},
      eprint={2605.31516},
      archivePrefix={arXiv},
      primaryClass={quant-ph},
      url={https://arxiv.org/abs/2605.31516}, 
}

@ARTICLE{ahlswede_2000_network_coding,
  author={Ahlswede, R. and Ning Cai and Li, S.-Y.R. and Yeung, R.W.},
  journal={IEEE Transactions on Information Theory}, 
  title={Network information flow}, 
  year={2000},
  volume={46},
  number={4},
  pages={1204-1216},
  doi={10.1109/18.850663}}

@book{ho2008_network_coding,
  title={Network coding: an introduction},
  author={Ho, Tracey and Lun, Desmond},
  year={2008},
  publisher={Cambridge University Press}
}

@ARTICLE{Song-QPIR-2021a,
  author={Song, Seunghoan and Hayashi, Masahito},
  journal={IEEE Transactions on Information Theory}, 
  title={Capacity of Quantum Private Information Retrieval With Multiple Servers}, 
  year={2021},
  volume={67},
  number={1},
  pages={452-463},
  doi={10.1109/TIT.2020.3022515}}

@INPROCEEDINGS{Allaix-2023a,
  author={Allaix, Matteo and Lu, Yuxiang and Yao, Yuhang and Pllaha, Tefjol and Hollanti, Camilla and Jafar, Syed},
  booktitle={GLOBECOM 2023 - 2023 IEEE Global Communications Conference}, 
  title={N-Sum Box: An Abstraction for Linear Computation over Many-to-one Quantum Networks}, 
  year={2023},
  volume={},
  number={},
  pages={5457-5462},
  doi={10.1109/GLOBECOM54140.2023.10437170}}

@ARTICLE{Yao-2024a,
  author={Yao, Yuhang and Jafar, Syed A.},
  journal={IEEE Transactions on Information Theory}, 
  title={The Capacity of Classical Summation Over a Quantum MAC With Arbitrarily Distributed Inputs and Entanglements}, 
  year={2024},
  volume={70},
  number={9},
  pages={6350-6370},
  doi={10.1109/TIT.2024.3397917}}

@ARTICLE{Yao-2025a,
  author={Yao, Yuhang and Jafar, Syed A.},
  journal={IEEE Transactions on Quantum Engineering}, 
  title={On the Capacity of Vector Linear Computation Over a Noiseless Quantum Multiple-Access Channel With Entangled Transmitters}, 
  year={2025},
  volume={6},
  number={},
  pages={1-20},
  doi={10.1109/TQE.2025.3620628}}

@inproceedings{Meng-2026a,
    author = {Meng, Ruoyu and Ramamoorthy, Aditya},
    title = {Precoding based protocols for entanglement assisted liner computation over a quantum MAC},
    booktitle = {ISIT 2026},
    year = {2026}, 
    note ={To appear}
}

\paragraph*{Acknowledgments} The authors thank Haneul Kim, Felix Leditzky, and Eric Chitambar for insightful discussions. This project is supported by Aliro Technologies, Inc. This project is supported by the Ministry of Education, Singapore, through grant T2EP20124-0005. This project is supported by the National Research Foundation, Singapore under the NRF Postdoctoral award.

\paragraph*{Author Contributions} IG and BD both contributed extensively to the paper.

\paragraph*{Author Information} The authors declare no competing interests. Correspondence may be addressed to qit.george@gmail.com or bdoolittle@aliroquantum.com.

\appendix

\section{Preliminaries}\label{app:preliminaries}
\paragraph{Basic Notation and Terminology} For basic notation we follow standard (quantum) information theory texts, e.g.~\cite{Wilde-Book,Tomamichel-Book} and refer to such for details beyond what we provide here. To index over the numbers, we will define the notation $[n] = \{0,1,...,n-1\}$.  We use capital Roman calligraphic letters to represent finite alphabets $\cW,\cX,\cY,\cZ$. We will consider multi-variable functions $f:\cW^{k} \to \cZ$ where $\cW^{k} \coloneq \bigtimes_{i \in [k]} \cW_{i}$. We recall the pre-image of a function $f:\cW^{k} \to \cZ$, $f^{-1}(z) \coloneq \{w^k \in \cW^{k} : f(w^{k}) = z\}$. When there are only two input variables and it improves presentation, we use $\cX$ and $\cY$ rather than $\cW_{1}$ and $\cW_{2}$. As an example of notation for probability distributions, $P_{X Y \vert Z}$ denotes a conditional distribution over $\cX \times \cY$ conditioned on $\cZ$ and $P_{XY \vert Z}(x,y \vert z)$ denotes the joint probability $(X = x, Y=y)$ conditioned on $Z=z$ according to said conditional distribution. Note that classical channels are identified with conditional probability distributions.

\paragraph{Quantum States} We denote finite-dimensional Hilbert spaces with capital letters e.g.~$A,B$, and $E$. The dimension of a Hilbert space is expressed as $\vert A \vert$. We denote the set of endomorphisms on $A$ by $\Lin(A)$ and the set of unitaries on $A$ by $\Unitary(A)$. The L\"{o}wner order on linear operators is denoted $X \geq Y$ if and only if $X-Y$ is positive semidefinite. We denote the set of positive semidefinite operators on $A$ by $\Pos(A)$. We denote the identity linear operator on $A$ by $I_{A}$. Quantum states (density matrices) are trace one positive semidefinite operators and we denote the set of density matrices on Hilbert space $A \otimes B$ by $\Density(A \otimes B)$. These represent quantum systems. We often will specify quantum states with subscripts denoting the spaces they act on, e.g. $\rho_{AB} \in \Density(A \otimes B)$. Note that a probability distribution on a finite alphabet $\cX$, $p_{X}$, can be represented as a quantum state on $\mbb{C}^{\vert \cX \vert} \eqqcolon X$, e.g. $\rho_{X} = \sum_{x} p_{X}(x) \dyad{x}$. This allows us to specify classical-quantum (CQ) states on $X$ and $A$ respectively via 
\begin{align}
    \rho_{XA} = \sum_{x \in \cX} p_{X}(x)\dyad{x} \otimes \rho^{x}_{A} \ ,
\end{align}
where $p_{X}$ is a probability distribution and $\{\rho^{x}_{A}\}_{x \in \cX} \subset \Density(A)$. Such states will be central to this work. We also let $\Density_{\leq}(A \otimes B)$ denote `sub-normalized' density matrices, which are positive semidefinite operators with trace at most one.

We denote the maximally mixed state on $A$ by $\pi_{A} = I_{A}/\vert A \vert$.  We denote the perfectly correlated state on $\cX \otimes \cX'$ distributed according to probability distribution $p_{X}$ as $\chi^{\vert p}_{XX'} = \sum_{x} p_{X}(x)\dyad{x}_{X} \otimes \dyad{x}_{X'}$. When the distribution is uniform, we simplify the notation to $\chi_{XX'}$. We define the maximally entangled state on $A \otimes A'$ where $A' \cong A$ by $\ket{\Phi}_{AA'} \coloneq \frac{1}{\sqrt{\vert A \vert}} \sum_{i \in [\vert A \vert ]} \ket{i}_{A}\ket{i}_{A'}$. We denote the isotropic state on $A \otimes A'$ of parameter $\lambda \in [0,1]$ by 
\begin{align}\label{eq:isotropic-state-defn}
    \rho_{\lambda} \coloneq (1-\lambda)\dyad{\Phi^{+}} + \lambda \pi_{AA'} \ . 
\end{align}

\paragraph{Quantum Channels} Quantum channels are completely positive, trace-preserving (CPTP) maps from the endomorphisms on one Hilbert space to another and represent dynamics on quantum systems. We denote the set of CPTP maps from endomorphisms on $A$ to endomorphisms on $B$ by $\Channel(A,B)$. We often add subscripts to make clear what the space the linear operators that the channels acts on are, e.g. $\cE_{A \to B} \in \Channel(A,B)$. As classical channels are a subset of quantum channels when one represents probability distributions as quantum states, we let $P_{Y \vert X}$ represent both the joint distribution $P_{Y \vert X}$ and the quantum channel representation $P_{X \to Y}$. Particularly important channels for this work are:
\begin{enumerate}
    \item the noiseless \textit{quantum} channel on a $d$-dimensional quantum system, which we denote $\id_{d} \coloneq \id_{\mbb{C}^{d} \to \mbb{C}^{d}}$, and 
    \item the noiseless \textit{classical} channel on a $d$-dimensional quantum system, which we denote $\Delta_{d}$.
\end{enumerate} 
In the above, $\Delta_{d}$ is really the $d$-dimensional completely dephasing channel with respect to some implicit computational basis. This is the appropriate choice for classical communication as it will transmit all classical information in that basis noiselessly but dephase any quantum state into that classical basis.

Many of the channels we will consider will have at least two input spaces or output spaces, which is straightforward to specify using the subscripts, e.g. $\cE_{AB \to CD}$ maps linear operators on $A \otimes B$ to $C \otimes D$. Such channels have further properties. For example, we will need the following notions of non-signaling for multipartite quantum channels.
\begin{definition}\cite{eggeling2002semicausal,on-quantum-ns-boxes}
    Let $\cE_{AB \to A'B'}$ be a (bipartite) quantum channel.
    \begin{enumerate}[itemsep=0pt]
        \item $\cE$ is $A \to B$ non-signaling, denoted $A \not \to B$, if there exists a channel $\cG_{A \to A'}$ such that $\Tr_{B'}[\Gamma^{\cE}] = \Gamma^{\cG} \otimes I_{B}$ where $\Gamma^{\cE}$ denotes the Choi operator $\cE$ and similarly for $\Gamma^{\cG}$.
        \item $\cE$ is non-signaling if it is both $A \not \to B$ and $B \not \to A$.
        \item $\cE$ is $A \to B$ semilocalizable if there exist channels $\cF^{0}_{A \to A'M}$, $\cF^{1}_{MB \to B'}$ such that 
        $$\cE = (\id_{A'} \otimes \cF^{1}) \circ (\cF^{0} \otimes \id_{B}) \ . $$
    \end{enumerate}
\end{definition}

\begin{proposition}\cite{eggeling2002semicausal} \cite[Exercise 2.5]{WatrousBook} \label{prop:NS-equiv-semilocalizable}
    A quantum channel $\cE_{AB \to A'B'}$ is $A \not \to B$ if and only if it is $B \to A$ semilocalizable. Moreover the channels $\cF^{0}_{B \to B'M}$ and $\cF^{1}_{AM \to A'}$ are such that $\vert M \vert \leq \vert B \vert \cdot \vert B' \vert$.
\end{proposition}

\paragraph{Multiple-Access Channels}
A multiple-access channel (MAC) is a classical channel with multiple input alphabets and a single output alphabet. A canonical physical example is a cell tower which receives messages from many cell phones. Mathematically, a MAC is represented by a conditional distribution over a joint input alphabet and a single output alphabet. For example, a MAC may take inputs $(x,y) \in \cX \times \cY$ and output a value $z \in \cZ$. Such an example is thus represented by a conditional distribution $P_{Z \vert XY}$. The above example straightforwardly generalizes to $k$ senders where the $i^{th}$ input is an alphabet $\cW_{i}$ and the total MAC is represented by a conditional distribution $P_{Z \vert W^{k}}$ where $\cW^{k} \coloneq \bigtimes_{i} \cW_{i}$ (note it is allowed that $\vert\cW_{i}\vert \neq \vert\cW_{i'}\vert$ for $i \neq i'$).

\paragraph{Quantum Multiaccess Networks and Signaling Dimension} For our purpose of thinking about (quantum) communication, we take a network as a set of nodes that may be identified as users. These nodes are then connected through various resources such as shared entanglement, classical or quantum communication channels, non-signaling boxes, and local processing which we will define as needed subsequently. The specification of these connections and resources, we call a `configuration.' In principle, certain nodes may take inputs which may be classical or quantum. Rather than formally define such a generic network and add too much notation, we simply define the specific network configurations relevant to this work as we go along (see e.g.~\cite{chiribella2009theoretical} for a development of generic quantum networks). 

For any $k \in \mbb{N}$, a $k$-sender multiple-access network (MN) is any configuration of $k+1$ nodes and resources such that we identify the first $k$ nodes as senders and the $k+1^{th}$ node. Such a MN can be used to simulate a MAC if sender $i$ takes an input from $\cW_{i}$ and the receiver ultimately outputs a value in $\cZ$ so that the entire MN defines a MAC from $\cW^{k}$ to $\cZ$. Note this allows us to define parameterized \textit{families} of MNs, which is important as our interest ultimately is in understanding the resources a MN needs to simulate specific MACs. As an example of a parameterized family: if there are two senders and we only specify that sender $1$ (resp.~$2$) is connected to the receiver by a noiseless quantum channel of dimension $d_{1}$ (resp.~$d_{2}$), then we have not specified what the senders do to process their inputs nor what the receiver does to process what it receives to output a value $z \in \cZ$. Thus, this induces a family of quantum MNs that induce a family of MACs. The following defines families of MNs that we will focus on in particular detail as well as the notion of `signaling dimension' it induces.
\begin{definition}[Main Families of MNs]\label{def:main-families-of-MNs}
    Let $k \in \mbb{N}$ and $\vec{d} \in \mbb{N}^{\times k}$.
    \begin{enumerate}
        \item $\mbf{C}(\vec{d})$ denotes the set of MNs where sender $i$ and the receiver are connected by $\Delta_{d_{i}}$, a noiseless classical channel of dimension $d_{i}$, and there are no further shared resources. We call these \textbf{classical MNs (CMNs)}.
        \item $\mbf{C}^{E}(\vec{d})$ denotes the same setting as $\mbf{C}(\vec{d})$ except the senders may share arbitrary entanglement. We call these \textbf{entanglement-assisted (EA) classical MNs (EACMNS)}.
        \item $\mbf{C}^{N}(\vec{d})$ denotes the same setting as $\mbf{C}(\vec{d})$ except the senders may share an arbitrary non-signaling box for pre-processing their inputs. We call these \textbf{fully classical non-signalling MNS (FCNSMNs)}.
        \item $\mbf{Q}(\vec{d})$ denotes the set of MNs where sender $i$ and the receiver are connected by $\id_{d_{i}}$, a noiseless quantum channel of dimension $d_{i}$, and there are no further shared resources. We call these \textbf{quantum MNs (QMNs)}.
        \item $\mbf{Q}^{E}(\vec{d})$ denotes the same setting as $\mbf{Q}(\vec{d})$ except the senders may share arbitrary entanglement. We call these \textbf{entanglement-assisted (EA) QMNs (EAQMNs)}
    \end{enumerate}
\end{definition}
Unlike in the main text, we parameterize the families of MNs above in terms of their signaling dimension $d$ rather than its logarithm, the number of communicated bits or qubits. This is both for notational simplicity and mathematical generality. We refer to an arbitrary family of MNs parameterized by the amount of signaling $\vec{d}$. For general MNs we use the notation $\mbf{S}(\vec{d})$, and if the number of senders is small, we write out the dimensions explicitly, \textit{e.g.} $\mbf{S}(d_1,d_2)$. Using this, we abstractly define the signaling dimension to quantify the amount of communication used by each sender in the MN.
\begin{definition}[Signaling Dimension of a MN]\label{def:signaling-dimension}
    Given a set of MNs parameterized by $\vec{d} \in \mbb{N}^{\times k}$, $\mbf{S}(\vec{d})$, we refer to $\vec{d}$ as the signaling dimension vector where $d_i$ denotes the signaling dimension of the $i^{th}$ sender's communication channel.
\end{definition}
\begin{definition}[Signaling Dimension]\label{def:signaling-dimension-of-pt-to-pt-chan} \cite{dallarno2017_no_hypersignaling}
For a 1-sender MN $\mbf{S}$ with classical input and output,\footnote{A 1-sender MN is equivalent to what is commonly called a point-to-point communication channel.} the signaling dimension is the smallest value $d \in \mbb{N}$ such that $\mbf{S}\subseteq \mbf{C}(d)$ for all input and output alphabet sizes.
\end{definition}
\begin{remark}[Signaling Dimension of Noiseless Quantum and Classical Channels]
\cite{doolittle_2021_certifying_classical_simulation_cost}
Since we consider noiseless communication channels, the signaling dimension of a classical channel is equivalent to size of the coding alphabet, for example, $n$-bit communication channel has signaling dimension $d=2^n$. The signaling dimension of a quantum channel is bound by the Hilbert space dimension of the quantum state coding alphabet, for example an $n$-qubit communication channel has signaling dimension $d=2^n$.

\end{remark}

% Third, using this above, we identify an element of a family of MNs in the following manner:
% \begin{align}
%     \cE \in \mbf{Q}^{EA}(\vec{d}) \text{ iff  } \cE \text{ can be implemented using a MN configuration in } \mbf{Q}^{EA}(\vec{d}) \ ,
% \end{align}
% and similarly for any other set of MNs {\color{blue} maybe can get rid of the above}.

 Next, we remark that each of these families of MNs are such that the receiver can receive all of the transmissions from the senders before doing any processing. That is, without loss of generality, one can partition the MN into an `encoding and transmission step' in which the senders take in their inputs, do their processing, and forward the resulting state to the receiver, and a `decoding step,' where the receiver applies a map to the total message. Most MN configurations without timing constraints satisfy this.\footnote{A family that does not \textit{directly} satisfy such a structure are a set of MNs where all senders and the receiver share randomness (Definition \ref{def:SR-assisted-MNs}). This is because the shared randomness means the total channel actually decomposes into a convex combination of MNs that decompose in this way. Nonetheless, one would expect most MNs to effectively be of the structure we describe.} The following formalizes the above idea and names such classes of MNs.
\begin{definition}[Simultaneous Message Passing MNs]\label{def:SMP-MN}
    We say a set of MNs $\mbf{S}$ are a set of simultaneous message passing (SMP) MNs if without loss of generality the total action of the MN can be decomposed into an `encoding-and-transmission channel' $\cE^{\net}$, which is the composition of everything the senders do to their inputs as well as the transmission to the receiver, and a `decoding channel' $\cD^{\net}$, which is what the receiver does. More formally, a set of MNs are SMP MNs if we may always express the total channel as $\cD^{\net} \circ \cE^{\net}$.
\end{definition}
\noindent We use the term SMP MNs because it extends the simultaneous message passing model of communication complexity introduced by Yao \cite{yao1979_communication_complexity}, in which the inputs are classical and the output is a Boolean function (see \cite{buhrman2010_nonlocality_communication_complexity} for a discussion of quantum results in this model). An example of a SMP MN is illustrated in Fig.~\ref{fig:SMP-MN}.

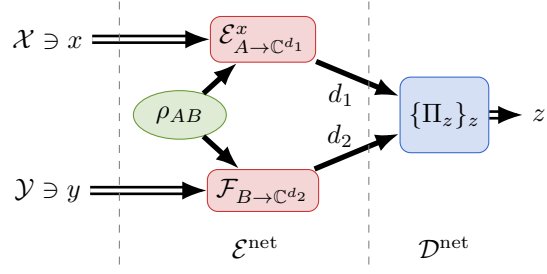
\begin{figure}
    \centering
        \begin{tikzpicture}
            % x node set with absolute coordinates
            \node[terminal] (x) at (-1.75,1) {$\cX \ni x$};
            \node[qsource] (rho) at (0,0) {$\rho_{AB}$};
            \node[terminal] (y) at (-1.75,-1) {$\cY \ni y$};    
            \node[proc_dev] (A) at (1.1,1) {$\cE_{A \to \mbb{C}^{d_{1}}}^{x}$};
            \node[proc_dev] (B) at (1.1,-1) {$\cF_{B \to \mbb{C}^{d_{2}}}$};
            \node[meas_dev, minimum height = 1cm] (R) at (3.5, 0) {$\{\Pi_{z}\}_{z}$};
            \node[terminal] (z) at (4.75, 0) {$z$};
            \node[terminal] (enc) at (1,-1.75) {$\cE^{\net}$};
            \node[terminal] (dec) at (3.5,-1.75) {$\cD^{\net}$};
        
            \path (rho) \qedge (A);
            \path (rho) \qedge (B);
            \path (x) \cedge (A);
            \path (y) \cedge (B);
            \path (A) \qedge node[below,pos=0.3] {$d_{1}$} (R);
            \path (B) \qedge node[above,pos=0.3] {$d_{2}$} (R);
            \path (R) \cedge (z);

            \draw[gray, dashed] (2.5,1.5) -- (2.5,-2);
            \draw[gray, dashed] (-0.8,1.5) -- (-0.8,-2);
        \end{tikzpicture} 
    \caption{An example of a simultaneous message passing multiaccess network (Definition \ref{def:SMP-MN}) with signaling dimension $\vec{d}=(d_{1},d_{2})$ as made explicit by the labeled arrows. In this case, senders $1$ and $2$ share entanglement through the state $\rho_{AB}$. They then process it conditioned on their respective inputs to send over the noiseless quantum channels to the receiver, who then measures the received joint state to output a value of $z$. This MN is simultaneous message passing as the two senders cannot interact and the receiver receives the entire joint state prior to having to process it.
    }
    \label{fig:SMP-MN}
\end{figure}

Finally, we make the following observations on the strength of different MNs for fixed signaling dimension.
\begin{proposition}\label{prop:MN-containments}
    Let $k \in \mbb{N}$ and $\vec{d} \in \mbb{N}^{\times k}$. Then,
    \begin{align}
        \mbf{C}(\vec{d}) \subseteq \mbf{C}^{E}(\vec{d}) \subseteq \mbf{C}^{N}(\vec{d}) \label{eq:classical-set-containments} \\
        \mbf{C}(\vec{d}) \subseteq \mbf{Q}(\vec{d}) \subseteq \mbf{Q}^{E}(\vec{d}) \ . \label{eq:quantum_set_containments}
    \end{align}
\end{proposition}
\begin{proof}
    The first inclusion in \eqref{eq:classical-set-containments} follows from the fact adding entanglement-assistance can only result in a larger set of network configurations. The second inclusion follows from the fact that independent classical channels with entanglement-assistance induces a classical non-signaling box from $\cW^{k} \to \cC^{k}$. To see this latter point, if the inputs are on Hilbert spaces $A_{i}$ for each $i \in [k]$, the users have shared state $\rho_{Q^{k}} \in \Density(\otimes_{i \in [k] Q_{i}})$, then they can apply local measurement channels $\cM^{i}_{Q_{i}A_{i}}$ to send over their local channels. Then the output of the $j^{th}$ party over the total channel for any input $\sigma_{A^{k}}$ is 
    $$\Tr_{C^{k}\setminus C_{j}}[(\otimes_{i \in [k]} \Delta_{id_{d_{i}}} \circ \cM_{i})(\rho_{Q^{k}} \otimes \sigma_{A^{k}})] = (\Delta_{id_{d_{i}}} \circ \cM_{j})(\rho_{Q_{j}} \otimes \sigma_{A_{j}}) \ , $$
    which is a well-defined channel from $A_{j}$ to $C_{j}$ that appends $\rho_{Q_{j}}$ to input $\sigma_{A_{j}}$ and then applies $\Delta_{id_{d_{i}}} \circ \cM_{i}$.
    
    The second set of containments in Eq.~\eqref{eq:quantum_set_containments} follow from noiseless quantum communication being more general than noiseless classical communication and that entanglement-assistance can only be a larger set of networks.
\end{proof}

\paragraph{Implementing Functions over Multiaccess Networks} 
We investigate the performance of computing multivariable functions over MNs with constrained communication resources. We quantify the performance of a MN using the success probability developed in reference \cite[Section 2.2.3]{doolittle2024operational_nonclassicality}.
\begin{definition}\label{def:function-simulation-success-prob}
    Let $f: \cW^{k} \to \cZ$ be a function. Let $Q_{Z \vert W^{k}}$ be a conditional distribution. Then the average success probability of computing $f$ over $Q_{Z \vert W^{k}}$ is
    \begin{align}\label{eq:function-guessing-success-prob}
        P_{S}^{f}(Q_{Z \vert W^{k}}) \coloneq \sum_{w^{k} \in \cW^{k}} \frac{1}{\vert \cW^{k} \vert} \delta_{f(w^{k}),z} Q_{Z \vert W^{k}}(z \vert w^{k}) \ . 
    \end{align}
    Let $\mbf{S}$ be a set of MNs. Then the optimal success probability of computing $f$ over this set is
    \begin{align}\label{eq:optimal-simulation-prob}
        P_{S}^{f}(\mbf{S}) \coloneq \sup_{Q_{Z \vert W^{k}} \in \mbf{S}} P_{S}^{f}(Q_{Z \vert W^{k}}) \ . 
    \end{align}
\end{definition}

\paragraph{On Shared Randomness Assistance}
As we restrict to simulating MACs measured by a cost function $f$ that is linear in a conditional distribution in \eqref{eq:function-guessing-success-prob}, we can use standard convex analysis arguments to show that sharing randomness between the senders and receiver does not improve the success probability of implementing the function $f$. It is well known that globally shared randomness induces convexity in the set of network communication strategies \cite{Bowles2015_nonclassicality_communication_networks,doolittle2024operational_nonclassicality}, and therefore no communication advantage is gained by sharing randomness, we explicitly prove it here for the case of MNs for completeness.

\begin{definition}[Shared-Randomness Assisted SMPs]\label{def:SR-assisted-MNs}
    Let $\mbf{S}(\vec{d})$ be a set of SMP MNs. We define $\mbf{S}^{SR}(\vec{d})$ as the set of MNs generated by the senders and the receiver sharing unbounded randomness.
\end{definition}
Clearly $\mbf{S}(\vec{d}) \subseteq \mbf{S}^{SR}(\vec{d})$ as the shared randomness can be ignored. Moreover, $\mbf{S}^{SR}(\vec{d})$ is the convex hull of the strategies implementable by $\mbf{S}(\vec{d})$ as the senders and receivers may use the shared randomness to average their strategy.

\begin{proposition}\label{prop:no-need-for-SR} 
    Let $\mbf{S}(\vec{d})$ be a set of SMP MNs defined by the senders sharing resources defined by a tuple of compact, convex sets each with a finite number of extreme points and the receiver's decoders being a compact, convex set with a finite number of extreme points. Then $P_{S}^{f}(\mbf{S}^{SR}(\vec{d})) = P_{S}^{f}(\mbf{S}(\vec{d}))$ and is achieved at an extreme point. In particular, for all $\vec{d}$ and function $f: \bigtimes_{i \in [k]} \cW_{i} \to \cZ$,
    \begin{align}
        P_{S}^{f}(\mbf{C}^{SR}(\vec{d})) = P_{S}^{f}(\mbf{C}(\vec{d})) = \max_{\{h_{i}: \cW_{i} \to \cC_{i}\}_{i \in [k]}} p_{g}(Z \vert C^{k})_{(\id_{Z} \otimes \bigotimes_{i} h_{i})(p_{ZW^{k}})} \label{eq:SR-does-not-help-CMN} \\
        P_{S}^{f}(\mbf{C}^{N,SR}(\vec{d})) = P_{S}^{f}(\mbf{C}^{N}(\vec{d})) \label{eq:SR-does-not-help-NS-ass-CMN} \ . 
    \end{align}
\end{proposition}
\begin{proof}
    By the assumption on the resources in the configuration being compact and convex, the Krein-Milman theorem implies each resource can be decomposed into a convex combination of the extreme points of the set. As the extreme points of a Cartesian product of convex sets is the Cartesian product of the extreme points of the convex sets and we assumed each resource has a finite number of extreme points, there is a finite number of extreme points of the total stragies. As $\mbf{S}^{SR}(\vec{d})$ generates the convex hull of the strategies in $\mbf{S}(\vec{d})$, we may conclude $Q_{Z \vert W^{k}} \in \mbf{S}^{SR}(\vec{d})$ can be decomposed into a finite convex combination of $\{Q^{\lambda}_{Z \vert W^{k}}\}_{\lambda \in \Lambda}$ where each $Q^{\lambda}_{Z \vert W^{k}}$ is constructed from extreme points of the resources. Then, by linearity in $Q_{Z \vert W^{k}}$ of the success probability of computing $f$
    \begin{align}
        P_{S}^{f}(Q_{Z \vert W^{k}}) = \sum_{\lambda \in \Lambda} p(\lambda) P_{S}^{f}(Q^{\lambda}_{Z \vert W^{k}}) \leq \max_{\lambda \in \Lambda} P_{S}^{f}(Q^{\lambda}_{Z \vert W^{k}}) \ .
    \end{align}
    As $Q_{Z \vert W^{k}} \in \mbf{S}^{SR}(\vec{d})$ was arbitrary and $\text{argmax}_{\lambda \in \Lambda} P_{S}^{f}(Q^{\lambda}_{Z \vert W^{k}}) \in \mbf{S}(\vec{d})$, we may conclude $P_{S}^{f}(\mbf{S}^{SR}(\vec{d})) \leq P_{S}^{f}(\mbf{S}(\vec{d})$. 

    To prove equations \eqref{eq:SR-does-not-help-CMN} and \eqref{eq:SR-does-not-help-NS-ass-CMN}, we consider the two cases separately. For $\mbf{C}(\vec{d})$ where $\vec{d} \in \mbb{N}^{\times k}$, note the set of resources are sender $i$'s local conditional distribution $P_{Y_{i} \vert W_{i}}$ where $\vert Y_{i} \vert = d_{i}$ for encoding and the receiver's decoder, which is a conditional distribution $P_{Z \vert Y^{k}}$. The extreme points of the conditional distributions are a finite set of deterministic functions when the input and output spaces are fixed. Thus the above conditions apply. 

    For $\mbf{C}^{N}(\vec{d})$, the set of resources is without loss of generality non-signaling boxes from $\cW^{k}$ to $\cY^{k}$ where $\vert Y_{i} \vert = d_{i}$ and the receiver's decoder, which is a conditional distribution $P_{Z \vert Y^{k}}$. The reason the non-signaling boxes are of this form without loss of generality is because the pre-processing non-signaling box must compose with the transmission channels determined by the signaling dimension. As non-signaling boxes of fixed input and output dimensions form a polytope, which by definition has a finite number of extreme points, the set of shared resource by the senders is a convex set of with a finite number of extreme points.
\end{proof}
\begin{remark}\label{rem:on-non-uniformity}
    We note a straightforward generalization shows the above argument still holds if one replaces $\frac{1}{\vert \cW^{k} \vert}$ in \eqref{eq:function-guessing-success-prob} with a different distribution over $\cW^{k}$.
\end{remark}

Proposition~\ref{prop:no-need-for-SR} neither shows that $P_{S}^{f}(\mbf{C}^{E}(\vec{d}))$ can be achieved (i.e.~a maximum rather than a supremum in \eqref{eq:optimal-simulation-prob}) nor that shared randomness cannot improve the success probability. The former is not shown because the set of quantum states is not compact without bounding the dimension of the entanglement-assistance. The latter is not shown because even in finite-dimensions, the extreme points of the density matrices is not finite. Nonetheless, a remarkable theorem by Frenkel and Weiner \cite{Frenkel2015_classical_information_n-level_quantum_system} shows that for point-to-point channel simulation, quantum communication with shared randomness is no better than classical communication with shared randomness.

\begin{proposition}\cite[Theorem 3]{Frenkel2015_classical_information_n-level_quantum_system} \label{prop:Frenkel-Weiner}
     Let $d \in \mbb{N}$. If $P_{Y \vert X} \in \mbf{Q}^{SR}(d)$, then $P_{Y \vert X} \in \mbf{C}^{SR}(d)$. 
\end{proposition}

\paragraph{Guessing Probability and Min-Entropy}
Given a classical-quantum state $\rho_{ZQ} = \sum_{z \in \cZ} p_{Z}(z)\dyad{z} \otimes \rho_{Q}^{z}$, the optimal probability of guessing the value of $z$ according to quantum mechanics using the $Q$ system, or \textit{guessing probability} is defined as 
\begin{align}\label{eq:guessing-probability}
    p_{g}(Z \vert Q)_{\rho} \coloneq \max_{\substack{\{\Gamma_{z}\}_{z \in \cZ} \subset \Pos(A): \\ \sum_{z} \Gamma_{z} = I_{Q}}} \sum_{z} p_{Z}(z)\Tr[\Gamma_{z}\rho^{z}_{Q}] \ ,  %\label{eq:guessing_prob_sdp}
\end{align}
where note the maximization is over positive-operator-valued-measures with $\vert \cZ \vert$ outcomes. The guessing probability in Eq.~\eqref{eq:guessing-probability} is equivalent to the SDP for maximizing success in quantum state discrimination for discriminating the states $\{\rho^{z}_{Q}\}_{z}$ drawn according to distribution $p_{Z}$ (see \cite[Section 3.1.2]{WatrousBook} for an introduction and \cite{Bae_2015} and references therein for further details). Given a possibly sub-normalized quantum state $\rho_{AB} \in \Density_{\leq}(A \otimes B)$, the min-entropy is defined as
\begin{align}
    H_{\min}(A\vert B)_{\rho} \coloneq \max_{\sigma_{B} \in \Density_{\leq}(B)} - \log \left\Vert \sigma_{B}^{-1/2} \rho_{AB} \sigma_{B}^{-1/2} \right\Vert \ .
\end{align}
It was shown in \cite{konig2009operational} that for classical-quantum states $\rho_{ZQ}$,
\begin{align}\label{eq:guess-prob-and-min-ent}
    \exp(-H_{\min}(Z \vert Q)_{\rho}) = p_{g}(Z \vert Q)_{\rho} \ . 
\end{align}
In the subsequent sections, we will apply the identity in Eq.~\eqref{eq:guess-prob-and-min-ent} and the following key properties of the guessing probability (equivalently, min-entropy of classical-quantum states).
\begin{proposition}[Properties of Guessing Probability] Let $\rho_{ZQ} = \sum_{z} p_{Z}(z)\dyad{z} \otimes \rho^{z}_{Q}$ be a classical-quantum state. \label{prop:guessing-prob-properties}
    \begin{enumerate}
        \item \textbf{Data Processing Inequality:} For any quantum channel $\cE_{Q \to Q'}$, $p_{g}(Z \vert Q')_{\cE(\rho)} \leq p_{g}(Z \vert Q)_{\rho}$.
        \item \textbf{Multiplicativity of Guessing over Independent States:} Let $\sigma_{XB}$ be a classical-quantum state. Then,
        \begin{align}
            p_{g}(XZ \vert BQ)_{\rho \otimes \sigma} = p_{g}(X \vert B)_{\sigma} \cdot p_{g}(Z \vert Q)_{\rho} \ .
        \end{align}
        \item \textbf{Bounds on Guessing Probabilities:} $\frac{1}{\vert \cZ \vert} \leq p_{g}(Z)_{\rho} \leq p_{g}(Z \vert Q) \leq 1$ where the first inequality is saturated when $\rho_{Z} = \pi_{Z}$, the second is saturated when $\rho_{ZQ} = \rho_{Z} \otimes \rho_{Q}$, and the final inequality is saturated when $\{\rho_{Q}^{z}\}_{z}$ are mutually orthogonal, i.e. $\Tr[\rho_{Q}^{z}\rho_{Q}^{z'}] \propto \delta_{z,z'}$.
        \item \textbf{Decomposing Classical Side-Information:} If $Y$ is a classical register so $\rho_{ZQY} = \sum_{y} p(y) \dyad{y} \otimes \rho^{y}_{ZQ}$, then $p_{g}(Z \vert QY)_{\rho} = \sum_{y} p(y) p_{g}(Z \vert Q)_{\rho^{y}_{ZQ}}$.
        \item \textbf{Classical Side-Information Chain Rule:} Let  $\sigma_{ZQY}$ be a classical-quantum-classical state. Then
        \begin{align}\label{eq:guessing-prob-chain-rule}
            p_{g}(Z \vert QY)_{\sigma} &\leq \vert Y \vert p_{g}(Z \vert Q)_{\sigma} \ .
        \end{align}
    \end{enumerate}
\end{proposition}
\begin{proof}[Proof Summary]
    Items 1 and 2 are Items vi. and viii. of Table 4.1 in \cite{tomamichel-thesis}. Item 3 may be proven using $H_{\min}(A) \leq \log \vert A \vert$, the data processing inequality, and the Holevo-Helstrom theorem (c.f.~\cite[Theorem 3.4]{WatrousBook}). Item 4 is \cite[Eq.~6.25]{Tomamichel-Book}. Item 5 follows immediately from \cite[Lemma 6.18]{Tomamichel-Book} and \eqref{eq:guess-prob-and-min-ent}. 
\end{proof}

\begin{remark}[On the Chain Rule] We stress the proof of Item 5 of Proposition \ref{prop:guessing-prob-properties} requires the $Y$ system to be classical. This is because it follows from the chain rule for $\rho_{ABY}$, $H_{\min}(A \vert BY) \geq H_{\min}(A \vert B) - \log \vert Y \vert$. This chain rule can be violated if $Y$ is not classical. The canonical example is $\rho_{ABC} = \Phi_{AC}^{+} \otimes \dyad{0}_{B}$. Then
\begin{align}
    H_{\min}(A \vert BC)_{\rho} = H_{\min}(A \vert C)_{\Phi^{+}} + 0 = -\log \vert A \vert \qquad \text{and} \qquad H_{\min}(A \vert B) = H_{\min}(A) = \log \vert A \vert \ .  
\end{align}
Thus, $H_{\min}(A \vert BC) - H_{\min}(A \vert B) = -2\log \vert C \vert < -\log \vert C \vert $.
\end{remark}

Lastly, to motivate Section \ref{app:entropy-chain-rule}, we observe in Proposition~\ref{prop:easier-to-guess-output-of-function} that it is strictly easier to guess the output of a function than its input, even with quantum side-information.
\begin{proposition}[Easier to Guess the Output of a Function]\label{prop:easier-to-guess-output-of-function}
    Let $\rho_{XYQ} = \sum_{x,y} p(x,y)\dyad{x} \otimes \dyad{y} \otimes \rho^{x,y}_{Q}$ be a CCQ state and $f:\cX \times \cY \to \cZ$ be a function. Define $\rho_{ZYQ} \coloneq \rho_{f(X,Y)YQ}$. Then $H_{\min}(Z \vert YQ) \leq H_{\min}(X \vert YQ)$. Equivalently, $p_{g}(Z \vert YE) \geq p_{g}(X \vert YQ)$.
\end{proposition}
\begin{proof}
   We have $H_{\min}(X \vert Y Q) = H_{\min}(Z X \vert Y Q) \geq H_{\min}(Z \vert YQ)$ where the first equality is using \cite[Lemma 4]{george2022finite} in the case $\ve = 0$ and the replacements $Y \to Z$, $X \to X$, $P \to Y$, $E \to Q$ and the second is that the entropy of a classical register is non-negative \cite[Lemma 6.17]{Tomamichel-Book}. This proves the min-entropy statement and the guessing probability statement follows from \eqref{eq:guess-prob-and-min-ent}.
\end{proof}

\section{Dense Network Coding}\label{app:dense-network-coding}
In this appendix we generalize the dense network coding protocol given in the main text. We will show that the group operation $\cdot$ of a group $(G,\cdot)$ of size $d^{2}$ that admits a projective unitary representation on $\mbb{C}^{d}$ can be densely network coded (Theorem \ref{thm:dense-network-coding}). We will call such groups ``tightly network codeable (TNC) groups," because they cannot be dense network coded on a smaller space. We then formalize that dense network coding does not transmit all the information about the inputs unlike using superdense coding in parallel (Proposition \ref{prop:DNC-does-not-transmit-inputs}). We also characterize TNC groups by showing they are in exact correspondence with `nice error bases' \cite{knill1996group} whose characterization has previously been studied \cite{klappenecker2002beyond} (Section \ref{sec:alg-struc-of-TNC}). Conveniently, this identification allows us to distinguish between dense network coding and dense coding by comparing our family of groups with Werner's generic characterization of superdense coding \cite{werner2001-dense-coding}.

\subsection{Dense Network Coding of the Group Operation}\label{sec:general-dense-network-coding}
We now provide the general strategy for dense network coding of the group operation. We begin by recalling the definition of a group and some related notation.
\begin{definition}
    A group is an ordered pair $(G,\cdot)$ where $G$ is a set and $\cdot: G \times G \to G$ is a binary operation on $G$, the group operation, satisfying:
    \begin{enumerate}
        \item Associativity: $(g \cdot h) \cdot f = g \cdot (h \cdot f)$ for all $g,h,f \in G$,
        \item Identity: there exists $e \in G$ such that $g \cdot e = e \cdot g = g$ for all $g \in G$, and
        \item Inverse: for each $g \in G$, there exists an element $g^{-1} \in G$ such that $g \cdot g^{-1} = e = g^{-1} \cdot g$. 
    \end{enumerate}
\end{definition}
For the rest of this section, we will consider sets of unitaries indexed by the elements of $G$, e.g.~$\{U_{g}\}_{g \in G}$. It follows that one may express the index of the unitary in terms of other group elements and the group operation, e.g.~$U_{g \cdot h}$ is $U_{g'}$ where $g' = g \cdot h$.

With the notation fixed, we begin by defining the following set of groups.
\begin{definition}[Tightly Network Codeable Groups]\label{def:TNC}
    Let $(G,\cdot)$ be a group of order $d^{2}$ such that there exist a set of unitaries $\cU \coloneq \{U_{g}\}_{g \in G} \subset \Unitary(\mbb{C}^{d})$ satisfying
    \begin{enumerate}
        \item $\cU$ are an orthonormal basis of $\Lin(\mbb{C}^{d})$ under the normalized Hilbert-Schmidt inner product, i.e.
        \begin{align}\label{eq:ortho-basis-of-unitaries}
            \Tr[U_{g}^{\dagger}U_{h}] = d \delta_{g,h} \quad \forall g,h \in G \ . 
        \end{align}
        \item There exist $\{\omega(g,h)\}_{g,h \in G} \subset \Unitary(\mbb{C})$ such that
        \begin{align}\label{eq:proj-group-behaviour}
            U_{g}U_{h} = \omega(g,h)U_{g \cdot h} \quad \forall g,h \in G \ . 
        \end{align}
    \end{enumerate}
    We call such a group a tightly network codeable (TNC) group and the set $\cU$ its TNC representation.
\end{definition}
These representation-theoretic conditions on the group are useful due to the following.

\begin{proposition}\label{prop:TNC-coding}
    Let $(G,\cdot)$ be a TNC group of order $d^{2}$ and $\cU \subset \Unitary(d)$ its TNC representation. Define the set of states
    \begin{align}
        \ket{\Phi_{g}} \coloneq U_{g} \otimes I \ket{\Phi} \quad g \in \mbb{G} \ .
    \end{align}
    Then $\{\ket{\Phi_{g}}\}_{g \in G}$ are an orthonormal basis of $\mbb{C}^{d} \otimes \mbb{C}^{d}$. 
\end{proposition}
\begin{proof}
    Orthonormality follows from 
    \begin{align}
        \bra{\Phi_{g}}\ket{\Phi_{h}} = \bra{\Phi}U_{g}^{\dagger}U_{h} \otimes I \ket{\Phi} = \frac{1}{d}\Tr[U_{g}^{\dagger}U_{h}] = \delta_{g,h} \ ,  
    \end{align}
    where we used $\bra{\Phi}K \otimes I \ket{\Phi} = \frac{1}{d}\Tr[K]$ for $K \in \Lin(\mbb{C}^{d})$ and \eqref{eq:ortho-basis-of-unitaries}.
\end{proof}
Given the above, we may define the $G$-Bell state measurement as the one that projects onto the set of states $\{\ket{\Phi_{g}}\}_{g \in G}$, i.e. for $K \in \Lin(\mbb{C}^{d} \otimes \mbb{C}^{d})$
\begin{align}\label{eq:G-BSM}
        \cM^{G}(K) \coloneq \sum_{g \in G} \Tr[K\dyad{\Phi_{g}}]\dyad{g}  \ . 
\end{align}

These definitions are sufficient to prove our general dense network coding theorem.

\begin{protocoldesc}[H]
\caption{Dense Network Coding for Tightly Network Codeable Group $(G,\cdot)$ of Order $d^{2}$}\label{prot:dense-network-coding-of-some-type}
		\textbf{Inputs:} \\
		\hspace{0.5cm}
		\begin{tabular}{ l l}
			$g \in G$ & Sender 1's Input  \\
            $h \in G$ & Sender 2's Input \\
            $\ket{\Phi} \in \mbb{C}^{d} \otimes \mbb{C}^{d}$ & Maximally Entangled State Shared by Senders 1 and 2
		\end{tabular}
		
		\vspace{0.5cm}
		
		\textbf{Protocol:} 
		\begin{enumerate}
            \item On input $g$, Sender 1 applies $U_{g}$ to their local system of $\ket{\Phi}$.
            \item On input $h$, Sender 2 applies $U_{h}^{T}$ to their local system of $\ket{\Phi}$.
            \item Senders 1 and 2 forward their systems to the receiver.
            \item The receiver applies the $G$-Bell state measurement \eqref{eq:G-BSM} to its received systems obtaining outcome $\hat{g} \in G$, which it outputs as the answer.
        \end{enumerate}
\end{protocoldesc}

\begin{theorem}\label{thm:dense-network-coding}
    Let $(G,\cdot)$ be a TNC group of order $d^{2}$ and $f(g,h) \coloneq g \cdot h$ be the function defined by the group operation. Protocol \ref{prot:dense-network-coding-of-some-type} computes $f$ perfectly. That is $P_{S}^{f}(\mbf{Q}^{E}(d,d)) = 1$. Moreover, this strategy is symmetrically minimal in the sense that
    \begin{align}\label{eq:sym-minimal}
        P_{S}^{f}(\mbf{Q}^{E}(d_{1},d_{2})) < 1 \text{ if } \min\{d_{1},d_{2}\} < d \, \& \max\{d_{1},d_{2}\} \leq d \ . 
    \end{align}
\end{theorem}
\noindent We remark \eqref{eq:sym-minimal} is a justification of the `tightly' modifier in naming TNC groups.
\begin{proof}
    For the main claim, as Protocol \ref{prot:dense-network-coding-of-some-type} is manifestly a strategy using entanglement-assistance and quantum communication, it suffices to show the receiver always outputs $g \cdot h$ when Sender 1 and 2's inputs are $g$ and $h$ respectively. The receiver receives 
    \begin{align}\label{eq:receivers-received-state}
        U_{g} \otimes U_{h}^{T}\ket{\Phi} = U_{g}U_{h} \otimes I\ket{\Phi} = \omega(g,h)U_{g \cdot h} \otimes I \ket{\Phi} = \omega(g,h)\ket{\Phi_{g \cdot h}} \ , 
    \end{align}
    where we have used the transpose trick and \eqref{eq:proj-group-behaviour}. By Born's rule, we can then calculate the probability of each outcome conditioned on the inputs:
    \begin{align}
        \Pr[\hat{g} \vert g , h] = \Tr[\dyad{\Phi_{\hat{g}}}\omega(g,h)\dyad{\Phi_{g \cdot h}} \omega(g,h)^{\ast}] = \delta_{\hat{g},g \cdot h} \ , 
    \end{align}
    where we have used the orthonormality of the states. Thus, for every pair of inputs $(g,h)$, the receiver outputs the product $g \cdot h$. This completes the proof of the main claim.

    To establish \eqref{eq:sym-minimal}, let $X$ and $Y$ denote the registers storing Sender 1 and 2's inputs respectively. Then we can compute the marginal on the output of the group operation stored in register $Z$,
    \begin{align}
        p_{Z} = \Tr_{XY}[p_{ZXY}] &= \Tr_{XY}\left[ \frac{1}{\vert G \vert^{2}} \sum_{g,h \in G} \dyad{g \cdot h}_{Z} \otimes \dyad{g}_{X} \otimes \dyad{h}_{Y}\right] \\
        &= \frac{1}{\vert G \vert^{2}} \sum_{g,h \in G} \dyad{g \cdot h} \\
        &= \frac{1}{\vert G \vert^{2}} \sum_{g \in G} \left(\sum_{h} \dyad{g \cdot h} \right) \\
        &= \frac{1}{\vert G \vert^{2}} \sum_{g \in G} I_{Z} \label{eq:sum-over-left-group-mult-is-id} \\
        &= \frac{1}{\vert G \vert}I_{Z} = \pi_{Z} \ , 
    \end{align}
    where to get \eqref{eq:sum-over-left-group-mult-is-id} we used the vectors $\{\ket{g}\}_{g \in G}$ define the computational basis for the $Z$ register and that left multiplication of the group by a group element is a bijection, so the sum over $h \in G$ in the previous line results in the identity. As such, $p_{Z}$ is the uniform distribution. Thus, by applying Proposition \ref{prop:upper-bound-on-EAQMN-by-classical-channel} with Sender $i$'s signaling dimension being $d_{i}$, Item 3 of Proposition \ref{prop:guessing-prob-properties}, and the fact $d^{2} = \vert G \vert = \vert Z \vert$,
    \begin{align}
         P_{S}^{\cdot}(\mbf{Q}^{EA}(d_{1} , d_{2})) \leq d_{1} \cdot d_{2} \cdot p_{g}(Z)_{\pi} = \frac{d_{1} \cdot d_{2} }{d^{2}} \ , 
    \end{align}
    which is strictly less than one if $\min\{d_{1},d_{2}\} < d$ and $\max\{d_{1},d_{2}\} \leq d$.
\end{proof}

\begin{proposition}\label{prop:upper-bound-on-EAQMN-by-classical-channel}
    For any function $f:\cX \times \cY \to \cZ$ and $d_{1},d_{2} \in \mbb{N}$, 
    \begin{align}
        P_{S}^{f}(\mbf{Q}^{E}(d_{1} , d_{2})) \leq d_{1} \cdot d_{2} \cdot p_{g}(Z)_{p_{Z}} \ ,
    \end{align}
    where $p_{Z}$ is the marginal of $p_{ZXY} \coloneq \sum_{x,y} p(x,y)\dyad{f(x,y)}_{Z} \otimes \dyad{x} \otimes \dyad{y}$.
\end{proposition}
\begin{proof}
    One can replace the entanglement-assisted senders with a single channel classical-to-quantum $\cE_{XY \to Q}$ where $Q \cong \mbb{C}^{d_{1} \cdot d_{2}}$ as the entanglement-assisted strategy implements a specific case of such a strategy, i.e. 
    \begin{align}
        P_{S}^{f}(\mbf{Q}^{E}(d_{1} , d_{2})) \leq P_{S}^{f}(\mbf{Q}(d_{1} \cdot d_{2})) \ . 
    \end{align}
    By the Frenkel-Weiner Theorem (Proposition \ref{prop:Frenkel-Weiner}), we may implement the optimal quantum strategy using a classical channel where the sender and receiver share arbitrary randomness, so we have
    \begin{align}
        P_{S}^{f}(\mbf{Q}(d_{1} \cdot d_{2})) \leq P_{S}^{f}(\mbf{C}^{SR}(d_{1} \cdot d_{2}) \ . 
    \end{align}
    By Proposition \ref{prop:no-need-for-SR} and Remark \ref{rem:on-non-uniformity}, even though we have not assumed the inputs are uniformly distributed, the shared randomness does not improve the success probability over using just the classical channel, 
    \begin{align}
        P_{S}^{f}(\mbf{C}^{SR}(d_{1} \cdot d_{2})) \leq P_{S}^{f}(\mbf{C}(d_{1} \cdot d_{2})) \ . 
    \end{align}
    % That is, the shared randomness does not help Finally, by    Then by the Frenkel-Weiner Theorem (Proposition \ref{prop:Frenkel-Weiner}), Proposition \ref{prop:no-need-for-SR}, and Remark \ref{rem:on-non-uniformity},
    %Proposition \ref{prop:reduc-cl-with-SR-to-cl},
    Combining these inequalities in sequence and using Lemma \ref{lem:reduction-to-guessing-probability}, we have
    \begin{align}
        P_{S}^{f}(\mbf{Q}^{E}(d_{1} , d_{2})) \leq \max_{T_{C\vert XY}} p_{g}(Z \vert C)_{(\id_{Z} \bigotimes T)(p_{ZXY})} , 
    \end{align}
    where $C$ is a classical register of dimension $d_{1} \cdot d_{2}$. Finally, using Item 5 of Proposition \ref{prop:guessing-prob-properties} completes the proof.
\end{proof}

\paragraph{Dense Network Coding Does Not Transmit the Entire Input} In the main text we state that dense network coding does not transmit all of the information in the inputs to the receiver. This distinguishes it from each sender using superdense coding to transmit the inputs to the receiver to compute the function. To formalize that not all the information is transmitted, we show that the receiver cannot generally guess the inputs of the function when the senders perform dense network coding. 
\begin{proposition}\label{prop:DNC-does-not-transmit-inputs}
    Consider Protocol \ref{prot:dense-network-coding-of-some-type} with a distribution over inputs $\{p_{XY}(g,h)\}_{g,h \in G}$. If the senders follow Protocol \ref{prot:dense-network-coding-of-some-type}, then the receivers ability to guess the inputs correctly is given by
    \begin{align}
        p_{g}(XY \vert Z) = \sum_{g' \in G} \; \max_{g,h \in G: g \cdot h = g'} p_{XY}(g,h) \ .
    \end{align}
    In particular, if the inputs are independent and uniform, i.e. $p_{XY}(g,h) = \frac{1}{\vert G \vert^{2}}$ for all $g,h\in G$, then $p_g(XY|Z) = \frac{1}{d^{2}}$ and the receiver's ability to guess $g$ and $h$ is bounded.
\end{proposition}
\begin{proof}
    For clarity, label the quantum systems of Senders 1 and 2 in Protocol \ref{prot:dense-network-coding-of-some-type} by $S_{1}$ and $S_{2}$ so that the initial shared state is $\ket{\Phi}_{S_{1}S_{2}}$. It follows from \eqref{eq:receivers-received-state} that the joint state of the inputs, output, and the quantum systems after encoding is
    \begin{align}
        \rho_{XYS_{1}S_{2}} = \sum_{g,h \in G} p_{XY}(g,h) \dyad{g}_{X} \otimes \dyad{h}_{Y} \otimes \dyad{g \cdot h}_{Z} \otimes \dyad{\Phi_{g \cdot h}}_{S_{1}S_{2}} \ . 
    \end{align}
    Our interest is in $p_{g}(XY \vert S_{1}S_{2})_{\rho}$ as the $S_{1}S_{2}$ systems are what are available to the receiver. To establish the generic claim in the proposition, we will first show 
    \begin{align}\label{eq:guessing-inputs-same-from-just-output}
        p_{g}(XY \vert S_{1}S_{2})_{\rho} = p_{g}(XY \vert Z)_{\rho}
    \end{align} 
    To see this, define the isometry $V_{Z \to S_{1}S_{2}} = \sum_{g \in G} \ket{\Phi_{g}}_{S_{1}S_{2}}\bra{g}_{Z}$, a direct calculation shows $V\rho_{XYZ}V^{\dagger} = \rho_{XYS_{1}S_{2}}$. Combining this with \eqref{eq:guess-prob-and-min-ent} and the invariance of the min-entropy under local isometries \cite[Corollary 6.11]{Tomamichel-Book}, we conclude \eqref{eq:guessing-inputs-same-from-just-output}. Then, using \cite[Eq.~6.27]{Tomamichel-Book}, 
    \begin{align}
        p_{g}(XY \vert Z)_{\rho} = \sum_{g' \in G} \; \max_{g,h \in G: g \cdot h = g'} p_{XYZ}(g,h,g') = \sum_{g' \in G} \; \max_{g,h \in G: g \cdot h = g'} p_{XY}(g,h) \ , 
    \end{align}
    where we have used the fact that the value of $Z$ is completely determined by the pair $(g,h)$. This proves the main claim. To establish the claim for independent and uniform inputs, note that in that case $p_{XY}(g,h) = \frac{1}{d^{4}}$ for all $g,h \in G$, so one obtains $d^{2} \cdot \frac{1}{d^{4}} = \frac{1}{d^{2}}$ as claimed.
\end{proof}
Proposition~\ref{prop:DNC-does-not-transmit-inputs} shows that when Protocol \ref{prot:dense-network-coding-of-some-type} is performed, the receiver is unable to guess the inputs encoded by the senders.  
This bound on the receiver's guessing probability implies that not all of the input information was transmitted to the receiver.
Another way of seeing the above property is by using the von Neumann entropy. In this case,
\begin{align}
    H(XY \vert S_{1}S_{2}) = H(XY \vert Z) = H(XYZ) - H(Z) = H(XY) - H(Z) \ ,
\end{align}
where the first equality is isometric invariance as in the above proof, the second is the chain rule for conditional entropy, and the final equality is isometric invariance and that $Z$ can be computed isometrically from $X$ and $Y$. When the input distributions are uniform, it follows that
\begin{align}
    H(XY) - H(Z) = \log(\vert G \vert^{2})-\log(\vert G \vert) = \log(\vert G \vert) = 2 \log(d) \ , 
\end{align}
which can be interpreted as the receiver lacking two of the four dits of information input to the dense network coding protocol.

\subsection{Algebraic Structure of Tightly Network Codeable Groups}\label{sec:alg-struc-of-TNC}
Theorem \ref{thm:dense-network-coding} shows the group operation of a TNC group can always be computed using a dense network coding protocol. It is therefore natural to investigate the algebraic structure of TNC groups and identify examples of such groups. An immediate observation is the importance of \eqref{eq:proj-group-behaviour}, which we restate for reference:
\begin{align}\label{eq:projective-representation}
    U_{g}U_{h} = \omega(g,h)U_{g \cdot h} \quad \forall g,h \in G \ , 
\end{align}
where $\omega(g,h) \in \Unitary(\mbb{C})$. Such a condition tells us that the set of unitaries is by definition a projective representation of the group $(G,\cdot)$. We will develop the significance of this condition through the rest of this subsection.

\paragraph{Distinction from the Algebraic Structure for Superdense Coding}
~We begin by distinguishing dense network coding from superdense coding. Werner identified that any set $G$ of size $d^{2}$ with unitaries $\{U_{g}\}_{g \in G} \subset \Unitary(\mbb{C}^{d})$ satisfying \eqref{eq:ortho-basis-of-unitaries} defines a teleportation/dense coding protocol that achieves the `best results' in terms of minimal dimensions for the entangled state and maximal advantage over classical resources \cite[Theorem 1]{werner2001-dense-coding}. In dense network coding, we extend the definition of the set $G$ by equipping it with a binary multiplication operator that makes it a group and requiring the group to admit a projective representation. To explain why dense network coding requires these extra conditions relative to superdense coding, it suffices to examine the goals of each task. In superdense coding, the sender's input is a label $g \in G$ and the goal is for the decoder to determine it correctly. It follows that any encoding of the $\vert G \vert$ symbols into quantum states that allows for perfect decoding suffices for superdense coding. Such an encoding that uses shared entanglement is guaranteed by \eqref{eq:ortho-basis-of-unitaries} as shown in Proposition \ref{prop:TNC-coding}. In contrast, dense network coding considers two independent inputs from the set, $g,h \in G$, and the decoder must determine the output of a \textit{function} of these inputs, $f(g,h)$. Thus, the encoding must allow each symbol $g \in G$ to be perfectly decoded, and must perfectly encode the functions output $f(g,h)$. When the function corresponds to the group operation, these requirements are guaranteed to be satisfied as shown in \eqref{eq:projective-representation} and in the proof of Theorem \ref{thm:dense-network-coding}. Thus, the necessary algebraic structure of superdense coding differs from dense network coding, which is due to the difference in the task.

\paragraph{Correspondence Between TNC Groups and Nice Error Bases} ~We now turn to further characterizing TNC groups. To this end, it will be easiest to first show they correspond to groups introduced by Knill for error correction \cite{knill1996group} whose algebraic structure has already  been studied \cite{werner2001-dense-coding, klappenecker2002beyond, klappenecker2003unitary}.
\begin{definition}\cite{knill1996group}(see also \cite{klappenecker2002beyond}) \label{def:NEB}
    A nice error basis on $\mbb{C}^{d}$ is a set of unitaries indexed by a group $(G,\cdot)$ of order $d^{2}$ with identity element $e$, $\cE = \{U_{g}\}_{g \in G} \subset \Unitary(\mbb{C}^{d})$, such that
    \begin{enumerate}[itemsep=0pt]
        \item $U_{e} = I$
        \item $\Tr[U_{g}] = d \delta_{g,e}$ 
        \item For all $g,h \in G$, there exists $\omega(g,h) \in \mbb{C}\setminus\{0\}$ such that $U_{g}U_{h} = \omega(g,h)U_{g \cdot h}$. 
    \end{enumerate}
    The group $G$ is called the index group of the nice error basis.
\end{definition}

\begin{proposition}\label{prop:correspondence-between-NEB-and-TNC}
    A group $(G,\cdot)$ is the index group of a nice error basis if and only if it is a TNC group.
\end{proposition}
\begin{proof}
    $(\Longrightarrow)$ Let $(G,\cdot)$ be the index group of a nice error basis. Let $\{U_{g}\}_{g \in G}$ be the corresponding nice error basis. We will construct a set of unitaries indexed by $G$ that satisfy the conditions of Definition \ref{def:TNC}. As noted in \cite{knill1996group}, one can re-normalize the $U_{g}$ into $\{U_{g}'\}_{g}$ such that $\det[U_{g}'] = 1$ for all $g \in G$, which results in $\omega'(g,h)$ defined by $U_{g}'U_{h'} = \omega'(g,h)U_{g\cdot h}$ being a root of unity for all $g,h \in G$. We now wish to construct $\{U_{g}''\}_{g}$ with constants $\omega''(g,h)$ such that $\omega''(g,g^{-1}) = 1$ for all $g$. This is a known normalization achievable by a gauge transformation,\footnote{In representation theory parlance, it is normalizing the 2-cocycle specified by $\omega'(g,h)$.} but for completeness we provide it. 

     Let $\{\lambda(g)\}_{g \in G} \subset \Unitary(\mbb{C}^{d})$ be not yet specified and $U_{g}'' \coloneq \lambda(g)U_{g}'$ for all $g \in G$. By direct calculation,
     \begin{align}\label{eq:gauge-trans-identity}
         \omega''(g,h) = \frac{\lambda(g)\lambda(h)}{\lambda(gh)} \omega'(g,h) \ .
     \end{align}
     First we set $\lambda(e) = 1$. Noting Item 3 of Definition \ref{def:NEB} guarantees $\omega'(g,e) = \omega'(e,g) = 1$ for all $g \in G$, one may use \eqref{eq:gauge-trans-identity} to conclude $\omega''(g,e) = \omega''(e,g) = 1$ for all $g \in G$. This guarantees $U_{g}U_{e} = U_{g} = U_{e}U_{g}$ as needed. Next, using the associativity of matrix multiplication, Item 3 of Definition \ref{def:NEB} implies the cocycle identity
    \begin{align}
        \omega'(a,b)\omega'(ab,c) = \omega'(a,bc)\omega'(b,c) \quad \forall a,b,c \in G \ .
    \end{align}
    By setting $a=g$, $b=g^{-1}$, $c=g$, we have $\omega'(g,g^{-1}) = \omega'(g^{-1},g)$ and using \eqref{eq:gauge-trans-identity} and that $\lambda(e) =1$, 
    \begin{align}\label{eq:gauge-trans-for-g-and-inverse}
        \omega''(g,g^{-1}) = \lambda(g)\lambda(g^{-1})\omega'(g,g^{-1}) = \lambda(g^{-1})\lambda(g)\omega'(g^{-1},g) = \omega''(g^{-1},g) \ .
    \end{align}
    Thus it suffices to choose the remaining $\lambda(g)$ to guarantee $\omega''(g,g^{-1}) = 1$. By \eqref{eq:gauge-trans-for-g-and-inverse}, our aim is to choose the remaining $\lambda(g)$ to satisfy $\lambda(g)\lambda(g^{-1}) = \omega'(g,g^{-1})^{-1}$ for all $g \in G$. To do this, we iterate over all $g \in G$ such that when the value of $\lambda(g)$ has not yet been fixed:
    \begin{itemize}
        \item if $g = g^{-1}$, we let $\lambda(g)$ be the principal root of the square root of $\omega'(g,g^{-1})^{-1}$ 
        \item if $g \neq g^{-1}$, we let $\lambda(g) =1$ and $\lambda(g^{-1}) = \omega'(g,g^{-1})^{-1}$.
    \end{itemize}
    As $\omega'(g,h) \neq 0$, the values of $\lambda(g)$ are always well-defined and as $g$ is the inverse of $g^{-1}$, this iteration completely defines the assignment of $\lambda(g)$. In total, we have constructed a NEB $\{U_{g}''\}_{g \in G}$ where $U_{e} = I$, $\omega(g,e)=\omega(e,g)=\omega(e,e)= 1$ for all $g \in G$, and $\omega(g,g^{-1}) = 1$ for all $g \in G$. Note this implies 
    \begin{align}\label{eq:NEB-converted-to-nice-rep}
        U_{g}U_{g^{-1}} = I \Longrightarrow U_{g}^{\dagger} = U_{g^{-1}} \quad \forall g \in G \ ,
    \end{align}
    which is a standard identity for a non-projective unitary representation of a group.
    
    Finally, we show $\{U_{g}''\}_{g \in G}$ is a TNC representation of $(G,\cdot)$. As the $\omega(g,h) \in \Unitary(\mbb{C})$, Item 2 of Definition \ref{def:TNC} is satisfied, thus we just need to prove Item 1. We have
    \begin{align}
        \Tr[U_{g}^{\dagger}U_{h}] = \Tr[U_{g^{-1}}U_{h}] = \omega(g^{-1},h)\Tr[U_{g^{-1} \cdot h}] = \omega(g^{-1},h) d \delta_{g^{-1} \cdot h, e} \ , 
    \end{align}
    where we used \eqref{eq:NEB-converted-to-nice-rep} and then Item 2 of Definition \ref{def:NEB}. Moreover, as the inverse of a group element is unique, $\delta_{g^{-1} \cdot h,e} = \delta_{g,h}$. If $g = h$, then $\omega(g^{-1},h) = \omega(g^{-1},g) = 1$ by our construction and if $g \neq h$, then $\delta_{g,h} = 0$. Thus, $\omega(g^{-1},h) d \delta_{g^{-1} \cdot h, e} = d \delta_{g,h}$. Putting this together, $\Tr[U_{g}^{\dagger}U_{h}] = d \delta_{g,h}$ for all $g,h \in G$ as wanted. Thus, if $(G,\cdot)$ is the index group of a nice error basis, it is also a TNC group.
    
   $(\Longleftarrow)$ Let $(G,\cdot)$ be a TNC group. Let $\{U_{g}\}_{g \in G}$ be its TNC representation. We will construct a NEB indexed by $G$ from this. Let $e \in G$ denote the identity element. Define $\{V_{g} \coloneq U_{e}^{\dagger}U_{g} \}_{g \in G}$. Then $V_{e} = U_{e}^{\dagger}U_{e} = I$ as $U_{e}$ is a unitary. Next,
   \begin{align}
       \Tr[V_{g}] = \Tr[U_{e}^{\dagger}U_{g}] = d \delta_{g,e} \ ,
   \end{align}
   where the second equality is Item 1 of Definition \ref{def:TNC}. Now, by Item 2 of Definition \ref{def:TNC}, $U_{e}U_{g} = \omega(e,g)U_{e \circ g} = \omega(e,g)U_{g}$ for $g \in G$. By multiplying on the left hand side by $U_{e}^{\dagger}$, $U_{g} = \omega(e,g)U_{e}^{\dagger}U_{g}$ for $g \in G$. Re-writing, $\omega(e,g)^{-1}U_{g} = U_{e}^{\dagger}U_{g} = V_{g}$ for $g \in G$. Therefore, for any $g,h \in G$,
   \begin{align}
       V_{g}V_{h} &= \omega(e,g)^{-1} U_{g} \; \omega(e,h)^{-1}U_{h} \\
       &= \omega(e,g)^{-1}\omega(e,h)^{-1}\omega(g,h)U_{g \cdot h} \\
       &= \frac{\omega(g,h)\omega(e,g \cdot h)}{\omega(e,g)\omega(e,h)} V_{g \cdot h} \\
       &\coloneq \omega'(g,h) V_{g \cdot h} \ , 
    \end{align}
    where the second equality uses Item 2 of Definition \ref{def:TNC} and the rest is just $V_{g} = \omega(e,g)^{-1}$. As the $\omega$ terms are in $\Unitary(\mbb{C})$, they are closed under multiplication and division, so $\omega'(g,h) \in \Unitary(\mbb{C})$ for all $g,h$. Thus, we have established $\{V_{g}\}_{g \in G}$ are a nice error basis according to Definition \ref{def:NEB} and thus $G$ is the index group of said nice error basis.
\end{proof}

Given the above, the characterization of TNC groups is the same as the characterization of nice error bases, which was done to a large degree in \cite{klappenecker2002beyond}. An example application of this is that we are able to conclude the following.
\begin{proposition}\cite[Theorem 2]{klappenecker2002beyond}\label{prop:Abelian-TNC-groups}
If $(G,\cdot)$ is Abelian and a TNC group, then $G$ is of symmetric type (i.e. there exists Abelian $H$ such that $G \cong H \times H$). Conversely, any finite Abelian group of symmetric type is a TNC group.
\end{proposition}
\begin{proof}
    This is a re-statement of \cite[Theorem 2]{klappenecker2002beyond} in light of Proposition \ref{prop:correspondence-between-NEB-and-TNC}.
\end{proof}
For non-Abelian TNC groups (equivalently nice error bases), a clean characterization is not known (see the discussion after the proof of Theorem 3 in \cite{klappenecker2002beyond}). Nonetheless, in Step 2 of the proof of Theorem 3 in \cite{klappenecker2002beyond}, it is shown that one can construct a TNC group $G$ from any group $H$ of central type with cyclic center, thus a TNC group need not be Abelian. 

\paragraph{Simple Examples of TNC Groups from Discrete Weyl Operators} ~
Proposition \ref{prop:Abelian-TNC-groups} and the discussion thereafter show the existence of many TNC groups. Similar to many topics in quantum information theory, we believe the most natural family of TNC groups are those induced by the Discrete Weyl operators (recall~\eqref{eq:discrete-weyl-operators}). As follows from their commutation relations, discrete Weyl operators acting on $\mbb{C}^{d}$ are in fact the projective representation of the product group $\mbb{Z}_{d} \times \mbb{Z}_{d}$. As $\mbb{Z}_{d}$ is a cyclic group, and thus Abelian, $\mbb{Z}_{d} \times \mbb{Z}_{d}$ is also an example of a group relevant by Proposition \ref{prop:Abelian-TNC-groups}. Thus, the reason computing $\oplus_{d}^{2}$ via dense network coding arises is because it is the corresponding group operation.

A similar observation can be made for computing the bitwise XOR of two pairs of bitstrings of length $n$. In this case, one uses the Pauli string representation of $\mbb{Z}_{2}^{n} \times \mbb{Z}_{2}^{n}$ (which is a group of order $2^{2n}$). This follows by defining the Pauli and Discrete-Weyl strings:
\begin{align}
    X_{\vec{x}} \coloneq  \bigotimes_{i \in [n]} X^{x_{i}} \quad Z_{\vec{z}} \coloneq \bigotimes_{i \in [n]} Z^{z_{i}} \quad W_{\vec{x},\vec{z}} \coloneq \bigotimes_{i \in [n]} X^{x_{i}}Z^{z_{i}} \quad \forall x,z \in \{0,1\}^{n} \ , 
\end{align}
where $X$ and $Z$ in this case are the shift and phase operators (recall \eqref{eq:phase-and-shift-operators}) for $d=2$, i.e. the Pauli $X$ and $Z$ operators. The operators $W_{\vec{x},\vec{z}}$ are unitaries on $\mbb{C}^{2^{n}}$ that, using the commutation relations for discrete Weyl operators \cite[Eq.~4.75]{WatrousBook}, satisfy $W_{\vec{x},\vec{z}}W_{\vec{x}',\vec{z}'} = (-1)^{\vec{x}' \cdot \vec{z}}W_{\vec{x} \oplus_{2}^{n} \vec{x}', \vec{z} \oplus_{2}^{n} \vec{z}'}$. As $2^{n} = \sqrt{\vert \mbb{Z}_{2}^{n} \times \mbb{Z}_{2}^{n} \vert}$, we have shown this is an example of a TCN group by definition, and thus another function that may be computed using dense network coding. Finally, we remark you can of course further generalize bitstrings to ditstrings to unify both examples.

\section{From Channel Simulation to Guessing Probability}\label{app:chan-sim-to-guessing-prob}
Having introduced dense network coding, our goal is to show that dense network coding requires both entanglement-assistance and quantum communication to be achieved. In this section, we take the first step in this direction by re-expressing the optimal success probability of computing $f$ over a family of simultaneous message passing (SMP) MNs (Definition \ref{def:SMP-MN}) as an optimization of the guessing probability over the possible encodings.

\begin{lemma}\label{lem:reduction-to-guessing-probability}
    Let $f: \cW^{k} \to \cZ$ be a function. Let $\mbf{S}$ be a set of $k$-senders, one receiver MNs that are SMP MNs (Definition \ref{def:SMP-MN}) so that we can always identify an `encoding-and-transmission channel' $\cE^{\net}_{W^{k} \to C}$ where the $C$ register is the total message the receiver receives. Then
    \begin{align}
        P_{S}^{f}(\mbf{S}) = \sup_{\cE^{\net}} p_{g}(Z \vert C)_{(\id_{Z} \otimes \cE^{\net}_{W^{k} \to C})(\rho_{ZW^{k}})} \ , 
    \end{align}
    where the maximization is over encoding-and-transmission channels possible according to the set $\mbf{S}$ and $\rho_{ZW^{k}} = \frac{1}{\vert \cW^{k} \vert} \sum_{w^{k}} \dyad{f(w^{k})} \otimes \dyad{w^{k}}$.
\end{lemma}
\begin{proof}
    First, we make the appropriate identifications for a specific encoding-and-transmission channel of the senders and choice of POVM by the receiver to induce a MAC. We then optimize over the decoder to obtain the probability of guessing for that specific encoding-and-transmission channel. Finally, we optimize over the choice of encoding-and-transmission channel according to the set $\mbf{S}$ to complete the proof.
    
    Let $\cE^{\net}_{W^{k} \to C}$ be an encoding-and-transmission channel possible according to the specified set of MNs $\mbf{S}$. To induce a MAC, the receiver then applies some POVM $\{\Gamma_{z}\}_{z \in \cZ}$ to output a value $z \in \cZ$. With slightly more generality than in the proposition, we can let
    the initial joint state be $\rho_{ZW} = \sum_{w^{k} \in \cW^{k}} \Pr[w^{k}] \dyad{f(w^{k})}_{Z} \otimes \dyad{w^{k}}_{W^{k}}$. By linearity, the joint state after the senders' processing is
    \begin{align}
        \rho_{ZC} &= \sum_{w^{k}} \Pr[w^{k}] \dyad{f(w^{k})}_{Z} \otimes \cE^{\net}_{W^{k} \to C}(\dyad{w^{k}}) \\
        &= \sum_{z} \dyad{z}_{Z} \otimes \sum_{w^{k} \in f^{-1}(z)} \Pr[w^{k}]\cE^{\net}_{W^{k} \to C}(\dyad{w^{k}}) \eqqcolon \sum_{z} \dyad{z}_{Z} \otimes \check{\rho}_{C}^{z}  \ .
    \end{align}
    Note this means $\Pr[z]\rho^{z}_{C} = \check{\rho}_{C}^{z}$. By Born's rule, the MAC induced by the encoding-and-transmission and the receiver's POVM is
    \begin{align}
        Q^{\Gamma,\cE}_{Z \vert W^{k}}(z \vert w^{k}) = \Tr[\Gamma_{z}\cE^{\net}_{W^{k} \to C}(\dyad{w^{k}})] \ .
    \end{align}
    With these identifications,
    \begin{align}
        \sum_{z} \Pr[z]\Tr[\Gamma_{z}\rho^{z}_{C}] = \sum_{z} \Tr[\Gamma_{z} \check{\rho}^{z}_{C}] &= \sum_{z} \Tr[\Gamma_{z} \sum_{w^{k} \in f^{-1}(z)} \Pr[w^{k}]\cQ^{\net}_{W^{k} \to C}(\dyad{w^{k}})] \\
        &= \frac{1}{\vert \cW^{k} \vert} \sum_{z} \sum_{w^{k} \in f^{-1}(z)} \Tr[\Gamma_{z} \cQ^{\net}_{W^{k} \to C}(\dyad{w^{k}})] \\
        &= P^{f}_{S}(Q^{\Gamma,\cE}_{Z\vert W^{k}}) \ ,
    \end{align}
    where the third equality uses the assumption the inputs are uniform and the final equality is by \eqref{eq:function-guessing-success-prob}. By optimizing over the receivers' choice of POVM,
    \begin{align}
        \max_{\{\Gamma_{z}\}} P^{f}_{S}(Q^{\Gamma,\cE}_{Z\vert W^{k}}) = \max_{\{\Gamma_{z}\}} \sum_{z} \Pr[z]\Tr[\Gamma_{z}\rho^{z}_{C}] = p_{g}(Z \vert C)_{\rho_{ZC}} = p_{g}(Z \vert C)_{(\id_{Z} \otimes \cE^{\net})(\rho_{ZW^{k}})} \ ,
    \end{align}
    where we just the used the different definitions. Finally,
    \begin{align}
        P_{S}^{f}(\mbf{S}) = \sup_{\cE^{\net}_{W^{k} \to C},\{\Gamma^{z}\}} \sum_{z} P_{S}^{f}(Q^{\Gamma,\cE}_{Z \vert W^{k}})  = \sup_{\cE^{\net}_{W^{k} \to C}}   p_{g}(Z \vert C)_{(\id_{Z} \otimes \cE^{\net})(\rho_{ZW^{k}})} \ . 
    \end{align}
    This completes the proof.
\end{proof}
The reduction of channel simulation to guessing probability in Lemma \ref{lem:reduction-to-guessing-probability} acts as our starting point for placing limitations on channel simulation.
Instead of addressing channel simulation directly, we will establish limits on this guessing probability that only depend on the configuration of the MNs so that we can bound $P_{S}^{f}(\mbf{S})$ for specific choices of $\mbf{S}$. Before doing this, we note that as the probability of guessing is convex in the state \cite[Item iv. of Table 4.1]{tomamichel-thesis}, even when the set of $\cE^{\net}$ allowed by a given configuration of SMP MNs $\mbf{S}$ is convex, the optimization in Lemma \ref{lem:reduction-to-guessing-probability} is (equivalent to) minimizing a concave function over a convex set. Such optimizations problems are generally not easy to solve (e.g. there are instances that are NP hard) \cite[Chapter 2]{Horst1994}, so the reduction to maximizing the guessing probability over a set of network configurations alone does not make the problem easy to solve.

\section{Entropy Chain Rule for Conditionally Bijective Functions}\label{app:entropy-chain-rule}
As Proposition \ref{prop:easier-to-guess-output-of-function} shows, it is generally easier to guess the output of a function than one of its inputs. To put bounds on the ability to guess the value of a function over a (quantum multiaccess) network, we thus require a method for controlling that it may be easier to guess the output than the input. Here we identify a class of functions where, given one of the inputs as side-information, guessing the output is no easier than guessing the other input. We begin by naming this class of functions. 
\begin{definition}\label{eq:cond-bijectivity}
    Let $f: \cX \times \cY \to \cZ$. We say $f$ is $Y$-conditionally bijective if the function $f_{y}(x) := f(x,y)$ for all $x \in \cX$ is a bijection for each $y \in \cY$. (Note this means it must be the case $\vert \cZ \vert = \vert \cY \vert$.)
\end{definition}
\noindent As the definition highlights, a conditionally bijective function is really just a set of bijections on $\cX$ indexed by $y \in \cY$. For intuition, we provide two examples of a conditionally bijective function. 
\begin{example}[Examples of Conditionally Bijective Functions]\label{ex:cond-bij-func}
    For simplicity, we let $\cX = \{0,1,...,\vert \cZ \vert -1 \} = \cZ$, which can always be the case by relabeling the alphabet.
    \begin{itemize}
        \item Trivial Conditionally Bijective Functions: For all $y$, the value of $x$ does not change, i.e. $f_{y}(x) = x$. In other words, $f_{y}$ is always the identity function. We call this `trivial' as it doesn't do anything to the input $x$ depending on $y$, but is conditionally bijective.
        \item Shifting Conditionally Bijective Functions: For all $y$, $f_{y}(x) = x + y \mod \vert \cZ \vert$.  We call this `shifting' as it would be that $f_{y}$ simply shifts the value of $x$ by $y$ modulo $\vert \cZ \vert$.
    \end{itemize}
\end{example}

The claim will be that for a conditionally bijective function, guessing the value of $Z$ given the value of $Y$ is as hard as guessing the value of $X$ given the value of $Y$. This should not be surprising: a bijection is invertible, so if someone knows the conditionally bijective function and the value of $y$, that person knows the effective function $f_{y}$ that was applied to $x$. As such, if one can guess $z$, they can equally well guess $x$ as $x = f_{y}^{-1}(z)$. One might hope this extends to how well one can guess with partial information. In Lemma \ref{lem:cond-Y-func-entropy} we prove this intuition holds for general distributions as measured for a large class of entropies, including the min-entropy, which by applying \eqref{eq:guess-prob-and-min-ent} allows us to conclude the following.
\begin{lemma}\label{lem:cond-bij-p-guess-and-H-min}
    Let $f: \cX \times \cY \to \cZ$ be $Y$-conditionally bijective. Let $\rho_{ZYE} = \rho_{f(X,Y)YE}$. Then 
    \begin{align}
        p_{g}(Z \vert YE)_{\rho} = p_{g}(X \vert YE)_{\rho} \quad \text{and} \quad H_{\min}(Z \vert YE)_{\rho} = H_{\min}(X \vert YE)_{\rho} \ . 
    \end{align}
\end{lemma}
\noindent We remark this result only requires the structure of the function to be conditionally bijective. For this reason, from an entropic perspective, there is no difference between different conditionally bijective functions. In particular, recalling Example \ref{ex:cond-bij-func}, there is no difference between applying the trivial conditionally bijective function (which does nothing) and the shifting conditional bijective function (which changes the value a great deal depending on the value of $Y$) because the same access to $Y$ is given in both cases.

The remainder of the section states the more general form of Lemma \ref{lem:cond-bij-p-guess-and-H-min} and provides the proof.
\begin{lemma}\label{lem:cond-Y-func-entropy}
    Let $f: \cX \times \cY \to \cZ$ be $Y$-conditionally bijective. Let $0 \leq \rho_{ZYE} = \rho_{f(X,Y)YE}$ where $\Tr[\rho] \leq 1$, i.e. it may be a subnormalized state. Then for any conditional entropy $\mbb{H}(A\vert B)_{\rho \vert \sigma} \coloneq -\mbb{D}(\rho \Vert I_{A} \otimes \sigma_{B})$ or $\mbb{H}(A \vert B)_{\rho} \coloneq \sup_{\sigma_{B}} H(A \vert B)_{\rho \vert \sigma}$ defined via a divergence $\mbb{D}$ that satisfies the data-processing inequality and is unitarily invariant,
    \begin{align}
        \mbb{H}(Z|YE)_{\rho} = \mbb{H}(X|YE)_{\rho} \ .
    \end{align}
    Moreover, any smooth min-entropy defined using a distance measure that satisfies data processing,\footnote{See \eqref{eq:purified-smooth-distance-smoothed-min-entropy} as an example} the same holds, i.e. under said conditions, for all $\ve \in [0,1]$,
    \begin{align}
        H^{\ve}_{\min}(Z|YE)_{\rho} = H^{\ve}_{\min}(X|YE)_{\rho} \ . 
    \end{align}
\end{lemma}
\begin{proof}
    We first prove the claim for an entropy $\mbb{H}$ on normalized states, then explain how the proof changes if the states are subnormalized, and finally explain the claim for smooth min-entropy. First, we observe that the global state
    \begin{align}
        \rho_{XYZE} = \sum_{x,y} p(x,y) \dyad{x} \otimes \dyad{y} \otimes \dyad{f(x,y)}_{Z} \otimes \rho^{x,y}_{E}
    \end{align} 
    is defined such we obtain $\rho_{XYE}$ and $\rho_{f(X,Y)YE}$ simply as marginals of this global state. In particular, it is immediate from the partial trace that we have the marginals
    \begin{align}
        \rho_{XYE} &= \sum_{x,y} p(x,y) \dyad{x} \otimes \dyad{y} \otimes \rho^{x,y}_{E} \label{eq:XYE-marginal-for-bijection} \\
        \rho_{ZYE} &= \sum_{z,y} q(z,y) \dyad{z} \otimes \dyad{y} \otimes \rho^{z,y}_{E} \ ,
    \end{align}
    where $q_{ZY}$ is some joint distribution and $\{\rho^{z,y}_{E}\}_{z,y}$ is some set of conditional quantum states. We need to establish a nice expression for $\rho_{ZYE}$ that makes clear the form of $q_{ZY}$ and the conditional states. As $f$ is $Y$-conditionally bijective, the joint probability of $z$ and $y$ is $p(z,y) = p(\hat{x},y)$ where $\hat{x}$ is \textit{the unique} $x \in \cX$ such that $f_{y}(\hat{x}) = z$. Thus, $q(z,y) = p( f_{y}^{-1}(z), y)$ for each $z,y$. For the same reason $\rho^{z,y}_{E} = \rho^{f^{-1}_{y}(z),y}_{E}$. Therefore, we may write 
    \begin{align}
        \rho_{ZYE} 
        &= \sum_{z,y} q(f_{y}^{-1}(z),y) \dyad{z}_{Z} \otimes \dyad{y} \otimes \rho^{f_{y}^{-1}(z),y}_{E} \nonumber \\
        &=  \sum_{x,y} p(x,y) \dyad{f(x,y)}_{Z} \otimes \dyad{y} \otimes \rho^{x,y}_{E} \ . \label{eq:form-of-rhoZYE}
    \end{align}
    Note that this \textit{only} worked because we appealed to the bijectivity. In effect, this shows $\rho_{XYE}$ and $\rho_{ZYE}$ are unitarily equivalent, but we make this explicit. Consider the unitary 
    \begin{align}\label{eq:bijection-unitary}
        U_{XY \to ZY} = \sum_{x,y} \ket{f_{y}(x)}\bra{x} \otimes \dyad{y} 
    \end{align} where we stress that this is a unitary because $X \cong Z$ and $f$ is bijective. We show this formally at the end of the proof for completeness. Now a direct calculation using \eqref{eq:XYE-marginal-for-bijection} and \eqref{eq:form-of-rhoZYE} will verify
    $U\rho_{XYE}U^{\dagger} = \rho_{ZYE}$.
    
    Moreover, $U\left(I_{X} \otimes Q_{YE}\right)U^{\dagger} = I_{Z} \otimes Q_{YE}$ for all classical-quantum positive semidefinite operators $Q_{YE} \geq 0$. This is a rather direct calculation that we provide. In the case that $Q = 0$, this is trivial, so we can assume $Q \geq 0$ and $Q \neq 0$. As we can always let $Q_{YE} = a \sigma_{YE}$ where $a > 0$, it suffices to now consider the case $\sigma_{YE} \in \Density(YE)$, where $\sigma_{YE} = \sum_{y} r(y) \dyad{y} \otimes \sigma^{y}_{E}$. Then
    \begin{align}
        & U\left(I_{X} \otimes \sigma_{YE}\right)U^{\dagger} \notag \\
        =& \left(\sum_{x_{1},y_{1}} \ket{f_{y_{1}}(x_{1})}\bra{x_{1}} \otimes \dyad{y_{1}} \right) \cdot \left(\sum_{x} \dyad{x} \otimes \sum_{y} r(y) \otimes \sigma^{y}_{E} \otimes \dyad{y}\right)  \cdot \left(\sum_{x_{2},y_{2}} \ket{x_{2}}\bra{f_{y_{2}}(x_{2})} \otimes \dyad{y_{2}} \right) \notag \\
        =& \sum_{x,y} \ket{f_{y}(x)}\bra{f_{y}(x)} \otimes r(y) \sigma^{y}_{E} \otimes \dyad{y} \\ 
        =& \sum_{y} I_{Z} \otimes r(y) \sigma^{y}_{E} \otimes \dyad{y} \\
        =& I_{Z} \otimes \sigma_{YE} \ ,
    \end{align}
    where the first equality uses \eqref{eq:bijection-unitary}, the second equality collapses orthogonal registers, the third is because we may sum over $x$ independently over everything else and $f_{y}$ is a bijection for each $y$, and the last is just regrouping terms.

    Therefore, for any unitarily invariant divergence $\mbb{D}$ and classical-quantum $Q_{YE} \geq 0$,
    \begin{align}
        \mbb{D}(\rho_{XYE} \Vert I_{X} \otimes Q_{YE}) 
        = \mbb{D}(U\rho_{XYE}U^{\dagger} \Vert UI_{X} \otimes Q_{YE}U^{\dagger}) 
        = \mbb{D}(\rho_{ZYE} \Vert I_{Z} \otimes Q_{YE}) \ . 
    \end{align}
    For any $\mbb{H}(X\vert YE)_{\rho \vert \sigma} \coloneq -\mbb{D}(\rho \Vert I_{X} \otimes \sigma_{YE})$ where $\sigma_{YB}$ is classical-quantum, this establishes the proof. For $\mbb{H}(X \vert YE)_{\rho} \coloneq \sup_{\sigma_{YE} \in \Density(YE)} \mbb{H}(X \vert YE)_{\rho \vert \sigma}$, note that as $\mbb{D}$ satisfies data processing, one may restrict $\sigma_{YB}$ to classical-quantum states and then apply the previous case.
    
    Next, note the above proof extends to subnormalized $\rho$ as at no point did we rely on normalization in the above, so one may replace $\rho_{XYE}$ with $\alpha \rho_{XYE}$ where $\alpha \in [0,1]$.

    Lastly, we establish the property for smooth min-entropy independently due to the smoothing. Let $H^{\ve}_{\min}(Z|YE)_{\rho} = H_{\min}(Z|YE)_{\wt{\rho}}$. Then $\wt{\rho}_{ZYE} \in \mathcal{B}^{\ve}(\rho)$ where $\mathcal{B}^{\ve}(\rho) := \{\wt{\rho} \in \Density_{\leq}(ZYE) : \Delta(\wt{\rho},\rho) \leq \ve\}$ for distance measure $\Delta$ being unitarily invariant and satisfying the data processing inequality. It follows $U^{\dagger}\wt{\rho}_{ZYE}U \in \mathcal{B}^{\ve}(U^{\dagger} \rho_{ZYE} U) = \mathcal{B}^{\ve}(\rho_{XYE})$. Then, as min-entropy is defined via max-divergence \cite[Eqs.~5.19 and 5.24]{Tomamichel-Book}, which is unitarily invariant, by our previous account
    \begin{align}
        H^{\ve}_{\min}(Z|YE)_{\rho} &= H_{\min}(Z|YE)_{\wt{\rho}}
        = H_{\min}(X|YE)_{U^{\dagger}\wt{\rho}U}
        \leq H_{\min}^{\ve}(X|YE)_{\rho} \ ,
    \end{align}
    where the inequality is because smooth entropy is a maximization. The reverse direction follows an identical argument using the reversed unitary. This completes the main part of the proof.

    \textit{Unitarity of \eqref{eq:bijection-unitary}}. To see it is a unitary,
    \begin{align}
        UU^{\dagger} =&  \left(\sum_{x,y} \ket{f_{y}(x)}\bra{x} \otimes \dyad{y}\right) \cdot \left(\sum_{\ol{x},\ol{y}} \ket{\ol{x}}\bra{f_{\ol{y}}(\ol{x})} \otimes \dyad{\ol{y}}\right) \\
        =&  \sum_{x,y} \ket{f_{y}(x)}\bra{f_{y}(x)} \otimes \dyad{y} \\
        =& \sum_{y} I_{Z} \otimes \dyad{y} \\
        =& I_{Z} \otimes I_{Y} \ ,
    \end{align}
    where the first equality is collapsing orthogonal states and the second is by the bijectivity property of $f_{y}$. Similarly,
    \begin{align}
        U^{\dagger}U =&  \left(\sum_{\ol{x},\ol{y}} \ket{\ol{x}}\bra{f_{\ol{y}}(\ol{x})} \otimes \dyad{\ol{y}}\right) \left(\sum_{x,y} \ket{f_{y}(x)}\bra{x} \otimes \dyad{y}\right) \\
        =&  \sum_{x,\ol{x},y} \ket{\ol{x}}\bra{f_{y}(x)}\ket{f_{y}(\ol{x})}\bra{\ol{x}} \otimes \dyad{y} \\
        =& \sum_{x,y} \dyad{x} \otimes \dyad{y} \\
        =& I_{Z} \otimes I_{Y} \ ,
    \end{align}
    where the first equality collapses orthogonal states and the second uses that $f_{y}$ is a bijection so $f_{y}(x) = f_{y}(\ol{x})$ if and only if $x = \ol{x}$.
\end{proof}

\section{Bounds on Guessing Conditionally Bijective Functions over Multiaccess Networks}\label{app:bounds-on-guessing-cond-bij-over-MNs}
We now turn to what may be seen as our main technical lemmata for establishing Theorem \ref{thm:main-text-communication-advantage}. These are generic results on constraining how well the output of a conditionally bijective function can be guessed over specific types of multiaccess networks as a function of the signaling dimension. We believe these results are of interest in their own right, and thus we state them as Theorems \ref{thm:unassisted-QMN-bound} and \ref{thm:non-signaling-box-bound}. 

\subsection{Bounds for Unassisted Quantum Multiaccess Networks}\label{sec:bounds-for-unassisted-QMAN}
\begin{figure}[h]
    \centering
        \begin{tikzpicture}
            % x node set with absolute coordinates
            \node[terminal] (x) at (-1,1) {$\cX \ni x$};
            \node[terminal] (y) at (-1,-1) {$\cY \ni y$};    
            \node[qsource] (A) at (1.1,1) {$\rho^{x}_{C_{1}}$};
            \node[qsource] (B) at (1.1,-1) {$\sigma^{y}_{C_{2}}$};
            \node[meas_dev, minimum height = 1cm] (R) at (2.75, 0) {$\{\Pi_{z}\}_{z}$};
            \node[terminal] (z) at (4, 0) {$z$};
        
            \path (x) \cedge (A);
            \path (y) \cedge (B);
            \path (A) \qedge (R);
            \path (B) \qedge (R);
            \path (R) \cedge (z);
        \end{tikzpicture} 
    \caption{The network considered in Section \ref{sec:bounds-for-unassisted-QMAN}, an unassisted quantum multiaccess network with signalling dimension $\vec{d} = (\vert C_{1} \vert, \vert C_{2} \vert)$.}
    \label{fig:QMN-unassisted}
\end{figure}
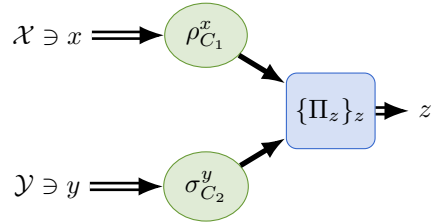

The following shows that a two-sender QMN without entanglement-assistance (depicted in Fig.~\ref{fig:QMN-unassisted}) is constrained in correctly guessing the output of a $Y$-conditionally bijective function $f: \cX \times \cY \to \cZ$ by the dimension of $\cX$ and the signaling dimension of sender $1$. Physically, the proof is a sequence of reductions where first we show the task is limited by Sender 1's ability to send $x \in \cX$ to the receiver over a quantum channel with output dimension $\vert C_{1} \vert$ (Lemma \ref{lem:no-ent-ass-reduc-to-point-to-point}), then we use the Frenkel-Weiner theorem (Proposition \ref{prop:Frenkel-Weiner}) to replace the quantum channel with a classical channel with the same signaling dimension and shared randomness (Corollary \ref{cor:reduc-point-to-point-to-cl-with-SR}), we then remove the shared randomness without loss of generality (Proposition \ref{prop:no-need-for-SR}), and finally we can use Item 5 of Proposition \ref{prop:guessing-prob-properties} to obtain a bound only in terms of the ability to guess the value of $x \in \cX$ and the size of the signaling dimension of party $1$, $\vert C_{1} \vert$.

\begin{theorem}\label{thm:unassisted-QMN-bound}
    Let $f:\cX \times \cY \to \cZ$ be $Y$-conditionally bijective and the two inputs be independent, i.e.~$p_{XY} = p_{X} \otimes q_{Y}$. Then for fixed Hilbert spaces $C_{1}$ and $C_{2}$,
    \begin{align}
         \sup_{\cE_{X \to C_{1}},\cF_{Y \to C_{2}}} p_{g}(Z \vert C_{1}C_{2})_{(\id_{Z} \otimes \cE \otimes \cF)(p_{ZXY})} \leq \vert C_{1} \vert  p_{g}(X)_{p} \ . 
    \end{align}
\end{theorem}
\begin{proof}
     We have the sequence of inequalities
    \begin{align}
        \sup_{\cE_{X \to C_{1}},\cF_{Y \to C_{2}}} p_{g}(Z \vert C_{1}C_{2})_{(\id_{Z} \otimes \cE \otimes \cF)(p_{ZXY})} &\leq \sup_{\cE_{\ol{X} \to C_{1}}} p_{g}(X \vert C_{1})_{(\id_{X} \otimes \cE)(\chi^{\vert p}_{X\ol{X}})} \\
        &\leq \sup_{P_{\widehat{X} \vert X} \in \mbf{C}^{SR}(\vert C_{1} \vert)} \sum_{x} p(x) P_{\widehat{X}\vert X}(x \vert x) \\
        & = \max_{h: \ol{X} \to \cC_{1}} p_{g}(X \vert C_{1})_{(\id_{X} \otimes h)(\chi^{\vert p}_{X\ol{X}})} \ 
    \end{align}
    where the first inequality is Lemma \ref{lem:no-ent-ass-reduc-to-point-to-point}, the second is Corollary \ref{cor:reduc-point-to-point-to-cl-with-SR}, and the final equality is Proposition \ref{prop:no-need-for-SR} combined with Remark \ref{rem:on-non-uniformity}. Moreover, as $C_{1}$ is classical, 
    \begin{align}
        \max_{h: \ol{X} \to \cC_{1}} p_{g}(X \vert C_{1})_{(\id_{X} \otimes h)(\chi^{\vert p}_{X\ol{X}})} \leq \vert C_{1} \vert p_{g}(X)_{\chi^{\vert p}} = \vert C_{1} \vert  p_{g}(X)_{p} \ ,
    \end{align}
    where the first inequality is Item 5 of Proposition \ref{prop:guessing-prob-properties} and the equality is simply the $X$ marginal of $\chi^{\vert p}$ is the distribution $p_{X}$.
\end{proof}

The rest of this subsection establishes the reductions used in the above theorem.
\begin{lemma}\label{lem:no-ent-ass-reduc-to-point-to-point}
    Let $f:\cX \times \cY \to \cZ$ be $Y$-conditionally bijective. Then for fixed Hilbert spaces $C_{1}$ and $C_{2}$,
    \begin{align}\label{eq:reduc-to-point-to-point}
        \sup_{\cE_{X \to C_{1}},\cF_{Y \to C_{2}}} p_{g}(Z \vert C_{1}C_{2})_{(\id_{Z} \otimes \cE \otimes \cF)(p_{ZXY})} \leq \sup_{\cE_{\ol{X} \to C_{1}}} p_{g}(X \vert C_{1})_{(\id_{X} \otimes \cE)(\chi^{\vert p}_{X\ol{X}})} \ ,
    \end{align}
    where $p_{ZXY} = \sum_{x,y} p_{X}(x)q_{Y}(y) \dyad{f(x,y)}_{Z} \otimes \dyad{x}_{X} \otimes \dyad{y}_{Y}$.
\end{lemma}
\begin{proof}
    First, for a fixed strategy defined by classical-quantum channels $\cE(\dyad{x}) = \rho^{x}_{C_{1}}$ and $\cF(\dyad{y}) = \sigma^{y}_{C_{2}}$, the total joint state is 
    \begin{align}
        \rho_{ZC_{1}C_{2}XY} = \sum_{x,y} p_{X}(x)q_{Y}(y)\dyad{f(x,y)}_{Z} \otimes \rho^{x}_{C_{1}} \otimes \sigma^{y}_{C_{2}} \otimes \dyad{x} \otimes \dyad{y} \ ,
    \end{align}
    where we used that $\cE,\cF$ are classical-to-quantum channels so that $\rho^{y}_{C_{2}} \coloneq \cF(\dyad{y})$ and similarly for $\rho^{x}_{C_{1}}$. 
    
    Our goal is to remove the $C_{2}$ register from the guessing probability. By data processing (Item 1 of Proposition \ref{prop:guessing-prob-properties}) with partial trace,
    \begin{align}
        p_{g}(Z \vert C_{1}C_{2})_{\rho} \leq p_{g}(Z \vert C_{1}C_{2}Y)_{\rho} \ ,
    \end{align}
    which physically would be like giving the receiver a copy of $Y$. We now will use a data processing argument to show with access to $Y$, the removal of the $C_{2}$ register does not change the guessing probability. This is because one can generate $C_{2}$ from $Y$. To see this, define the isometry that copies the $Y$ register, $V_{Y \to \ol{Y}Y} = \sum_{y} \ket{yy}\bra{y}$. Then,
    \begin{align}
        \cF_{\ol{Y} \to C_{2}}(V\rho_{ZC_{1}Y}V^{\dagger}) &= \sum_{x,y} p_{X}(x)q_{Y}(y)\dyad{f(x,y)}_{Z} \otimes \rho^{x}_{C_{1}} \otimes \cF(\dyad{y}_{\ol{Y}}) \otimes \dyad{y}_{Y} \\
        &= \sum_{x,y} p_{X}(x)q_{Y}(y)\dyad{f(x,y)}_{Z} \otimes \rho^{x}_{C_{1}} \otimes \sigma^{y}_{C_{2}} \otimes \dyad{y}_{Y} \\
        &=  \rho_{ZC_{1}C_{2}Y} \label{eq:reconstructing-C2-from-Y}
    \end{align}
    We then use this with the guessing probability: 
    \begin{align}
        p_{g}(Z \vert C_{1}C_{2}Y)_{\rho} \geq p_{g}(Z \vert C_{1}Y)_{\rho} = p_{g}(Z \vert C_{1}Y\ol{Y})_{V\rho V^{\dagger}} &\geq p_{g}(Z \vert C_{1}YC_{2})_{\cF_{\ol{Y} \to C_{2}}(V\rho V^{\dagger})} = p_{g}(Z \vert C_{1}C_{2}Y)_{\rho} \ ,
    \end{align}
    where the first inequality is data processing with partial trace, the first equality is isometric invariance on the conditioning system, the second inequality is data processing, and the final equality is re-ordering registers and \eqref{eq:reconstructing-C2-from-Y}. Thus, 
    \begin{align}
        p_{g}(Z \vert C_{1}C_{2}Y)_{\rho} = p_{g}(Z \vert C_{1} Y)_{\rho} \ . 
    \end{align} 
    
    As $f$ is $Y$-conditionally bijective, by Lemma \ref{lem:cond-bij-p-guess-and-H-min}, $p_{g}(Z \vert C_{1}Y) = p_{g}(X \vert C_{1}Y)$. As $\rho_{C_{1}XY} = \rho_{C_{1}X} \otimes q_{W_{2}}$, i.e. there is independence of $Y$ from the other registers, $p_{g}(X \vert C_{1}Y)_{\rho} = p_{g}(X \vert C_{1})_{\rho} = p_{g}(X \vert C_{1})_{(\id_{X} \otimes \cE_{\ol{X} \to C_{1}})(\chi^{\vert p}_{X\ol{X}})}$. Combining the above inequalities, we obtain
    \begin{align}
        p_{g}(Z \vert C_{1}C_{2})_{\rho} \leq p_{g}(X \vert C_{1})_{(\id_{X} \otimes \cE_{\ol{X} \to C_{1}})(\chi^{\vert p}_{X\ol{X}})} \ .
    \end{align}
    By optimizing over the choice of strategy, we conclude
    \begin{align}
        \sup_{\cE_{X \to C_{1}},\cF_{Y \to C_{2}}} p_{g}(Z \vert C_{1}C_{2})_{(\id_{Z} \otimes \cE \otimes \cF)(p_{ZXY})} \leq \sup_{\cE_{\ol{X} \to C_{1}}} p_{g}(X \vert C_{1})_{(\id_{X} \otimes \cE)(\chi^{\vert p}_{X\ol{X}})} \ .
    \end{align}
    This is the promised reduction to a point-to-point channel problem.
\end{proof}

%Explanation of why we need to do more work
The reason the above lemma is not sufficient to establish Theorem \ref{thm:unassisted-QMN-bound} is that the $C_{1}$ register is quantum and thus we cannot appeal to Item 5 of Proposition \ref{prop:guessing-prob-properties}. To resolve this, we use the Frenkel-Weiner theorem (Proposition \ref{prop:Frenkel-Weiner}).

% To resolve this, we introduce shared-randomness-assisted QMNs and CMNs and then make a further reduction to a classical channel using the Frenkel-Weiner theorem.

% \begin{definition}\label{def:SR-assisted-MNs}
%     Let $k \in \mbb{N}$ and $\vec{d} \in \mbb{N}^{\times k}$.
%     \begin{enumerate}
%         \item $\mbf{C}^{SR}(\vec{d})$ denotes the set of MNs where sender $i$ and the receiver are connected by $\Delta_{d_{i}}$, a noiseless classical channel of dimension $d_{i}$, and all the senders and the receiver share unbounded randomness.
%         \item $\mbf{Q}^{SR}(\vec{d})$ denotes the set of MNs where sender $i$ and the receiver are connected by $\id_{d_{i}}$, a noiseless quantum channel of dimension $d_{i}$, and all the senders and the receiver share unbounded randomness.
%     \end{enumerate}
% \end{definition}

% With this notation introduced, we formally state the Frenkel-Weiner theorem in our formalism.

% We now use the Frenkel-Weiner theorem.
\begin{corollary}\label{cor:reduc-point-to-point-to-cl-with-SR}
    For any $\chi^{\vert p}_{X\ol{X}} = \sum_{x} p_{X}(x) \dyad{x}^{\otimes 2}$, the strategy for decoding $X$ over all quantum channels with output dimension $\vert C \vert$ is upper bounded by the optimal strategy using a classical channel with the same output dimension and shared randomness. Formally,
    \begin{align}\label{eq:guessing-prob-Frenkel-Wiener}
        \sup_{\cE_{\ol{X} \to C}} p_{g}(X \vert C)_{(\id_{X} \otimes \cE)(\chi^{\vert p}_{X\ol{X}})} \leq \sup_{P_{\widehat{X} \vert X} \in \mbf{C}^{SR}(\vert C \vert)} \sum_{x} p(x) P_{\widehat{X}\vert X}(x \vert x) \ . 
    \end{align}
\end{corollary}
\begin{proof}
    \sloppy We begin by re-expressing the guessing probability in the LHS of \eqref{eq:guessing-prob-Frenkel-Wiener}. For a quantum channel $\cE_{\ol{X} \to C}$, we have
    \begin{align}
        p_{g}(X \vert C)_{(\id_{X} \otimes \cE)(\chi^{\vert p}_{X\ol{X}})} = \max_{\{\Gamma_{x}\}} \sum_{x} p(x) \Tr[\Gamma_{x}\cE(\dyad{x})] \ .
    \end{align}
    We can define the classical channel $P_{\widehat{X} \vert X}(\widehat{x} \vert x) = \Tr[\Gamma_{\widehat{x}}\cE(\dyad{x})] = \Tr[\Gamma_{\widehat{x}} (\id_{\vert C \vert}(\cE(\dyad{x})))]$. The final expression highlights this conditional distribution can be implemented between a sender and receiver using a $\vert C \vert$-dimensional identity channel. Thus, by Proposition \ref{prop:Frenkel-Weiner}, this classical channel can be implemented between the sender and receiver using a noiseless classical channel of dimension signaling dimension $\vert C \vert$, $\Delta_{\vert C \vert}$, and shared randomness. This completes the proof.
\end{proof}

\subsection{Bounds from Not Signaling in One Direction}\label{sec:bounds-on-not-signaling-in-one-direc}
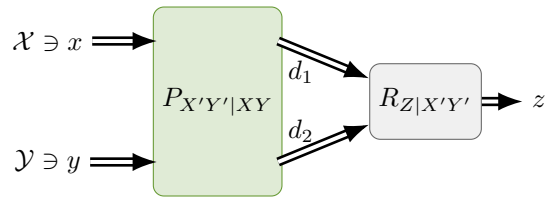
\begin{figure}[h]
    \centering
    \begin{tikzpicture} %NSFCMN
            \node[terminal] (x) at (-1,0.8) {$\cX \ni x$};
            \node[terminal] (y) at (-1,-0.8) {$\cY \ni y$};    
            \node[prep_dev, minimum height = 2.5cm] (AB) at (1.25,0) {$P_{X'Y' \vert XY}$}; 
            \node[dev, minimum height = 1cm] (R) at (4, 0) {$R_{Z \vert X'Y'}$};
            \node[terminal] (z) at (5.5, 0) {$z$};

            \path (x) \cedge (0.45,0.8) ;
            \path (y) \cedge (0.45,-0.8) ;
            \path (2.05,0.8) \cedge node[below,pos=0.25] {$d_{1}$} (R);
            \path (2.05,-0.8) \cedge node[above,pos=0.25] {$d_{2}$} (R);
            \path (R) \cedge (z);
        \end{tikzpicture}
    \caption{The network considered in Section \ref{sec:bounds-on-not-signaling-in-one-direc}, a fully classical bipartite channel that does not signal from the $X$ system to the $Y$ system.}
    \label{fig:FCNSC}
\end{figure}

The following theorem, Theorem \ref{thm:non-signaling-box-bound}, shows the same sort of bounds as in Theorem \ref{thm:unassisted-QMN-bound} hold in the 2-sender case if the 2 senders share a classical-input, classical-output channel that does not signal in one direction (see Fig.~\ref{fig:FCNSC}). To prove this, we need the following proposition that follows from the relation between a channel being signaling and semi-localizable (Proposition \ref{prop:NS-equiv-semilocalizable}).
\begin{figure}[h]
    \centering
    \begin{tikzpicture} %decomposition
            \node[terminal] (A) at (0,1.25) {$A$};
            \node[terminal] (Y) at (0,-1.25) {$Y$}; \node[prep_dev, minimum height = 3cm] (AB) at (2,0) {$\cE_{AY \to A'Y'}$}; 
            \node[terminal] (Ap) at (3.75,1.25) {$A'$};
            \node[terminal] (Yp) at (3.75,-1.25) {$Y'$};
            \node (equals) at (4.5,0) {$=$};

            \node[terminal] (A2) at (5.25,1.25) {$A$};
            \node[terminal] (Y2) at (5.25,-1.25) {$Y$}; 
            \node[dev, minimum height = 1.25cm] (Tbox) at (7.2,-1.25) {$T_{Y \to Y\ol{Y}'Y'}$};
            \node[qsource] (rhoM) at (10,0.6) {$\rho^{y' \vert y}_{M}$};
            \node[proc_dev, minimum height = 1.5cm] (proc_chan) at (12,1.15) {$\cF_{AM \to A'}$};
            \node[terminal] (Yp2) at (14,-1.25) {$Y'$};
            \node[terminal] (Ap2) at (14,1.25) {$A'$};

            \path (A) \qedge (1.1,1.25); 
            \path (Y) \cedge (1.1,-1.25);
            \path (2.9,1.25) \qedge (3.5,1.25);
            \path (2.9,-1.25) \cedge (3.5,-1.25);
            \path (Y2) \cedge (6.25,-1.25);
            \path (8.15,-0.78) \cedge node[midway,below] {$Y\ol{Y}'$} (rhoM);
            \path (8.15,-1.25) \cedge (Yp2);
            \path (A2) \qedge (11.2,1.25);
            \path (rhoM) \qedge (proc_chan);
            \path (12.8,1.25) \qedge (Ap2);

             \draw[dashed, gray] (5,0) -- (14,0);
        \end{tikzpicture}
    \caption{The channel decomposition of $\cE_{AY \to A'Y'}$ with non-signaling restriction $A \not \to Y$ proven in Proposition \ref{prop:copying-localizable}. The dotted grey line highlights that only classical communication needs to be transmitted.}
    \label{fig:channel-decomp}
\end{figure}
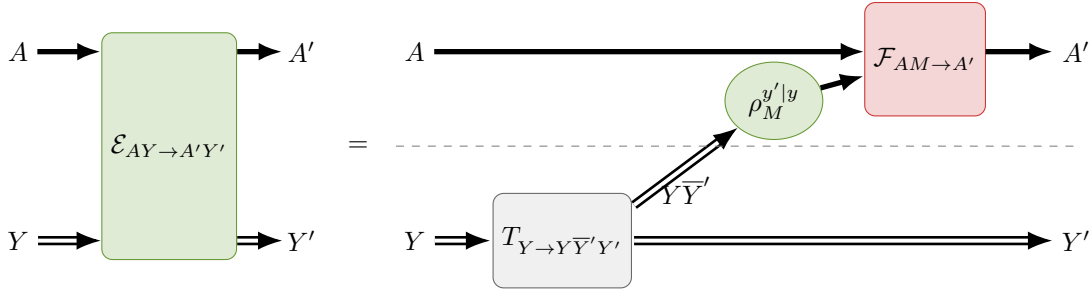
\begin{proposition}\label{prop:copying-localizable}
    Let $\cE_{AY \to A'Y'}$ be a bipartite channel where $Y$ and $Y'$ are classical registers not necessarily of the same dimension. If $\cE$ is $A \not \to Y$ signaling, then there exists a conditional distribution $p_{Y' \vert Y}$, a set of quantum states $\{\rho^{y' \vert y}_{M}\}_{y',y}$, and a quantum channel $\cF_{AM \to A'}$ such that 
    \begin{align}
        \cE = (\cF_{AM \to A'} \otimes \id_{Y'}) \circ (\id_{A} \otimes \cP_{Y\ol{Y}' \to M} \otimes \id_{Y'}) \circ (\id_{A} \otimes T_{Y \to Y\ol{Y}'Y'}) \ ,
    \end{align}
    where $\cP_{Y\ol{Y}' \to M}$ is the classical-to-quantum channel $\cP_{Y\ol{Y}' \to M}(K) = \sum_{y,y'} \bra{y}\bra{y'}K\ket{y}\ket{y'} \rho^{y' \vert y}_{M}$
    and $T_{Y \to Y\ol{Y}'Y'}$ is the classical channel
    \begin{align}\label{eq:T-channel-def}
        T_{Y \to Y\ol{Y}'Y'}(K) = \sum_{y,y'} p(y' \vert y)(\dyad{y}_{Y} \otimes \dyad{y'}_{\ol{Y}'} \otimes \dyad{y'}_{Y'})K(\dyad{y}_{Y} \otimes \dyad{y'}_{\ol{Y}'} \otimes \dyad{y'}_{Y'})^{\dagger} \ . 
    \end{align}
    In particular, this shows $\cE$ can be implemented solely communicating classical information from the $Y$ system to the $A'$ system by forwarding copies of the $Y$ input and $Y'$ output (See Fig.~\ref{fig:channel-decomp}).
\end{proposition}
\begin{proof}
    By Proposition \ref{prop:NS-equiv-semilocalizable}, $\cE_{AY \to A'Y'} = (\id_{Y'} \otimes \cF^{1}_{M\to A'}) \circ (\cF^{0}_{Y \to Y'M} \otimes \id_{A})$ where $M \leq \vert \cY \vert \cdot \vert \cY' \vert$. Moreover, as $Y$ and $Y'$ are classical, we can completely dephase before and after the map on these registers, i.e. $\cE_{AY \to A'Y'} = \Delta_{Y'} \circ \cE \circ \Delta_{Y}$. Therefore,
    \begin{align}
        \cE_{AY \to A'Y'} = [\cF^{1}_{AM\to A'} \otimes \id_{Y'} ] \circ [\id_{A} \otimes (\Delta_{Y'} \circ \cF^{0}_{Y \to MY'} \circ \Delta_{Y})] \ . 
    \end{align}
    By Proposition \ref{prop:dephase-input-makes-c-to-q-channel}, $\cF^{0}_{Y \to MY'} \circ \Delta_{Y}$ is itself a classical-to-quantum channel, so there exists $\{\rho^{y}_{MY'}\}_{y \in \cY}$ such that $(\cF^{0}_{Y \to MY'} \circ \Delta_{Y})(\dyad{y}) = \rho^{y}_{MY'}$ for all $y \in \cY$. As dephasing the $Y'$ system of $\rho^{y}_{MY'}$ results in a classical-quantum system, there exists quantum states $\{\rho^{y' \vert y}_{M}\}_{y',y}$ and conditional distribution $p_{Y' \vert Y}$ such that $(\id_{M} \otimes \Delta_{Y'})(\rho^{y}_{MY'}) = \sum_{y' \in \cY'} p_{Y' \vert Y}(y' \vert y)\dyad{y'}_{Y'} \otimes \rho^{y' \vert y}_{M}$ for all $y \in \cY$. In total, this allows us to conclude the classical-to-classical-quantum channel $\widehat{\cF}^{0} \coloneq \Delta_{Y'} \circ \cF^{0}_{Y \to MY'} \circ \Delta_{Y}$ is defined by
    \begin{align}
        \widehat{\cF}^{0}_{Y \to MY'}(\dyad{y})  = \sum_{y' \in \cY'} p(y' \vert y)\dyad{y'} \otimes \rho_{M}^{y' \vert y} \ .
    \end{align}
    It follows $\cE = \cF^{1} \circ \widehat{\cF}^{0}$. Using the conditional distribution from $\widehat{\cF}^{0}$, define the classical channel $T_{Y \to YY'\ol{Y}'}(\dyad{y}) = \sum_{y} p(y' \vert y) \dyad{y}_{Y} \otimes \dyad{y'}_{Y'} \otimes \dyad{y'}_{\ol{Y}'}$. Using the states $\rho^{y' \vert y}$ from $\widehat{\cF}^{0}$, define the classical-to-quantum channel $\cP(\dyad{y} \otimes \dyad{y'}) = \rho^{y' \vert y}_{M}$. By direct calculation, $\cE = \cF^{1} \circ \cP_{\ol{Y}'Y \to M} \circ T_{Y \to MY'}$. Defining $\cF \coloneq \cF^{1}$ completes the proof.
\end{proof}

We now use the above to establish our limit on computing a conditionally bijective function with $2$ senders who share a classical-to-classical box that cannot signal in one direction.
\begin{theorem}\label{thm:non-signaling-box-bound}
    Let $f:\cX \times \cY \to \cZ$ be a $Y$-conditionally bijective function. Let $\cE_{XY \to X'Y'}$ be non-signaling from $X$ to $Y$, i.e. $X \not \to Y$. Let the inputs on $X$ be independent of the inputs on $Y$, i.e.~$p_{XY} = p_{X} \otimes q_{Y}$. Then
    \begin{align}
        p_{g}(Z \vert X'Y')_{(\id_{Z} \otimes \cE)(p_{ZXY})} \leq \vert X' \vert p_{g}(X) \ . 
    \end{align}
\end{theorem}
\begin{proof}
    First, note that as the non-signaling box is fully classical, it may be written as a conditional distribution $E_{X'Y' \vert XY}$. This allows us to write a joint distribution over the inputs and outputs:
    \begin{align}
        p_{ZXYX'Y'} &= (\id_{Z} \otimes E_{X'Y \vert XY})(p_{ZXY}) \ , \label{eq:joint-dist-of-NS-box-in-and-out}
    \end{align}
    where $p_{ZXY} = \sum_{x,y} p_{X}(x)p_{Y}(y) \dyad{f(x,y)}_{Z} \otimes \dyad{x}_{X} \otimes \dyad{y}_{Y}$. Then using data processing of partial trace and that $f$ is $Y$-conditionally bijective,
    \begin{align}
        p_{g}(Z\vert X'Y')_{p} \leq p_{g}(Z\vert X'Y'Y)_{p} = p_{g}(X \vert X'Y'Y)_{p} \ .
    \end{align}
    Our goal is to upper bound $p_{g}(X \vert X'Y'Y)_{p}$. Given the relation between guessing probability and min-entropy in \eqref{eq:guess-prob-and-min-ent}, we do this by lower bounding $H_{\min}(X \vert X'Y'Y)_{p}$.

    First, we remove the classical information $X'$ using the chain rule for classical information \cite[Lemma 6.18]{Tomamichel-Book}:
    \begin{align}
        H_{\min}(X \vert X'Y'Y) \geq H_{\min}(X \vert Y'Y) - \log\vert X' \vert \ .
    \end{align}
    Next, we want to remove $Y'$. By Proposition \ref{prop:copying-localizable}, $E_{X'Y' \vert XY} = \cF_{XM \to X'} \circ \cP_{\ol{Y}'Y \to M} \circ T_{Y \to Y\ol{Y}'Y'}$ where $T$ copies the input $Y$ into output $Y$ (see \eqref{eq:T-channel-def}). It follows that 
    \begin{align}
        p_{XYY'} = (\id_{X} \otimes \Tr_{\ol{Y}'} \circ T_{Y \to Y \, \ol{Y}' \, Y'})(p_{XY}) \eqqcolon (\id_{X} \otimes \widehat{T}_{Y \to YY'})(p_{XY}) \ .
    \end{align}
    As such, by data processing,
    \begin{align}\label{eq:C2-carries-no-info}
        H_{\min}(X \vert Y'Y)_{p} \leq H_{\min}(X \vert Y)_{p} \leq H_{\min}(X \vert Y'Y)_{(\id_{X} \otimes \widehat{T})(p_{XY})} = H_{\min}(X\vert Y'Y)_{p}
    \end{align}
    Thus, $H_{\min}(X \vert Y'Y)_p = H_{\min}(X \vert Y)$. As $p_{XY} = p_{X} \otimes q_{Y}$, by Item 3 (or Item 2) of Proposition \ref{prop:guessing-prob-properties},
    \begin{align}
         H_{\min}(X \vert Y)_{p} = H_{\min}(X)_{p} \ .
    \end{align}
    Combining all the steps, we have $H_{\min}(X \vert X'Y'Y) \geq H_{\min}(X)_{p} - \log \vert X' \vert$. Combining this with  \eqref{eq:guess-prob-and-min-ent},
    \begin{align}
        p_{g}(Z \vert X'Y') \leq \exp(-H_{\min}(X\vert X'Y'Y)) \leq \exp(-H_{\min}(X) + \log\vert X' \vert) = \vert X' \vert p_{g}(X) \ .
    \end{align}
    This completes the proof.
\end{proof}
\begin{remark}
    Note that we used the box the senders share only has classical inputs and outputs. First, the existence of a channel $\widehat{T}$ that generates $Y'$ from $Y$ while preserving $Y$ used Proposition \ref{prop:copying-localizable}, which needs $Y$ and $Y'$ to be classical. Second, we used that $X'$ is classical to have the bound $\log \vert X' \vert$. Finally, $X$ had to be classical to talk of a classical function $f$.
\end{remark}

We provide the following for completeness.
\begin{proposition}\label{prop:dephase-input-makes-c-to-q-channel}
    Let $d \in \mbb{N}$ and $\cE_{\mbb{C}^{d} \to B}$ be a quantum channel. Then $\cE \circ \Delta_{d}$ is a classical-to-quantum (i.e. state preparation) channel, that is there exist density matrices $\{\rho^{y}_{B}\}_{y}$ such that $(\cE \circ \Delta_{d})(X) = \sum_{y} \sqrt{\rho^{y}}\bra{y}K\ket{y}\sqrt{\rho^{y}}$ for all $K \in \Lin(\mbb{C}^{d})$.
\end{proposition}
\begin{proof}
    We will prove this using the Kraus operators of a quantum channel. A set of Kraus operators for $\Delta_{d}$ are $\{\dyad{i}\}_{i \in [d]}$ for the chosen classical basis. Let $\{A_{k}\}_{k}$ be a set of Kraus operators for $\cE$. Then for any operator $K \in \Lin(\mbb{C}^{d})$, 
    \begin{align}
        (\cE \circ \Delta_{d})(K) = \sum_{k} A_{k} \left(\sum_{y} \dyad{y}X\dyad{y}\right) A_{k}^{\dagger} &= \sum_{y} \left[ \bra{y}K\ket{y} \left( \sum_{k} A_{k}\dyad{y}A_{k}^{\dagger} \right) \right] \\
        &= \sum_{y} \bra{y}K\ket{y} \cE(\dyad{y}) \\
        &\coloneq \sum_{y} \bra{y}K\ket{y}\rho^{y}_{B} \\
        &= \sum_{y} \sqrt{\rho^{y}_{B}}\bra{y}K\ket{y}\sqrt{\rho^{y}}_{B} \ ,
    \end{align}
    where the second equality is linearity, the third is that $\{A_{k}\}$ are the Kraus operators of $\cE$, the definition of the $\rho^{y}_{B}$ use that $\cE$ is a CPTP map, so $\rho^{y}_{B}$ is always a quantum state, and the final equality is linearity again.
\end{proof}

\section{The Communication Advantage of Dense Network Coding}\label{app:establishing-communication-advantage}
In this section we establish the generalized form of Theorem \ref{thm:main-text-communication-advantage}. This section is split into four pieces: introducing `doubly-conditionally bijective functions,' using the results of Appendix \ref{app:bounds-on-guessing-cond-bij-over-MNs} to place limits on FCNSMNs and EAQMNs (recall Definition \ref{def:main-families-of-MNs}), showing these limits on FCNSMNs and EAQMNs can be exactly achieved using simple CMNs when computing ditwise addition, and finally combining these results to obtain our main theorem (Theorem \ref{thm:one-shot-info-advantage}) and the asymptotic version for communication complexity theorists (Theorem \ref{thm:asymptotic-comm-complexity-adv}).

\subsection{Doubly-Conditionally Bijective Functions}\label{sec:DCB-functions}
Here we define and give a brief structural account of doubly-conditionally bijective functions which will be needed to understand results in subsequent parts of this section.
\begin{definition}\label{def:conditionally-bijective}
    Let $f: \cX \times \cY \to \cZ$ be a function. We say $f$ is `doubly-conditionally bijective' (DCB) if $f$ is both $\cX$- and $Y$-conditionally bijective. That is,
    \begin{enumerate}
        \item for all $x \in \cX$, the functions $f_{x}:\cY \to \cZ$ defined via $f_{x}(y) = f(x,y)$ are bijections and
        \item for all $y \in \cY$, the functions $f_{y}:\cX \to \cZ$ defined via $f_{y}(x) = f(x,y)$ are bijections. 
    \end{enumerate}
    Noting a DCB function requires $\vert \cX \vert = \vert \cY \vert = \vert \cZ \vert$, we say a DCB function is of order $n$ if $\vert \cX \vert = n$.
\end{definition}
 For motivation of the subsequent structural claims, we begin with examples of DCB functions that include those we showed can be computed using dense network coding in Section \ref{app:dense-network-coding}.
\begin{example}[Examples of Doubly-Conditionally Bijective Functions] \label{ex:DCB-functions} ~
    \begin{itemize}
        \item Addition modulo $d$: $\oplus_{d}: \mbb{Z}_{d} \times \mbb{Z}_{d} \to \mbb{Z}_{d}$ defined by $x \oplus_{d} y = (x+y) \mod d$. 
        \item Bitwise XOR: $\oplus^{n}_{2}: \mbb{Z}_{2}^{n} \times \mbb{Z}_{2}^{n} \to \mbb{Z}_{2}^{n}$ defined by $x \oplus^{n} y = (x_{1} \oplus_{2} y_{1}, x_{2} \oplus_{2} y_{2}, ... , x_{n} \oplus_{2} y_{n})$.
        \item Ditwise Addition modulo $d$: $\oplus^{2}_{d}: \mbb{Z}_{d}^{2} \times \mbb{Z}_{d}^{2} \to \mbb{Z}_{d}^{2}$ defined by $(x_{1},x_{2}) \oplus^{2}_{d} (y_{1},y_{2}) = (x_{1} \oplus_{d} y_{1}, x_{2} \oplus_{d} y_{2})$.
    \end{itemize}
\end{example}

Clearly all of the above examples of DCB functions are in fact examples of group operations. To understand how DCB functions differ from group operations of TNC groups, we now show that DCB functions are always equivalent to computing multiplication of a \textit{quasi}group (Proposition \ref{prop:correspondence-between-DCB-and-quasi-group}).\footnote{Recall that a quasigroup $(Q,\cdot)$ is a set $Q$ with a binary multiplication operation $\cdot$ such that for every pair of elements $a,b \in Q$ there exist unique $x,y \in Q$ such that $a \cdot x = b$ and $y \cdot a =b$. This is sometimes called the Latin square property.} To show this, we provide a few definitions.
\begin{definition}
    The \textit{multiplication table} of a two-variable function $f: \cX \times \cY \to \cZ$ is a $\vert \cX \vert \times \vert \cY \vert$ array where entry $(x,y)$ takes the value $f(x,y)$.
\end{definition}
\begin{definition}
    A Latin square of order $n$ is an $n \times n$ array in which each cell contains a single symbol from a set of $n$ distinct elements such that each symbol occurs exactly once in each row and each column. (See Table \ref{tab:Latin-Square} for an example.)
\end{definition}

\begin{table}
    \centering
    \begin{tabular}{c c c c c}
         0 & 1 & 2 & 3 & 4  \\ 
         1 & 0 & 3 & 4 & 2 \\
         2 & 3 & 4 & 0 & 1 \\
         3 & 4 & 1 & 2 & 0 \\
         4 & 2 & 0 & 1 & 3
    \end{tabular}
    \caption{Example of a Latin square. Each entry represents the value of applying a binary operation $\cdot$ to the set $\{0,1,2,3,4\}$. This Latin square does not correspond to multiplication table of a group as it implies the binary operation $\cdot$ is not associative as may be verified by direct calculation.}
    \label{tab:Latin-Square}
\end{table}

\begin{proposition}
    The multiplication table of a DCB function is a Latin square and a Latin square defines the multiplication table of a DCB.
\end{proposition}
\begin{proof}
    We begin with the first claim. Let $f$ be a DCB of order $n$. Without loss of generality, represent it as $f: [n] \times [n] \to [n]$. A row of the multiplication table is $[f_{x}(0), f_{x}(1), ..., f_{x}(n-1)]$. As $f_{x}$ is a bijection, this is a permutation of the elements of $[n]$ and thus each element appears exactly once. By the same argument using that $f_{y}$ is a bijection, in each column of the multiplication table, each element appears exactly once. Thus by definition, the multiplication table is a Latin square of order $n$.

    We now prove the second claim. Consider a Latin square $L$. Define $f(x,y) \coloneq L_{x,y}$ for all $x,y \in [n]$. As each row of $L$ contains every element of $n$ exactly once, this implies $f_{x}$ is a bijection for every $x$. As each column of $L$ contains every element of $n$ exactly once, this implies $f_{y}$ is a bijection for every $y$. Thus $f$ is a DCB function by definition.
\end{proof}

Combining the above with the fact a Latin square is the multiplication table (Cayley table) of a quasigroup on $n$ elements \cite[Theorem 1.11]{dinitz2007handbook} allows us to conclude the following.
\begin{proposition}\label{prop:correspondence-between-DCB-and-quasi-group}
    The multiplication table of a DCB function is the multiplication table of a quasigroup and vice-versa.
\end{proposition}
\noindent An example of a DCB function that is computing a quasigroup rather than a group is specified by the Latin square (equivalently multiplication table) given in Table \ref{tab:Latin-Square}. The reason it is not a group operation is because the Latin square implies the binary operation is not associative.

\subsection{Bounds without Entanglement-Assistance or Quantum Communication}
We now bound the ability to compute any DCB function using any resources weaker than entanglement-assistance and quantum communication in terms of the signaling dimension. To do this, we focus on $\mbf{C}^{N}(d_{1},d_{2})$ and $\mbf{Q}(d_{1},d_{2})$, which will allow us to preclude achievability without both entanglement-assistance and quantum communication, because these cases limit $\mbf{C}^{E}(d_{1},d_{2})$, $\mbf{C}^{SR}(d_{1},d_{2})$, and $\mbf{C}(d_{1},d_{2})$ by Propositions \ref{prop:MN-containments} and \ref{prop:no-need-for-SR}.

\begin{proposition}\label{prop:impossibility-of-weaker-resource-MNs}
    Let $f:\cX \times \cY \to \cZ$ be doubly-conditionally bijective. Let $\vert \cX \vert =d$ and $d_{1}, d_{2} \in \mbb{N}$, then 
    \begin{align}
        \max\{P_{S}^{f}(\mbf{Q}(d_{1},d_{2})), P_{S}^{f}(\mbf{C}^{N}(d_{1},d_{2}))\} \leq \frac{\min\{d_{1},d_{2}\}}{d} \ . 
    \end{align}
\end{proposition}
\begin{proof}
    If the senders can send quantum signals but share no entanglement, they are in the model that Theorem \ref{thm:unassisted-QMN-bound} considers where $\vert C_{1} \vert = d_{1}$, $\vert C_{2} \vert = d_{2}$. As $f$ is $\cX$- and $Y$-conditionally bijective, we can apply Theorem \ref{thm:unassisted-QMN-bound} to each input variable, and thus minimize over the parties' dimension. Thus,
    \begin{align}
        P_{S}^{f}(\mbf{Q}(d_{1},d_{2})) \leq  \frac{\min\{d_{1},d_{2}\}}{d} \ .
    \end{align}
    
    If the senders may pre-process using a classical-input, classical-output non-signaling box (i.e. does not signal in either direction), then the composition of this box with their noiseless communication channels is a non-signaling box of input dimensions $\vert \cX \vert $ and $\vert \cY \vert$ and output dimensions of $d_{1}$ and $d_{2}$. This means that we are in the scenario of Theorem \ref{thm:non-signaling-box-bound} where both $X \not \to Y$ and $Y \not \to X$ are satisfied. As $f$ is doubly-conditionally bijective, we can minimize over the choice of party. Thus,
    \begin{align}
        P_{S}^{f}(\mbf{C}^{N}(d_{1},d_{2})) \leq \frac{\min\{d_{1},d_{2}\}}{d} \ . 
    \end{align}
\end{proof}

\subsection{An Achievable Strategy for \texorpdfstring{$\oplus_{d}$}{} using Classical Multiaccess Networks}
We now show the upper bounds in Proposition \ref{prop:impossibility-of-weaker-resource-MNs} cannot be generically improved as there exists a classical strategy for computing ditwise addition that matches the upper bounds whenever $d_{1}$ and $d_{2}$ both divide $d$. 
\begin{proposition}\label{prop:CMN-lower-bounds-for-xor-d} Let $d \coloneq \vert \cX \vert$. Let $d_{1}, d_{2} \in \mbb{N}$ such that $\max\{d_{1},d_{2}\} \leq d$. Let $q_{1}$ and $r_{1}$ (resp.~$q_{2}$ and $r_{2}$) be the quotient and remainder of dividing $d$ by $d_{1}$ (resp.~$d_{2}$). Then
    \begin{equation}
    \begin{aligned}
         P_{S}^{\oplus_{d}}(\mbf{C}(d_{1},d_{2})) &\geq \frac{1}{d^{2}} \Bigg[ (r_{1}r_{2} + (q_{1}-r_{1})(q_{2}-r_{2}))\min\{q_{1},q_{2}\} + r_{1}r_{2} \\
        & \hspace{2cm} + (q_{1}-r_{1})r_{2}\min\{q_{1},q_{2}+1\} + r_{2}(q_{2}-r_{2})\min\{q_{1}+1,q_{2}\} \Bigg] \ . 
    \end{aligned}
    \end{equation}
In particular, if $d_{1}$ and $d_{2}$ both divide $d$, i.e. there are no remainders,
\begin{align}\label{eq:CMN-lower-bounds-for-xor-d-when-divide}
    P_{S}^{\oplus_{d}}(\mbf{C}(d_{1},d_{2})) &\geq \frac{d_{1}d_{2}\min\{q_{1},q_{2}\}}{d^{2}} = \frac{\min\{d_{1},d_{2}\}}{d} \ , 
\end{align}
which matches the upper bound of Proposition \ref{prop:impossibility-of-weaker-resource-MNs}.
Furthermore, if $d = d_{1}^{2} = d_{2}^{2}$,
\begin{align}\label{eq:CMN-lower-bounds-for-xor-d-when-divide-and-squares}
    P_{S}^{\oplus_{d}}(\mbf{C}(d_{1},d_{2})) &\geq \frac{1}{\sqrt{d}} \ .
\end{align}
\end{proposition} 
\begin{proof}
    We first describe the strategy and then bound its success probability.
    
    (\textit{Strategy})  Let Alice partition $[d]$ into $d_{1}$ sets of the form $S^{A}_{i} = \{(iq_{1}, iq_{1}+1...,iq_{1}+(q_{1}-1) + \mbf{1}\{i \leq r_{1}\}\}$ for $i = [d_{1}]$. Alice then sends $i$ to the receiver when she sees an element in $S^{A}_{i}$, which defines her encoder $f_{A}$. Bob does similarly by defining the $d_{2}$ sets $\{S^{B}_{j}\}_{j \in [d_{2}]}$ similarly and sending $j$ when $y \in S^{B}_{j}$, which defines his encoder $f_{B}$. When the receiver receives $(i,j)$, it guesses the value $iq_{1} \oplus_{d} jq_{2}$, which defines the decoder.
    
    (\textit{Lower Bound}) First we re-express the success probability of this strategy in terms of the cardinality of specific sets:
    \begin{align}
        P_{S} &= \sum_{x,y} \Pr[x,y]\mbf{1}\{f(x,y)=f_{A}(x)q_{1} \oplus_{d} f_{B}(y)q_{2}\} \\
        &= \frac{1}{d^{2}}\sum_{x,y}\mbf{1}\{f(x,y)=f_{A}(x)q_{1} \oplus_{d} f_{B}(y)q_{2}\} \\
        &= \frac{1}{d^{2}} \sum_{i \in [d_{1}], j \in [d_{2}]} \sum_{x \in S^{A}_{i}, y \in S_{j}^{B}}\mbf{1}\{f(x,y)=iq_{1} \oplus_{d} jq_{2}\} \\
        &= \frac{1}{d^{2}} \sum_{i \in [d_{1}], j \in [d_{2}]} \vert\{(x,y) \in S^{A}_{i} \times S_{j}^{B} :  f(x,y) = iq_{1} \oplus_{d} jq_{2} \}  \vert \label{eq:specific-xor-d-guessing-strategy-step-1}
    \end{align}
    where the first equality the definitions of the encoders/decoders, the second is the uniform randomness of the inputs, the third is that the $S^{A}_{i}$ (resp.~$S^{B}_{j}$) sets partition $\cX$ (resp.~$\cY$), and the final equality is by definition of the set and the indicator function. 
    
    We now need to calculate the size of the sets in the final equality. To that end, note that for $x,y,x',y' \in [d]$, $x \oplus_{d} y = x' \oplus_{d} y'$ if and only if for some $t \in [d]$, $x' = x - t \mod d$ and $y' = y - t \mod d$, i.e. they are shifted by the same value. It follows  
    \begin{align}
        \vert\{(x,y) \in S^{A}_{i} \times S_{j}^{B} :  f(x,y) = iq_{1} \oplus_{d} jq_{2} \}  \vert &= \vert \{t \in [d]: iq_{1}+t \in S_{i}, \, , jq_{2}-t \in S_{j} \} \vert \\
        &= \begin{cases} \min\{q_{1},q_{2}\} & i > r_{1} \, , j > r_{2} \\ 
        \min\{q_{1}+1,q_{2}\} & i \leq r_{1} , j > r_{2} \\
        \min\{q_{1},q_{2}+1\} & i > r_{1} , j \geq r_{2} \\
        \min\{q_{1},q_{2}\}+1 & i \leq r_{1} , j \leq r_{2}
        \end{cases} \label{eq:specific-xor-d-guessing-strategy-step-2} \ .
    \end{align}
    The second equality is because $S_{i}^{A}$ (resp.~$S_{j}^{B}$) contains $(q_{1} +\mbf{1}\{i \leq r_{1}\})$ (resp.~$(q_{2} + \mbf{1}\{j \leq r_{2}\})$~) distinct values, so it is certainly an upper bound, and it is achieved for every $(i,j)$ by considering the set 
    $$\{iq_{1}+\min\{q_{1},q_{2}\}-1 +\mbf{1}\{i \leq r_{1}\}-t,jq_{2}+\min\{q_{1},q_{2}\}-1 + \mbf{1}\{j \leq r_{2}\}-t\}_{t \in [\min\{q_{1}+\mbf{1}\{i \leq r_{1}\},q_{2}+\mbf{1}\{j \leq r_{2}\}\}]} \subset S^{A}_{i} \times S^{B}_{j} \ . $$
    Finally, combining \eqref{eq:specific-xor-d-guessing-strategy-step-1} and \eqref{eq:specific-xor-d-guessing-strategy-step-2},
    \begin{align}
        P_{S} &= \frac{1}{d^{2}} \Bigg[ \left(\sum_{i \leq r_{1}, j \leq r_{2}}   \min\{q_{1},q_{2}\}+1\right) + \left(\sum_{i > r_{1}, j > r_{2}}  \min\{q_{1},q_{2}\}\right) \\
        &\hspace{2cm} + \left(\sum_{i > r_{1} , j \leq r_{2}} \min\{q_{1},q_{2}+1\}\right) + \left(\sum_{i \leq r_{1} , j > r_{2}} \min\{q_{1}+1,q_{2}\}\right)
        \Bigg] \\
        &=\frac{1}{d^{2}} \Bigg[ (r_{1}r_{2} + (q_{1}-r_{1})(q_{2}-r_{2}))\min\{q_{1},q_{2}\} + r_{1}r_{2} \\
        & \hspace{2cm} + \left(\sum_{i > r_{1} , j \leq r_{2}} \min\{q_{1},q_{2}+1\}\right) + \left(\sum_{i \leq r_{1} , j > r_{2}} \min\{q_{1}+1,q_{2}\}\right) \Bigg] \\
        &=\frac{1}{d^{2}} \Bigg[ (r_{1}r_{2} + (q_{1}-r_{1})(q_{2}-r_{2}))\min\{q_{1},q_{2}\} + r_{1}r_{2} \\
        & \hspace{2cm} + (q_{1}-r_{1})r_{2}\min\{q_{1},q_{2}+1\} + r_{2}(q_{2}-r_{2})\min\{q_{1}+1,q_{2}\} \Bigg]
    \end{align}
    The special case in \eqref{eq:CMN-lower-bounds-for-xor-d-when-divide} follow as in the case $d_{1}$ and $d_{2}$ divide $d$, then \eqref{eq:specific-xor-d-guessing-strategy-step-2} simplifies to $\min{q_{1},q_{2}}$, which when plugged into \eqref{eq:specific-xor-d-guessing-strategy-step-1} results in the bound. \eqref{eq:CMN-lower-bounds-for-xor-d-when-divide-and-squares} follows from $q_{1} = q_{2} = \sqrt{d}$ when $d = d^{2}_{1} = d^{2}_{2}$.
\end{proof}

We now extend the above strategy to computing $\oplus_{d}^{2}$ as this is a case where we showed perfect dense network coding in the main text.
\begin{corollary}\label{cor:CMAN-lb-for-2-xor-d}
    Let $d \in \mbb{N}$ and $d \geq d_{1},d_{2},d_{3},d_{4} \in \mbb{N}$. Then,
    \begin{align}\label{eq:2-xor-d-CMN-lowerbound}
        P_{S}^{\oplus_{d}^{2}}(\mbf{C}(d_{1}d_{2},d_{3}d_{4})) \geq P_{S}^{\oplus_{d}}(\mbf{C}(d_{1},d_{3})) \cdot P_{S}^{\oplus_{d}}(\mbf{C}(d_{2},d_{4})) \ . 
    \end{align}
    In particular, when $d_{1}$, $d_{2}$, $d_{3}$, and $d_{4}$ all divide $d$, $d_{1} \leq d_{3}$, and $d_{2} \leq d_{4}$, then this bound matches the upper bound in Proposition \ref{prop:impossibility-of-weaker-resource-MNs}.
\end{corollary}
\begin{proof}
    By defining a bijection from $[d_{1}d_{2}]$ to $[d_{1}] \times [d_{2}]$, we can treat Alice's classical channel as two independent classical channels, i.e. $\Delta_{d_{1}d_{2}} \cong \Delta_{d_{1}} \otimes \Delta_{d_{2}}$. The same idea works for Bob using dimensions $d_{3}$ and $d_{4}$. Moreover, as the inputs are uniform and independent, the joint input state is the same as two independent, copies of the input for computing $\oplus_{d}$ as may formally be verified:
    \begin{align*}
        &\rho_{Z_{1}Z_{2}X_{1}X_{2}Y_{1}Y_{2}} \\
        &\coloneq \frac{1}{d^{4}} \sum_{q,r,s,t \in \mbb{Z}_{d}} \dyad{q \oplus_{d} s}_{Z_{1}} \otimes \dyad{r \oplus_{d} t}_{Z_{2}} \otimes \dyad{q}_{X_{1}} \otimes \dyad{r}_{X_{2}} \otimes \dyad{s}_{Y_{1}} \otimes \dyad{t}_{Y_{2}} \\
        &= \left(\frac{1}{d^{2}} \sum_{q,s \in \mbb{Z}_{d}} \dyad{q \oplus_{d} s}_{Z_{1}} \otimes \dyad{q}_{X_{1}} \otimes \dyad{s}_{Y_{1}}  \right) \otimes \left(\frac{1}{d^{2}} \sum_{r,t \in \mbb{Z}_{d}} \dyad{r \oplus_{d} t}_{Z_{1}} \otimes \dyad{r}_{X_{1}} \otimes \dyad{t}_{Y_{1}}  \right) \ .
    \end{align*}
    Combining these points, a valid strategy is for Alice and Bob to try and compute $q \oplus_{d} s$ and $r \oplus_{d} s$ independently, and in particular they may use the strategy given in Proposition \ref{prop:CMN-lower-bounds-for-xor-d}. By independence, this success of computing $(q,r) \oplus^{2}_{d} (s,t)$ is just the product of computing $q \oplus_{d} s$ and $r \oplus_{d} t$. This proves \eqref{eq:2-xor-d-CMN-lowerbound}. That this lower bound matches the upper bound in Proposition \ref{prop:impossibility-of-weaker-resource-MNs} when $d_{1},d_{2},d_{3},$ and $d_{4}$ all divide $d$, $d_{1} \leq d_{3}$ and $d_{2} \leq d_{4}$, we use the following sequence of inequalities:
    \begin{align}
         P_{S}^{\oplus_{d}}(\mbf{C}(d_{1},d_{3})) \cdot P_{S}^{\oplus_{d}}(\mbf{C}(d_{2},d_{4})) &= \frac{\min\{d_{1},d_{3}\}}{d}\frac{\min\{d_{2},d_{4}\}}{d} \\
         &= \frac{d_{1}d_{2}}{d^{2}} \\
         &= \frac{\min\{d_{1}d_{2},d_{3}d_{4}\}}{d^{2}} \\
         &\geq \max\{P_{S}^{\oplus_{d}^{2}}(\mbf{Q}(d_{1}d_{2},d_{3}d_{4})), P_{S}^{\oplus_{d}^{2}}(\mbf{C}^{N}(d_{1}d_{2},d_{3}d_{4}))\} \ , 
    \end{align}
    where the first equality uses \eqref{eq:CMN-lower-bounds-for-xor-d-when-divide} which holds by our assumptions on the dimensions and the inequality is Proposition \ref{prop:impossibility-of-weaker-resource-MNs}.
\end{proof}
The above both shows that our analysis is tight and that, in general, classical resources should be used unless \textit{both} entanglement-assistance and quantum communication are available. We conjecture that classical resources achieve at least close to the same success probability as configurations that do not have both entanglement-assistance and quantum communication even when the conditions on the dimensions given in Corollary \ref{cor:CMAN-lb-for-2-xor-d} do not hold. As Corollary \ref{cor:CMAN-lb-for-2-xor-d} shows there always exists larger $d$ where classical communication is sufficient, we take the above as evidence for our conjecture.

\paragraph{On CMN Achievable Bounds for Other DCB Functions}
The above shows how to rather generically code over a CMN for $\oplus_{d}$ and $\oplus_{d}^{2}$. It is reasonable to ask what can be said for other DCB functions. Strategies clearly depend on defining specific choices of compression functions and the decoding function, e.g. the mappings defined by the sets $S_{i}^{A}$ and $S_{j}^{B}$ and the function that maps $(i,j)$ to a guess of the inputs in the proof of Proposition \ref{prop:CMN-lower-bounds-for-xor-d}. The specific compression functions will depend on the structure of the DCB function, hence computing lower bounds is an onerous task in general. Indeed, we do not derive lower bounds for computing any other DCB using a CMN. However, we note that one can reduce the problem to equivalence classes of DCB functions using the identification of their multiplication tables with Latin squares. To illustrate this, note that Alice and Bob may re-label their input alphabets using bijections and preserve the the success probability strategy as one simply re-defines the compression functions appropriately. Similarly, a bijective re-labeling of the outputs preserves the success probability as one can apply the same re-labeling to the output of the previous decoding function. Thus, the success probability of a DCB function is invariant under triples of bijections on the inputs and output. This is the same equivalence that defines `isotopy classes' of Latin squares, c.f.~\cite[Part III, 1.12]{dinitz2007handbook}. Thus, one need only establish lower bounds for one function 
\begin{proposition}
    Let $f$ and $g$ be DCB functions of order $n$ whose multiplication tables are Latin squares in the same isotopy class. Then for any $d_{1},d_{2} \in \mbb{N}$, $P_{S}^{f}(\mbf{C}(d_{1},d_{2})) = P_{S}^{g}(\mbf{C}(d_{1},d_{2}))$.
\end{proposition}
\noindent While structurally pleasing, the number of isotopy classes of Latin squares grows super-exponentially in the order (see \cite[Part III, 1.16]{dinitz2007handbook}), limiting the practicality of the above.

\subsection{General Statement of Communication Advantage}
To present our main result, we state two versions of the theorem for different communities. First, we state the form relevant to one-shot quantum information theory.
\begin{theorem}\label{thm:one-shot-info-advantage}
    Let $d \in \mbb{N}$ and $f:\cX \times \cY \to \cZ$ be a DCB function of order $d^{2}$ whose multiplication table is the multiplication table of a TNC group (Definition \ref{def:TNC}). For any $d \leq d_{1},d_{2} < d^{2} \eqqcolon \vert \cX \vert$,
    \begin{equation}
    \begin{aligned}
        \max\{P^{f}_{S}(\mbf{C}(d_{1},d_{2})),P^{f}_{S}(\mbf{C}^{E}(d_{1},d_{2})),P^{f}_{S}(\mbf{C}^{N}(d_{1},d_{2})),P^{f}_{S}(\mbf{Q}(d_{1},d_{2}))\} & \leq \frac{\min\{d_{1},d_{2}\}}{d^{2}} \\
        &<1 = P_{S}^{f}(\mbf{Q}^{E}(d_{1},d_{2}))  \ .
    \end{aligned}
    \end{equation}
    Moreover, if $f = \oplus_{d}^{2}$ and there exist integers $d_{1}' \leq d_{3}'$$d_{2}' \leq d_{4}'$ that all divide $d$ such that $d_{1} = d_{1}'d_{2}'$, $d_{2} = d_{3}'d_{4}'$, the maximum and the first inequality can be replaced with an equality, thereby showing that without both entanglement-assistance and quantum communication, one generally may as well just use classical communication under these signaling dimension constraints.
\end{theorem}
\begin{proof}
    The first inequality follows from Proposition \ref{prop:impossibility-of-weaker-resource-MNs} and the containments in Proposition \ref{prop:MN-containments}. The strict inequality follows from the assumptions on $d,d_{1},d_{2}$. The final equality follows from Theorem \ref{thm:dense-network-coding} and the assumptions on $d,d_{1},d_{2}$ as one may embed the strategy given in Protocol \ref{prot:dense-network-coding-of-some-type} into the specified quantum channels. The  remaining claim about when equality holds follows from Corollary \ref{cor:CMAN-lb-for-2-xor-d}.
\end{proof}

Next, we formally state the asymptotic claims likely more relevant to communication complexity.
\begin{theorem}\label{thm:asymptotic-comm-complexity-adv}
    For $n \in \mbb{N}$, let $d_{n} = 2^{n}$. Consider sequences $(d_{1,n})_{n}, (d_{2,n})_{n}$ satisfying $d_{n} \leq d_{1,n}, d_{2,n} \leq d_{n}^{2\alpha}$ for some $\alpha \in [1/2,1)$ for all $n$. Then $P_{S}^{\oplus_{d_{n}}^{2}}(\mbf{Q}^{E}(d_{1,n},d_{2,n})) = 1$ for all $n$, but for
    $$\mbf{S}(d_{1,n},d_{2,n}) \in \{\mbf{C}(d_{1,n},d_{2,n}),\mbf{C}^{E}(d_{1,n},d_{2,n}),\mbf{C}^{N}(d_{1,n},d_{2,n}),\mbf{Q}(d_{1,n},d_{2,n})\} \ , $$
    $P_{S}^{\oplus_{d_{n}}^{2}}(\mbf{S}(d_{1_{n}},d_{2_{n}})) = O(\exp(-n))$, i.e. without both entanglement and quantum communication, the success probability goes to $0$ exponentially fast in $n$. \\
    
    \noindent In particular, there exists a sequence of DCB functions such that there is a linear advantage in the communication complexity of computing in the simultaneous message passing model using shared entanglement and quantum communication, but no advantage over classical communication without both.
\end{theorem}
\begin{proof}
    Using our assumptions and applying Theorem \ref{thm:dense-network-coding}, $P_{S}^{\oplus_{d_{n}}^{2}}(\mbf{Q}^{E}(d_{1,n},d_{2,n})) \geq P_{S}^{\oplus_{d_{n}}^{2}}(\mbf{Q}^{E}(d_{n},d_{n})) = 1$ for all $n$. Similarly, by applying Theorem \ref{thm:one-shot-info-advantage}, $P_{S}^{\oplus_{d_{n}}^{2}}(\mbf{S}(d_{1_{n}},d_{2_{n}})) \leq \frac{d_{1,n}}{d_{n}^{2}} = d_{n}^{2(\alpha-1)} = 2^{-2n(1-\alpha)} \in O(\exp(-n))$. The final statement about linear advantage follows from choosing $\alpha = 1/2$, so that the signaling dimension of each party is the square root of the size of the alphabet.
\end{proof}

\begin{remark}
The theorem in the main text (Theorem \ref{thm:main-text-communication-advantage}) is recovered by letting $d_{1,n} = 2^{n} = d_{2,n}$ and considering $\log(d_{1,n}),\log(d_{2,n})$ so that it is measured in qubits.
\end{remark}

\section{Noise-Robustness of Communication Advantage}\label{app:noise-robustness}
In this section we formalize the noise-robustness of the communication advantage in Theorem \ref{thm:one-shot-info-advantage}. 
Evidence of such noise-robustness was numerically observed while computing the bitwise XOR of bitstrings of length 2 in \cite{doolittle2024operational_nonclassicality}. A mathematical argument that noise-robustness should hold in general proceeds as follows. By Lemma \ref{lem:reduction-to-guessing-probability} and \eqref{eq:guess-prob-and-min-ent}, the ability to guess the function is controlled by the output min-entropy. The min-entropy is uniformly continuous \cite{Marwah_2022,Bluhm_2024}. Thus if noise is injected into the encoders, transmission state, or initial shared entangled state continuously, the probability of guessing should degrade in a continuous manner. However, this argument does not account for noise in the decoder. Here, in Theorem \ref{thm:noise-robustness}, we prove a form of noise-robustness that allows us to also account for noise in the decoder.

As already alluded to, there can be noise in the shared state, the encoder, the transmission, and the decoder. These noisy components remain modeled by a quantum state and quantum channels. To quantify noise in the state we use the trace distance 
\begin{align}
    \text{TD}(\rho,\sigma) \coloneq \frac{1}{2} \left\Vert \rho - \sigma \right\Vert_{1} \ , 
\end{align} 
and to quantify noise in the channels we use the diamond distance \cite{Kitaev_1997}
\begin{align}
    \frac{1}{2}\Vert \cE_{A \to B} - \cF_{A \to B} \Vert_{\diamond} 
    &=\sup_{\rho_{RA} \in \Density(R \otimes A)} \frac{1}{2} \left\Vert (\id_{R} \otimes (\cE - \cF))(\rho) \right\Vert_{1} \\
    &= \max_{\dyad{\psi}_{A'A} \in \Density(A' \otimes A)} \frac{1}{2} \left\Vert (\id_{R'} \otimes (\cE - \cF))(\dyad{\psi}_{A'A}) \right\Vert_{1}\ , \label{eq:diamond-norm-state-optimization}
\end{align}
where $A' \cong A$. For clarity, we summarize the properties of the diamond norm to which we will appeal.
\begin{proposition}[Diamond Norm Properties]\label{prop:diamond-norm-properties}
Let $\cE^{1}_{A \to B}$, $\cF^{1}_{A \to B}$, $\cE^{2}_{B \to C}$, and $\cG^{2}_{B \to C}$ be quantum channels.
    \begin{enumerate}
        \item $\Vert \cE^{2} \circ \cE^{1} - \cF^{2} \circ \cF^{1}\Vert_{\diamond} \leq \Vert \cE^{1} - \cF^{1} \Vert_{\diamond} + \Vert \cE^{2} - \cF^{2} \Vert_{\diamond}$.
        \item $\Vert \id_{R} \otimes \cE^{1} - \id_{R} \otimes \cF^{1} \Vert_{\diamond} = \Vert \cE^{1} - \cF^{1} \Vert_{\diamond}$ for any dimension of $R$.
        \item $\Vert \cE^{1} \otimes \cE^{2} - \cF^{1} \otimes \cF^{2} \Vert_{\diamond} \leq \Vert \cE^{1} - \cF^{1} \Vert_{\diamond} + \Vert \cE^{2} - \cF^{2} \Vert_{\diamond}$.
    \end{enumerate}
\end{proposition}
\begin{proof}[Proof Summary]
    Item 1 is Item 2 of \cite[Proposition 3.48]{WatrousBook}. Item 2 is \cite[Corollary 3.47]{WatrousBook}. To establish Item 3, we use the series of inequalities
    \begin{align}
        \Vert \cE^{1} \otimes \cE^{2} - \cF^{1} \otimes \cF^{2} \Vert_{\diamond} &= \Vert \cE^{1} \otimes \cE^{2} - \cF^{1} \otimes \cF^{2} \Vert_{\diamond} \\
        &= \Vert \cE^{1} \otimes \cE^{2} - \cE^{1} \otimes \cF^{2} + \cE^{1} \otimes \cF^{2} - \cF^{1} \otimes \cF^{2} \Vert_{\diamond} \\
        &\leq \Vert \cE^{1} \otimes \cE^{2} - \cE^{1} \otimes \cF^{2} \Vert_{\diamond} + \Vert \cE^{1} \otimes \cF^{2} - \cF^{1} \otimes \cF^{2} \Vert_{\diamond} \\
        &= \Vert \cE^{1} \otimes (\cE^{2} - \cF^{2}) \Vert_{\diamond} + \Vert (\cE^{1} - \cF^{1}) \otimes \cF^{2} \Vert_{\diamond} \\
        &= \Vert \cE^{2}-\cF^{2}\Vert_{\diamond} + \Vert \cE^{1} - \cF^{1} \Vert_{\diamond} \ ,
    \end{align}
    where the inequality is the triangle inequality and the fourth equality uses the multiplicativity of the completely bounded $1$-norm over tensor products \cite[Theorem 3.49]{WatrousBook}.
\end{proof}

We also recall the following well-known fact that follows from the triangle inequality of the Schatten $1$-norm and the definition of diamond and trace distance.
\begin{proposition}\label{prop:decompose-processing-and-shared-state-error}
    For any states $\rho_{A},\sigma_{A}$ and channels $\cE_{A \to B}, \cF_{A \to B}$, 
    \begin{align}
    \frac{1}{2}\Vert \cE(\rho) - \cF(\sigma) \Vert_{1} \leq \frac{1}{2}\Vert \cE - \cF \Vert_{\diamond} + \text{TD}(\rho,\sigma) \ . 
    \end{align}
\end{proposition}

We will denote the noisy components of an implemted protocol with superscripts $r$ for `real' and the noiseless components with $id$ for `ideal,' e.g. the real and noisy encoder are denoted $\cE^{r}$ and $\cE^{id}$ respectively. We define the following notions of error: 
\begin{align}
    \text{TD}(\sigma^{r},\sigma^{id}) \coloneq \ve_{st} \quad \frac{1}{2}\Vert \cE^{r}_{i} - \cE^{id}_{i} \Vert_{\diamond} \coloneq \ve_{enc,i} \quad \frac{1}{2}\Vert \cT^{r}_{i} - \cT^{id}_{i} \Vert_{\diamond} \coloneq \ve_{tr,i} \quad \frac{1}{2}\Vert \cD^{r} -\cD^{id} \Vert_{\diamond} \coloneq \ve_{dec} \ , 
\end{align}
which respectively correspond to the difference between the ideal and noisy shared entanglement, the difference between party $i$'s ideal and noisy encoder and transmission channel, and the difference between the receiver's ideal and noisy decoder.

Next we show we can measure how often $Z' = Z$ according to the joint distribution $p_{ZZ'}$ generated from $p_{Z}$ and some effective channel $P_{Z'\vert Z}$ where $\cZ' \cong \cZ$ is the same as the trace distance between $p_{ZZ'}$ and the perfect correlation $\chi^{\vert p}_{ZZ'}$. 
\begin{lemma}\label{lem:distance-from-perfect-correlation}
    Consider classical-classical states 
    \begin{align}
        p_{ZZ'} = \sum_{z} p(z) \sum_{z'} p(z'\vert z) \dyad{z}_{Z} \otimes \dyad{z'}_{Z'} \quad \chi^{\vert p}_{ZZ'} = \sum_{z} p(z) \dyad{z}_{Z} \otimes \dyad{z}_{Z'} \ . 
    \end{align}
    Then $\Pr_{p}[Z \neq Z'] = \text{TD}(p_{ZZ'},\chi^{\vert p}_{ZZ'})$, i.e. the trace distance between these states is the probability of $Z'$ not taking the same value as $Z$.
\end{lemma}
\begin{proof}
    This is a direct calculation that we provide for completeness:
    \begin{align}
        \text{TD}(p_{ZZ'},\chi^{\vert p}_{ZZ'}) &= \frac{1}{2} \left\Vert  \sum_{z} p(z) \sum_{z'} p(z'\vert z) \dyad{z}_{Z} \otimes \dyad{z'}_{Z'} - \sum_{z} p(z) \dyad{z}_{Z} \otimes \dyad{z}_{Z'} \right\Vert_{1} \\ 
        &= \frac{1}{2} \sum_{z} p(z) \sum_{z'} \left\vert p(z' \vert z) - \delta_{z',z}\right\vert \\
        &= \frac{1}{2} \sum_{z} p(z) \left[\left\vert p(z'=z \vert z) - 1\right\vert + \sum_{z' \neq z} \vert p(z'|z) \vert \right] \\
        &=\frac{1}{2} \sum_{z} p(z) \left[\left\vert 1 - \sum_{z' \neq z} p(z'|z)  - 1\right\vert + \sum_{z' \neq z} p(z'|z) \right] \\
        &= \frac{1}{2} \sum_{z} p(z) 2\left[\sum_{z' \neq z} p(z' \vert z) \right] \\
        &= \sum_{z} \sum_{z' \neq z} p(z,z') \\
        &= \Pr_{p}[Z' \neq Z] \ ,
    \end{align}
    where the second equality is pulling out $p(z)$ and simplifying the trace norm on these diagonal operators,
\end{proof}

Using the above facts and known properties of the diamond distance, we prove our main noise-robustness result.
\begin{theorem}\label{thm:noise-robustness}
    Consider a noisy implementation of Protocol \ref{prot:dense-network-coding-of-some-type} using an entanglement-assisted quantum multiaccess network described by channels and resource state ($\cE^{r}_{1}$,$\cE^{r}_{2}$,$\cT^{r}_{1}$,$\cT^{r}_{2}$,$\cD^{r}$,$\sigma^{r}_{AB}$). Then the probability that the output $Z'$ is the correct value of the function,  $\Pr[Z = Z']$, satisfies 
    \begin{align} 
       \Pr[Z = Z'] \geq 1 - (\ve_{st} + \ve_{enc,1} + \ve_{enc,2} + \ve_{tr,1} +\ve_{tr,2} + \ve_{dec}) \ . 
    \end{align}
\end{theorem}
\begin{proof}
    To simplify the notation, let 
    \begin{align}\label{eq:simplified-encoder-trans-notation}
        \cE^{x} \coloneq \cE^{x}_{1} \otimes \cE^{x}_{2} \quad \cT^{x} \coloneq \cT^{x}_{1} \otimes \cT^{x}_{2} \text{ for } x \in \{r,id\} \ . 
    \end{align} 
    First, as we are trying to implement the perfect strategy given in Protocol \ref{prot:dense-network-coding-of-some-type}, we know the ideal strategy results in the global state being perfectly correlated, i.e.
    \begin{align}
        \id_{Z} \otimes (\cD^{id} \circ \cT^{id} \circ \cE^{id})(\rho_{ZXY} \otimes \sigma^{id}_{AB}) = \chi^{\vert p}_{ZZ'} \ . 
    \end{align}
    Moreover the real strategy results in a joint classical-classical state 
    \begin{align}
        (\id_{Z} \otimes \cD^{r} \circ \cT^{r} \circ \circ \cE^{r})(\rho_{ZXY} \otimes \sigma^{r}_{AB}) \eqqcolon p_{ZZ'} \ . 
    \end{align}
    Thus, by Lemma \ref{lem:distance-from-perfect-correlation}, the probability of guessing the value of $Z$ incorrectly is
    \begin{align}
        \Pr_{p}[Z \neq Z'] &= \text{TD}(p_{ZZ'},\chi_{ZZ'}) \\
        &= \text{TD}(\id_{Z} \otimes (\cD^{id} \circ \cT^{id} \circ \cE^{id})(\rho_{ZW^{2}_{1}} \otimes \sigma^{id}_{AB}),(\id_{Z} \otimes \cD^{r} \circ \cT^{r} \circ \circ \cE^{r})(\rho_{ZXY} \otimes \sigma^{r}_{AB})) \ . \label{eq:error-robust-bound-step-1}
    \end{align}
    We aim to upper bound Eq.~\eqref{eq:error-robust-bound-step-1} so that we can obtain a lower bound on the probability of correctly guessing the value of $Z$.

    For notational simplicity, let $\cG^{x} \coloneq \cD^{x} \circ \cT^{x} \circ \cE^{x}$ for $x \in \{r,id\}$. We then have
    \begin{align}
        & \text{TD}(\id_{Z} \otimes \cG^{id}(\rho_{ZW^{2}_{1}} \otimes \sigma^{id}_{AB}),(\id_{Z} \otimes \cG^{r})(\rho_{ZW^{2}_{1}} \otimes \sigma^{r}_{AB})) \\
        &\leq \frac{1}{2}\Vert \id_{Z} \otimes \cG^{r} - \id_{Z} \otimes \cG^{id} \Vert_{\diamond} + \text{TD}(\rho_{ZW^{2}_{1}} \otimes \sigma^{id}_{AB},\rho_{ZW^{2}_{1}} \otimes \sigma^{r}_{AB}) \\
        &= \frac{1}{2} \Vert \id_{Z} \otimes \cG^{r} - \id_{Z} \otimes \cG^{id} \Vert_{\diamond} + \text{TD}(\sigma^{id}_{AB}, \sigma^{r}_{AB}) \\
        &= \frac{1}{2}\Vert \cG^{r} - \cG^{id} \Vert_{\diamond} + \text{TD}(\sigma^{id}_{AB}, \sigma^{r}_{AB}) \\
        &= \frac{1}{2}\Vert \cG^{r} - \cG^{id} \Vert_{\diamond} + \ve_{st} \label{eq:error-robust-bound-step-2}
    \end{align}
    where the inequality is Proposition \ref{prop:decompose-processing-and-shared-state-error}, the first equality follows from the multiplicativity of the Schatten $1$-norm over tensor products, the second equality is Item 2 of Proposition \ref{prop:diamond-norm-properties}, and the final equality is our definition of error in the state. Next, by our definitions and Item 1 of Proposition \ref{prop:diamond-norm-properties},  
    \begin{align}
        \frac{1}{2} \Vert \cG^{r} - \cG^{id} \Vert_{\diamond} &= \frac{1}{2}\Vert \cD^{r} \circ \cT^{r} \circ \cE^{r} -  \cD^{id} \circ \cT^{id} \circ \cE^{id} \Vert_{\diamond} \\ 
        &\leq \frac{1}{2} \Vert \cE^{r} - \cE^{id} \Vert_{\diamond} + \frac{1}{2}\Vert \cT^{r} - \cT^{id}\Vert_{\diamond} + \frac{1}{2} \Vert \cD^{r} - \cD^{id} \Vert_{\diamond} \ . \label{eq:error-robust-bound-step-3} 
    \end{align}
    Finally, by our definitions and Item 3 of Proposition \ref{prop:diamond-norm-properties},
    \begin{gather}
       \frac{1}{2}\Vert \cE^{r} - \cE^{id} \Vert_{\diamond} = \frac{1}{2}\Vert \cE^{r}_{1} \otimes \cE^{r}_{2} - \cE^{id}_{1} \otimes \cE^{id}_{2} \Vert_{\diamond}  \leq \frac{1}{2}\Vert \cE^{r}_{1} - \cE^{id}_{1} \Vert_{\diamond} + \frac{1}{2}\Vert \cE^{r}_{2} - \cE^{id}_{2} \Vert_{\diamond} = \ve_{enc,1}+\ve_{enc,2} \notag \\
       \frac{1}{2} \Vert \cT^{r} - \cT^{id} \Vert_{\diamond} = \frac{1}{2} \Vert \cT^{r}_{1} \otimes \cT^{r}_{2} - \cT^{id}_{1} \otimes \cT^{id}_{2} \Vert_{\diamond} \leq \frac{1}{2} \Vert \cT^{r}_{1} - \cT^{id}_{2} \Vert_{\diamond} + \frac{1}{2} \Vert \cT^{r}_{2} - \cT^{id}_{2} \Vert_{\diamond} = \ve_{tr,1} + \ve_{tr,2} \ . \label{eq:error-robust-bound-step-4} 
    \end{gather}
    Combining Eqs.~\eqref{eq:error-robust-bound-step-1}, \eqref{eq:error-robust-bound-step-2}, \eqref{eq:error-robust-bound-step-3}, \eqref{eq:error-robust-bound-step-4}, and using our definitions,
    \begin{align}
        \Pr_{p}[Z \neq Z'] \leq  \ve_{st} + \ve_{enc,1} +\ve_{enc,2} + \ve_{tr,1} +\ve_{tr,2} + \ve_{dec} \ . 
    \end{align}
    Noting the probability of guessing the correct value of $Z$ is $\Pr_{p}[Z=Z'] = 1 - \Pr_{p}[Z \neq Z']$ completes the proof.
\end{proof}
\begin{remark}
    This analysis can straightforwardly be extended to establish noise robustness in  similar network computation settings.
\end{remark}

\paragraph{Derivation of Concrete Example} 
Here we present the derivation of a concrete noise model that is used for Fig.~\ref{fig:main-text-noise-robustness}. The noise model is as follows:
\begin{enumerate}
    \item The ideal state is replaced with an isotropic state $\rho_{\lambda}$ (recall \eqref{eq:isotropic-state-defn}).
    \item The encoder is ideal.
    \item \sloppy Each ideal transmission $\id_{d}$ is replaced with the composition of an erasure and depolarizing channel $\cD^{p,e}_{\mbb{C}^{d} \to \mbb{C}^{d+1}} \coloneq \cE^{e} \circ \cD^{p}_{\mbb{C}^{d} \to \mbb{C}^{d}}$ where
    \begin{align}
        \cD^{p}(K) = (1-p)K + p\Tr[K]\pi \quad \cE^{e}(K) = (1-e)K + e\Tr[K]\dyad{\perp}  \ .
    \end{align}
    We note this model can be interpreted optically as a channel with transmittivity $\eta = (1-e)$ where the loss is heralded and, when the signal is not lost, has undergone depolarization.
    \item The decoder is a probabilistic mixture of the ideal Bell state measurement $\cM^{\BSM}$ and a measurement whose outcome is uniformly random independent of the input, which can be represented as the replacer channel that traces out the input and prepares the maximally mixed state:
    \begin{align}
        \cM^{\omega}(K) = (1-\omega)\cM^{\BSM}(K) + \omega \cR^{\pi}(K) \ .
    \end{align}
\end{enumerate}
Using this model, we find the following bound.
\begin{proposition}
    Using the noise model given above to implement Protocol \ref{prot:dense-network-coding-of-some-type} for an $f$ defined by the multiplication table of a TNC group of order $d^{2}$ (Definition \ref{def:TNC}), we have 
    \begin{align}
        P^{f}_{S}(\lambda,p,e,\omega) \geq 1 - \lambda\frac{d^{2}-1}{d^{2}} - 2\left[e + p \frac{d^{2}-1}{d^{2}} \right] - \omega(1-1/d) \ . 
    \end{align}
\end{proposition}
\begin{remark}
    One recovers the function $P_{S}^{f}(e,p,n) \coloneq 1 - 2(e + p\frac{2^{2n}-1}{2^{2n}})$ in the main text by letting $\lambda = 0$, $\omega =0$, and using $d = 2^{n}$.
\end{remark}
\begin{proof}
    As the ideal state is the maximally entangled state, by Proposition \ref{prop:trace-dist-isotropic-states}, $\ve_{st} = \text{TD}(\rho_{\lambda},\rho_{0}) =  \lambda \frac{d^{2}-1}{d^{2}}$. As the ideal transmission channel is the identity channel, we have
    \begin{align}
        \ve_{tr,i} = \frac{1}{2} \Vert \cD^{p,e} - \cD^{0,0} \Vert_{\diamond} &= \frac{1}{2}\Vert \cE^{e} \circ \cD^{p}_{\mbb{C}^{d} \to \mbb{C}^{d}} - \cE^{0} \circ \id_{d} \Vert_{\diamond} \\
        &\leq \frac{1}{2}\Vert \cE^{e} - \cE^{0} \Vert_{\diamond} + \frac{1}{2} \Vert \cD^{p} - \id_{d} \Vert_{\diamond} = e + p\frac{d^{2}-1}{d^{2}} \ ,
    \end{align}
    where the second equality uses that $\cD^{0,0}$ is $\id_{d}$ embedded into acting on $\mbb{C}^{d} \oplus \ket{\perp}$, the inequality is Item 1 of Proposition \ref{prop:diamond-norm-properties}, and the last equality is Propositions \ref{prop:diamond-dist-erasure-channels} and \ref{prop:diamond-dist-of-depol-channels}. Finally, as the ideal decoder is the group Bell state measurement which is a projective measurement, by Proposition \ref{prop:diamond-dist-of-noisy-PVM},
    \begin{align}
        \ve_{dec} = \frac{1}{2} \Vert \cM^{w} - \cM^{BSM} \Vert_{\diamond} = \omega(1 - \frac{1}{d}) \ .
    \end{align}
    Thus, as we have established upper bounds on the terms being subtracted in Theorem \ref{thm:noise-robustness}, applying Theorem \ref{thm:noise-robustness} completes the proof.
\end{proof}

\paragraph{Auxiliary Propositions} The rest of this section are proofs of the trace and diamond distance for various standard families of states/channels for which we could not find references. This can be skipped with no loss unless one is interested in these types of calculations.

The following is surely known, but we could not find it, so we provide it. 
\begin{proposition}\label{prop:trace-dist-isotropic-states}
    For $\lambda,\lambda' \in [0,1]$, the distance between two isotropic states on $\mbb{C}^{d} \otimes \mbb{C}^{d}$ is
    \begin{align}
        \text{TD}(\rho_{\lambda},\rho_{\lambda'}) = \vert \lambda - \lambda' \vert \frac{d^{2}-1}{d^{2}} \ . 
    \end{align}
\end{proposition}
\begin{proof}
    First we reduce this to computing $\text{TD}(\dyad{\Phi},\pi)$. As trace distance is symmetric in the arguments, without loss of generality, $\lambda' \geq \lambda$ so
    \begin{align}
        \text{TD}(\rho_{\lambda},\rho_{\lambda'}) &= \frac{1}{2} \Vert (1-\lambda)\dyad{\Phi} + \lambda \pi - (1-\lambda')\dyad{\Phi} - \lambda' \pi \Vert_{1} \\ 
        &= \frac{1}{2} \Vert (\lambda' - \lambda)\Phi^{+} - (\lambda'-\lambda)\pi \Vert_{1} \\
        &= \vert \lambda' - \lambda \vert  \cdot \text{TD}(\dyad{\Phi},\pi) \ . 
    \end{align}
    Next,
    \begin{align}
        \text{TD}(\dyad{\Phi},\pi) = \frac{1}{2}\Vert \pi - \dyad{\Phi} \Vert_{1} &= \frac{1}{2} \left\Vert \frac{1}{d^{2}} \sum_{x,y} \dyad{\Phi_{x,y}} - \dyad{\Phi} \right\Vert_{1} \\
        &= \frac{1}{2} \left\Vert \frac{1}{d^{2}} \sum_{(x,y) \neq (0,0)} \dyad{\Phi_{x,y}} - \left(1-\frac{1}{d^{2}}\right) \dyad{\Phi} \right\Vert_{1} \\
        &= \frac{1}{2} \left[ \frac{1}{d^{2}} \sum_{(x,y) \neq (0,0)}  \left\Vert \dyad{\Phi_{x,y}} \right\Vert_{1} + \left\vert - \left(1-\frac{1}{d^{2}}\right) \right\vert  \right] \\
        &= \frac{1}{2} \left[ \frac{d^{2}-1}{d^{2}} + \frac{d^{2}-1}{d^{2}} \right] \\
        &= \frac{d^{2}-1}{d^{2}} \ ,
    \end{align}
    where we decomposed the maximally mixed state into Bell states and then used pairwise orthogonality to split the $L_{1}$-norm.
\end{proof}

\begin{proposition}\label{prop:diamond-dist-erasure-channels}
    For a fixed dimension $d$ and erasure probabilities $e,e' \in [0,1]$, $\frac{1}{2}\Vert\cE^{e} - \cE^{e'} \Vert_{\diamond} = \vert e' - e \vert$.
\end{proposition}
\begin{proof}
    As the diamond distance is symmetric in its arguments, we may without loss of generality assume $e' \geq e$. It is a direct calculation using \eqref{eq:diamond-norm-state-optimization}:
    \begin{align}
        \frac{1}{2}\Vert \cE^{e} - \cE^{e'} \Vert_{\diamond} &= \sup_{\dyad{\psi}_{A'A}} \frac{1}{2} \Vert (\id_{A'} \otimes (\cE^{e} - \cE^{e'}))(\dyad{\psi}) \Vert_{1} \\
        &= \sup_{\psi_{A'A}} \frac{1}{2} \Vert (1-e)\dyad{\psi}_{A'A} - (1-e')\dyad{\psi}_{A'A} + e \psi_{A'} \otimes \dyad{\perp} - e' \psi_{A'} \otimes \dyad{\perp} \Vert_{1} \\ 
        &= \sup_{\psi_{A'A}} \frac{1}{2} \Vert (e'-e)\dyad{\psi}_{A'A} \Vert_{1} + \frac{1}{2}\Vert (e'-e) \psi_{A'} \otimes \dyad{\perp} \Vert_{1} \\ 
        &= \frac{1}{2} \cdot 2 \vert e' - e \vert \ , 
    \end{align}
    where we used that $\dyad{\psi}_{A'A}$ and $\psi_{A'} \otimes \dyad{\perp}$ are in orthogonal subspaces to decompose the trace distance.
\end{proof}

In the following two computations we will use the SDP for the diamond distance.
\begin{proposition}\cite{khatri-book}\label{prop:diamond-dist-SDP}
    Let $\cN_{A \to B}$, $\cM_{A \to B}$ be quantum channels. Then
    \begin{align}
        \frac{1}{2}\Vert \cN - \cM \Vert_{\diamond} &= \sup_{\rho_{A'} \geq 0, \Omega_{A'A} \geq 0} \{ \Tr[\Omega_{A'B}(\Gamma^{\cN}-\Gamma^{\cM})] : \Omega_{A'B} \leq \rho_{A'} \otimes I_{B} \, , \, \Tr[\rho_{A'}] = 1 \} \label{eq:diamond-dist-primal} \\
        &= \inf_{\mu \geq 0 \, , \, Z_{A'B} \geq 0} \{ \mu : Z_{A'B} \geq  \Gamma^{\cN} - \Gamma^{\cM} \, , \, \mu I_{A'} \geq \Tr_{B}[Z] \} \label{eq:diamond-dist-dual} \ ,
    \end{align}
    where $\Gamma^{\cN} \coloneq d_{A} \cdot (\id_{A'} \otimes \cN)(\dyad{\Phi^{+}})$ is the Choi operator.
\end{proposition}

\begin{proposition}\label{prop:diamond-dist-of-depol-channels}
    For dimension $d$ and noise parameters $p,p' \in [0,1]$, 
    $$ \frac{1}{2}\Vert \cD_{\mbb{C}^{d} \to \mbb{C}^{d}}^{p} - \cD_{\mbb{C}^{d} \to \mbb{C}^{d}}^{p'} \Vert_{\diamond}  = \vert p' - p \vert \frac{d^{2}-1}{d^{2}} \ . $$
\end{proposition}
\begin{proof}
    First we reduce the problem to determining the diamond distance between the identity channel and the replacer channel $\cR^{\pi}(X) = \Tr[X]\pi$. As the diamond norm is symmetric in the arguments, we can without loss of generality assume $p' \geq p$. Now,
    \begin{align}
       \frac{1}{2}\Vert \cD^{p} - \cD^{p'} \Vert_{\diamond} &= \sup_{\psi_{A'A}} \frac{1}{2} \Vert (1-p)\psi + p\psi_{A'} \otimes \pi_{A} - (1-p')\psi + p' \psi_{A'} \otimes \pi_{A} \Vert_{1} \\
        &= \sup_{\psi_{A'A}} \frac{1}{2} \Vert (p'-p)\psi - (p'-p) \psi_{A'} \otimes \pi_{A} \Vert_{1} \\
        &= \vert p' - p \vert \sup_{\psi_{A'A}} \frac{1}{2} \Vert \psi - \psi_{A'} \otimes \pi_{A} \Vert_{1} \\
        &= \vert p' - p \vert \frac{1}{2}\Vert\id - \cR^{\pi} \Vert_{\diamond} \ ,
    \end{align}
    where the final equality just follows from the action of these channels. 

    Now we compute $\frac{1}{2}\Vert \id - \cR^{\pi} \Vert_{\diamond}$. Using Proposition \ref{prop:diamond-dist-SDP},
    \begin{align}\label{eq:diamond-dist-SDP-for-id-vs-rep}
        \frac{1}{2}\Vert \id - \cR^{\pi} \Vert_{\diamond} &= \sup_{\rho_{A'} \geq 0, \Omega_{A'A} \geq 0} \{ \Tr[\Omega_{A'A}(d\dyad{\Phi} - I \otimes \pi)] : \Omega_{A'A} \leq \rho_{A'} \otimes I_{A} \, , \, \Tr[\rho_{A'}] = 1 \} \ .
    \end{align}
     We now use symmetry to reduce this to a simple linear program. Note $d\dyad{\Phi}$ and $I \otimes \pi$ are invariant under conjugation by $U \otimes \ol{U}$, and thus are invariant under the twirling map $\cT(\cdot) \coloneq \int (U \otimes \ol{U}) \cdot (U \otimes \ol{U})^{\dagger} dU$. Noting for any feasible $\rho_{A'}$, $\Omega_{A'A}$ we have 
    \begin{gather}
        \Tr[\Omega_{A'A}(d\dyad{\Phi} - I \otimes \pi)] = \Tr[\Omega_{A'A}\cT(d\dyad{\Phi} - I \otimes \pi)] = \Tr[\cT(\Omega_{A'A})(d\dyad{\Phi} - I \otimes \pi)] \\
        \cT(\Omega_{A'A}) \leq \cT(\rho_{A'} \otimes I_{A}) = \pi_{A'} \otimes I_{A} \ ,
    \end{gather}
    where the first follows from the self-adjointedness of $\cT$ and the second follows from $\cT$ being a completely positive map, we may conclude that without loss of generality there is an optimizer of the form $(\cT(\Omega_{A'A}),\pi_{A'})$. As is well-known, by Schur-Weyl duality, $\dyad{\Phi}$ and $I \otimes I - \dyad{\Phi} \coloneq \Pi^{\perp}$ form a basis for $U \otimes \ol{U}$-invariant operators. Thus, we have
    \begin{align}
        \frac{1}{2}\Vert \id - \cR^{\pi} \Vert_{\diamond} &= \sup_{x,y \geq 0} \left\{ \Tr[(x\dyad{\Phi}+y\Pi^{\perp})(d\dyad{\Phi} - I \otimes \pi)] : x\dyad{\Phi}+y\Pi^{\perp} \leq \pi_{A'} \otimes I_{A} \right\} \ . 
    \end{align}
    We further simplify this optimization problem. We simplify the objective function by direct calculation:
    \begin{align*}
        \Tr[(x\dyad{\Phi}+y\Pi^{\perp})(d\dyad{\Phi} - I \otimes \pi)] &= xd + yd \Tr[\Pi^{\perp}\dyad{\Phi}] - x\Tr[\dyad{\Phi}I \otimes \pi] - y \Tr[\Pi^{\perp}I \otimes \pi] \\ 
        &= xd  - x\Tr[\pi^{2}] - y \Tr[(I \otimes I) I \otimes \pi] + y \Tr[\dyad{\Phi} I \otimes \pi] \\ 
        &= xd - x/d - yd + y/d \\
        &= (d-1/d)(x-y) \ . 
    \end{align*}
    We simplify the constraint by re-writing $\pi_{A'} \otimes I_A = \frac{1}{d} I_{A'} \otimes I_{A} = \frac{1}{d}(\dyad{\Phi} + \Pi^{\perp})$, so that by orthogonality of these operators, the constraint becomes two constraints: $x \leq 1/d$ and $y \leq 1/d$. Thus, we have
    \begin{align}
         \frac{1}{2}\Vert \id - \cR^{\pi} \Vert_{\diamond} &= (d-1/d)\sup\{x-y: 0 \leq x \leq 1/d \, , 0 \leq y \leq 1/d \} \ . 
    \end{align}
    Clearly this is achieved when $y = 0$ and $x = 1/d$. Thus, $\frac{1}{2}\Vert \id - \cR^{\pi} \Vert_{\diamond} = \frac{d^{2}-1}{d^{2}}$. 
\end{proof}
\begin{remark}
    Note that as a depolarizing channel generates an isotropic state when acting on half of a maximally entangled state, Propositions \ref{prop:diamond-dist-of-depol-channels} and \ref{prop:trace-dist-isotropic-states} imply $\frac{1}{2}\Vert \cD^{p} - \cD^{p'} \Vert_{\diamond}$ as expressed in \eqref{eq:diamond-norm-state-optimization} is optimized by the maximally entangled state. 
\end{remark}

\begin{proposition}\label{prop:diamond-dist-of-noisy-PVM}
    Let $\{\ket{\psi_{z}}\}_{z}$ be a complete set of orthonormal vectors. Consider the corresponding projective measurement $\cM_{A \to Z}(X) = \sum_{z} \bra{\psi_{z}} X \ket{\psi_{z}} \dyad{z}$ and define its convex combination with the replacer channel $\cM^{p} \coloneq (1-p)\cM + p\cR^{\pi}$ for $p \in [0,1]$. Then
    \begin{align}
        \frac{1}{2}\Vert \cM^{p} - \cM \Vert_{\diamond} = p \cdot (1-1/d) \ . 
    \end{align}
\end{proposition}
\begin{proof}
    First, we reduce the problem to the diamond distance between the perfect measurement and the replacer channel.
    \begin{align}
       \frac{1}{2}\Vert \cM^{p} - \cM \Vert_{\diamond} &\coloneq \sup_{\dyad{\psi}_{A'A} \in \Density(A' \otimes A)} \frac{1}{2} \Vert (\id_{A'} \otimes [(1-p)\cM + p\cR^{\pi} - \cM])(\dyad{\psi}) \Vert_{1} \\
        &= p \cdot \sup_{\dyad{\psi}_{A'A} \in \Density(A' \otimes A)} \frac{1}{2} \Vert (\id_{A'} \otimes [\cR^{\pi} - \cM])(\dyad{\psi}) \Vert_{1} \\
        &= p \cdot \frac{1}{2}\Vert \cR^{\pi} - \cM \Vert_{\diamond} \ .
    \end{align}
    where we just used the definitions and simplified terms. Now, we use weak duality of SDPs to determine the optimal value. By defining the Choi operator in the basis of $\{\ket{\psi_{z}}\}_{z}$, $\Gamma^{\cM} = \sum_{z} \dyad{\psi_{z}} \otimes \dyad{z}$. Thus, $\Gamma^{\cR}-\Gamma^{\cM} = \sum_{z} \dyad{\psi_{z}} \otimes (\dyad{z} - \pi)$. By direct calculations, 
    \begin{align}
        \Omega_{A'Z} \coloneq \dyad{\psi_{z}} \otimes \dyad{z} \quad \rho_{A'} \coloneq \dyad{\psi_{z}}
    \end{align}
    is feasible for the primal problem of $\frac{1}{2}\Vert \cR^{\pi} - \cM \Vert_{\diamond}$ according to \eqref{eq:diamond-dist-primal} and obtains a value $1-1/d$. Similarly, direct calculations will confirm
    \begin{align}
        Z_{A'A} \coloneq \sum_{z} \dyad{\psi_{z}} \otimes (1-1/d)\dyad{z} \quad \mu = 1 - 1/d
    \end{align}
    is feasible for the primal problem of $\frac{1}{2}\Vert \cR^{\pi} - \cM \Vert_{\diamond}$ according to \eqref{eq:diamond-dist-dual} and obtains a value $1-1/d$. As we have matching upper and lower bounds of $1-1/d$, by weak duality (see e.g. \cite[Proposition 2.27]{khatri-book}), we conclude the value is $1-1/d$.
\end{proof}

\section{Amplification of Advantage for More Senders}\label{app:amplification-of-advantage}
In this section we show how to amplify the advantage between EAQMNs and CMNs in Theorem \ref{thm:one-shot-info-advantage} by including more senders in MNs. Under the relevant conditions on the SMP MNs, this is an immediate consequence of our reduction of simulation to guessing probability (Lemma \ref{lem:reduction-to-guessing-probability}) and the multiplicativity of the guessing probability over independent pairs of states (Item 2 of Proposition \ref{prop:guessing-prob-properties}). For clarity of presentation, we first introduce the general type of SMP MN for which our amplification holds and use it to establish our generic amplification lemma (Lemma \ref{lem:gen-structural-amplification}), we then identify the relevant instances of SMP MNs and show how these give rise to our quantitative separation (Theorem \ref{thm:product-theorem}), and finally we show how this recovers the version of the theorem given in the main text.

We begin with the general structure we need for our decomposition lemma.
\begin{definition}[$\cP$-Independent SMP MNs]\label{eq:partition-independent-SMP-MN}
    Let $m \leq k \in \mbb{N}$, $\vec{d} \in \mbb{N}^{\times k}$. Let $\cP \coloneq (S_{1},S_{2},...,S_{m})$ be a partitioning of $[k]$ into $m$ pieces. For each sender $i \in \{1,...,m\}$, fix an input space $Q_{i}$. Let $\mbf{S}(\vec{d})$ be a set of $k$-sender SMP MNs (Definition \ref{def:SMP-MN}) with specified signaling dimension (Definition \ref{def:signaling-dimension}). Let $\msf{E}_{\text{tot}} \subseteq \Channel(\otimes_{i \in \{1,...,k\}} Q_{i}, \otimes_{i \in \{1,...,k\}} C_{i})$ denote the set of encoding-then-transmission maps possible given the input spaces and the structure of $\mbf{S}(\vec{d})$. We say $\mbf{S}(\vec{d})$ is a set of $\cP$-independent SMP MNs if there exist sets of channels $\{\msf{E}_{j} \subseteq \Channel(\otimes_{i \in S_{j}} Q_{i}, \otimes_{i \in S_{j}} C_{i})\}_{j \in \{1,...,m\}} $ such that
    \begin{align}\label{eq:partition-independent-total-channels}
        \msf{E}_{\text{tot}} = \{\otimes_{j \in [m]} \cE_{j} : \cE_{j} \in \msf{E}_{j} \ , \forall j \in \{1,...,m\}\} \ . 
    \end{align}
    That is, the set of possible encoding-then-transmission maps as determined by $\mbf{S}(\vec{d})$ decompose into channels that are independent according to the partition $\cP$.
\end{definition}
Before continuing, we make some remarks about the above definition. First, a simple example of a set of IC SMP MNs is $4$ senders, where sender $i$ is connected to the receiver by $\id_{d_{i}}$, senders $1$ and $2$ can share arbitrary entanglement with each other, senders $3$ and $4$ can share arbitrary entanglement with each other, and no other resources are available. This is a $\cP$-independent SMP MNs where $\cP = (\{1,2\},\{3,4\})$ because every encoding-then-transmission map decomposes into the tensor product of two encoding-then-transmission maps: one built using senders $1$ and $2$ and one built using senders $3$ and $4$. Second, the specification of the input space was important in the above definition as it allowed us to define the tensor product decomposition of the encoding-then-transmission maps, which was not specified in Definition \ref{def:SMP-MN}.

With the main definition introduced, we state the general decomposition lemma.
\begin{lemma}\label{lem:gen-structural-amplification}
    Let $m \leq k \in \mbb{N}$ and $\vec{d} \in \mbb{N}^{\times k}$. Let $\cP \coloneq (S_{1},S_{2},...,S_{m})$ be a partitioning of $[k]$ into $m$ pieces. For $j \in \{1,...,m\}$, let $\wt{\cW}_{j} \coloneq \bigtimes_{i \in S_{j}} \cW_{i}$ and $f_{j}: \wt{\cW}_{j} \to \cZ_{j}$ be a function. Define the product function $f:\cW^{k} \to \cZ^{m}$ by $f = \bigtimes_{j \in \{1,...,m\}} f_{j}$. Let $\mbf{S}(\vec{d})$ be a set of $\cP$-independent SMP MNs. For $j \in \{1,...,m\}$, define $\wt{C}_{j} \coloneq \otimes_{i \in S_{j}} C_{i}$. Then
    \begin{align}
        P_{S}^{f}(\mbf{S}(\vec{d})) = \prod_{j \in \{1,...,m\}} \sup_{\cE_{j} \in \msf{E}_{j}} p_{g}(Z_{j} \vert \wt{C}_{j})_{ }\ , 
    \end{align}
\end{lemma}
\begin{proof}
    First, using the definition of the product function and the uniformity and independence of the inputs according to Definition \ref{def:function-simulation-success-prob},
    \begin{align}
        \rho_{Z^{m}W^{k}} &= \frac{1}{\vert \cW^{k} \vert} \sum_{w^{k} \in \cW^{k}} \dyad{f(w^{k})} \otimes \dyad{w^{k}} \\
        &= \bigotimes_{j \in \{1,...,m\}} \left( \frac{1}{\vert \wt{\cW}_{j} \vert} \sum_{\wt{w}_{j} \in \wt{\cW}_{j}} \dyad{f_{j}(\wt{w}_{j})} \otimes \dyad{w^{k}} \right) \eqqcolon \bigotimes_{j \in \{1,...,m\}} \rho_{Z_{j}\wt{W}_{j}} \ . 
    \end{align}
    Next, using that $\mbf{S}(\vec{d})$ is a set of $\cP$-independent SMP MNs, for any encoding-then-transmission map $\cE^{\net} \in \msf{E}_{\text{tot}}$, it may be expressed as $\bigotimes_{j \in \{1,...,m\}} \cE_{j}$ for some $\{\cE_{j} \in \msf{E}_{j}\}_{j \in \{1,...,m\}}$. It follows from these identifications that
    \begin{align}
        \rho_{Z^{m}C^{k}} = (\id_{Z^{m}} \otimes \cE^{\net})(\rho_{Z^{m}W^{k}}) &= \left(\bigotimes_{j \in \{1,...,m\}} \id_{Z_{j}} \otimes \cE_{j} \right)\left(\bigotimes_{j \in \{1,...,m\}} \rho_{Z_{j}\wt{W}_{j}}\right) \\ 
        &=\otimes_{j \in \{1,...,m\}} (\id_{Z_{j}} \otimes \cE_{j})(\rho_{Z_{j}\wt{W}_{j}}) \ . \label{eq:independent-state-decomp}
    \end{align}
    Finally, starting from Lemma \ref{lem:reduction-to-guessing-probability} with the relabelings $Z \to Z^{m}$ and $C \to C^{k}$,
    \begin{align}
        P_{S}^{f}(\mbf{S}(\vec{d})) &= \sup_{\cE^{\net} \in \msf{E}_{\text{tot}}} p_{g}(Z^{m} \vert C^{k})_{(\id_{Z} \otimes \cE^{\net}_{W^{k} \to C^{k}})(\rho_{ZW^{k}})} \\
        &= \sup_{\cE_{j} \in \msf{E}_{j} \forall j} p_{g}(Z^{m} \vert C^{k})_{\otimes_{j} (\id_{Z_{j}} \otimes \cE_{j})(\rho_{Z_{j}\wt{W}_{j}})} \\
        &= \sup_{\cE_{j} \in \msf{E}_{j} \forall j} \prod_{j \in \{1,...,m\}} p_{g}(Z_{j} \vert \wt{C}_{j})_{(\id_{Z_{j}} \otimes \cE_{j})(\rho_{Z_{j}\wt{W}_{j}})} \\
        &= \prod_{j \in \{1,...,m\}} \sup_{\cE_{j} \in \msf{E}_{j}} p_{g}(Z_{j} \vert \wt{C}_{j})_{(\id_{Z_{j}} \otimes \cE_{j})(\rho_{Z_{j}\wt{W}_{j}})} \ , 
    \end{align}
    where the second equality is \eqref{eq:independent-state-decomp}, the third equality is Item 2 of Proposition \ref{prop:guessing-prob-properties}, and the fourth is \eqref{eq:partition-independent-total-channels}. This completes the proof.
\end{proof}

The goal of the remainder of this section is to apply the previous lemma to show we can amplify the gap between a classical network and an entanglement-assisted quantum network in Theorem \ref{thm:one-shot-info-advantage} by considering an appropriate product function and increasing the number of senders. Note that EACMNs and FCNSMNs as defined in Definition \ref{def:main-families-of-MNs} will not be $\cP$-Independent SMP MNs for $k \geq 3$ and any partitioning $\cP$ as the shared entanglement or total non-signaling box does not allow the total encoding-and-transmission map to decompose over tensor products. Therefore, to get the strongest amplification result possible with the results of this work, we introduce a more restricted form of assistance for more senders.
\begin{definition}[Pairwise-Assisted MNs]\label{def:pairwise-assisted-MNs}
    Let $k \in \mbb{N}$ and $\vec{d} \in \mbb{N}^{\times 2k}$.
    \begin{enumerate}
        \item $\mbf{C}^{P-E}(\vec{d})$ denotes the same setting as $\mbf{C}(\vec{d})$ except for all $i \in \{1,...,k\}$, sender $2i-1$ and $2i$ may share arbitrary entanglement. We call these \textbf{pairwise-entanglement-assisted classical MNs, (P-EACMNs).} 
        \item $\mbf{Q}^{P-E}(\vec{d})$ denotes the same setting as $\mbf{Q}(\vec{d})$ except for all $i \in \{1,...,k\}$, sender $2i-1$ and $2i$ may share arbitrary entanglement. We call these \textbf{pairwise-entanglement-assisted quantum MNs, (P-EAQMNs).} 
        \item $\mbf{C}^{P-N}(\vec{d})$ denotes the set of MNs where for all $i \in \{1,...,k\}$, sender $2i-1$ and $2i$ share a fully classical non-signaling box. We call these \textbf{pairwise fully classical non-signaling MNS (P-FCNSMNs)}. 
    \end{enumerate} 
\end{definition}

\begin{theorem}[Product Theorem]\label{thm:product-theorem}
    Let $k \in \mbb{N}$, $\vec{d} \in \mbb{N}^{2k}$, and $\{f_{i}: \cW_{2\cdot i-1} \times \cW_{2 \cdot i} \to \cZ_{i}\}_{i \in \{1,...,k\}}$ be a set of functions. Define the product function $f: \times_{i \in \{1,...,k\}} f_{i}$. Then
    \begin{align}
        P_{S}^{f}(\mbf{C}(\vec{d})) &= \prod_{i \in \{1,...,k\}} P_{S}^{f_{i}}(\mbf{C}(d_{2\cdot i -1},d_{2 \cdot i})) \\
        P_{S}^{f}(\mbf{C}^{P-E}(\vec{d})) &= \prod_{i \in \{1,...,k\}} P_{S}^{f_{i}}(\mbf{C}^{E}(d_{2\cdot i -1},d_{2 \cdot i})) \\
        P_{S}^{f}(\mbf{C}^{P-N}(\vec{d})) &= \prod_{i \in \{1,...,k\}} P_{S}^{f_{i}}(\mbf{C}^{N}(d_{2\cdot i -1},d_{2 \cdot i})) \ . 
    \end{align}
\end{theorem}
\begin{proof}
    \sloppy In all three cases, the encoding-then-transmission channels decompose over the partition $\cP = (\{1,2\},\{2,3\},...,\{2k-1,k\})$ in such a manner as to induce the relevant encoding-then-transmission channels for the subfunctions and the bipartite simulation setting. To see this:
    \begin{itemize}
        \item For CMNs, this is because the encoding-then-transmission channel is of the form 
        $$\otimes_{i \in \{1,...,2k\}} (\Delta_{d_{i}} \circ \cE_{\cW_{i} \to \mbb{C}^{d_{i}}}) \ , $$
        where $\cE_{\cW_{i} \to \mbb{C}^{d_{i}}}$ is party $i$'s encoding map, which we remark without loss of generality is a classical-to-classical channel.
        \item For P-EACMNs, this is because the encoding-then-transmission channel is of the form 
        $$ \otimes_{i \in \{1,...,k\}} (\Delta_{d_{2i-1}} \otimes \Delta_{d_{2i}}) \circ ((\cM_{W_{2i-1}Q_{2i-1} \to \mbb{C}^{d_{2i-1}}} \otimes \cM_{W_{2_{i}}Q_{2_{i}} \to \mbb{C}^{d_{i}}})(\sigma_{Q_{2i-1}Q_{2i}})) \ ,  $$
        where $\cM_{W_{i}Q_{i} \to \mbb{C}^{d_{i}}}$ denotes party $i$'s measurement channel and $\sigma_{Q_{2i-1}Q_{2i}}$ is the state shared by parties $2i-1$ and $2i$ for $i \in \{1,...,k\}$.
        \item For P-FCNSMNs, this is because the encoding-then-transmission channel is of the form 
        \begin{align}
            \otimes_{i \in \{1,...,k\}} N_{W_{2i-1}W_{2i} \to C_{2i-1}C_{2i}} \ ,
        \end{align}
        where $N_{W_{2i-1}W_{2i} \to C_{2i-1}C_{2i}}$ is a fully classical non-signaling box between parties $2i-1$ and $2i$ where this is without loss of generality as non-signalling boxes are closed under local pre-processing so we have absorbed the parties' local encodings into the definition of the box.
    \end{itemize}
    Thus, the results follow from Lemma \ref{lem:gen-structural-amplification}.
\end{proof}

One can of course apply the above theorem in generality for conditionally bijective functions by combining it with the bounds in Theorems \ref{thm:unassisted-QMN-bound} and \ref{thm:non-signaling-box-bound}. Here, we present a slight abstraction of Theorem \ref{thm:main-text-amplification-theorem}. We provide the same result parameterized in terms of the dimension rather than the log of the dimension and in terms of arbitrary conditionally bijective function for consistency with the supplementary material. Theorem \ref{thm:main-text-amplification-theorem} thus follows by taking the log of the dimension and setting each function to $\oplus^{2}_{2^{n}}$.
\begin{corollary}
    Let $n,k \in \mbb{N}$. Let $\vec{d} = 2^{n} \cdot \vec{1}$. For all $i \in \{1,...,2k\}$, let $\vert \cW_{i} \vert = 2^{2n}$ , and $f_{i}: \cW_{2i-1} \times \cW_{2i} \to \cZ_{i}$ be a DCB defined by the multiplication of a TNC group of order $2^{2n}$ (Recall Definition \ref{def:TNC}). Define the product function $f \coloneq \bigtimes_{i \in \{1,...,k\}} f_{i}$. Then
    \begin{align}
        1 = P_{S}^{f}(\mbf{Q}^{P-E}(\vec{d})) > \frac{1}{2^{n \cdot k}} \geq  \max\{P_{S}^{f}(\mbf{Q}(\vec{d})),P_{S}^{f}(\mbf{C}^{P-N}(\vec{d}))\} \geq P_{S}^{f}(\mbf{C}^{P-E}(\vec{d})) \geq P_{S}^{f}(\mbf{C}(\vec{d})) \ . 
    \end{align}
    In other words, there is an exponential suppression in the success probability of simulating the function in the number of senders.
\end{corollary}
\begin{proof} We explain the entanglement-assisted quantum achievability and then the limitations of the other network setups. \\

\noindent \textit{Achievability}: By the definition of $\vec{d}$, each sender shares a $\id_{2^{n}}$ channel with the receiver. By the definition of pairwise-entanglement-assisted QMNs, each pair of senders $2i-1$ and $2i$ can share arbitrary entanglement. As such, as follows from Theorem \ref{thm:dense-network-coding}, each appropriate pair of senders can use the encoding given in Protocol \ref{prot:dense-network-coding-of-some-type} and then the receiver can use the appropriate generalized Bell state measurement for each pair of signals independently to correctly determine the value $Z_{i}$. In this way, the receiver correctly guesses the value of $Z_{i}$ for all $i \in \{1,...,k\}$. Thus, $P_{S}^{f}(\mbf{Q}^{P-E}(\vec{d}))=1$. \\

\noindent \textit{Limitations of Other Setups:} By applying Theorem \ref{thm:product-theorem} followed by Theorem \ref{thm:one-shot-info-advantage} and using that our choices of dimension imply for all $i \in \{1,...,k\}$ $d_{2i-1} = d_{2i} = 2^{n}$, so $\frac{\min\{d_{2i-1},d_{2i}\}}{\vert \cW_{i} \vert} = \frac{1}{2^{n}}$ for all $i \in \{1,...,k\}$ completes the proof.
\end{proof}

\section{Measurement-Device-Independent Secret Key Growing}\label{app:MDI-QKG}
In this section we formalize the claims made in Section \ref{sec:MDI-QKG} of the main text.

\subsection{Formalization of Private Distributed Computation Claim}
 In the main text, we observe that one can use a trusted, ideal 2-sender EAQMN to compute $\oplus_{d}^{2}$ in such a manner that it is private from the receiver (Observation \ref{obs:private-distributed-computation}). To formalize this, we first define notions of private distributed computation.\footnote{For simplicity, we focus on 2-input functions, but the ideas straightforwardly generalize.}
\begin{definition}\label{def:priv-against-outsiders}
    Let $f:\cX \times \cY \to \cZ$ be a function. Let $\wt{\cE} \in \Channel(X \otimes Y, \hat{Z} \otimes E)$ be a classical-classical to classical-quantum channel. For input joint distribution $p_{XY}$, define the joint output state
    \begin{align*}
        \rho_{X'Y'Z\hat{Z}E} \coloneq \left( \id_{X'Y'Z} \otimes \wt{\cE} \right)\left(\sum_{x,y} p_{XY}(x,y) \dyad{x}_{X'} \otimes \dyad{y}_{Y'} \otimes \dyad{f(x,y)}_{Z} \otimes \dyad{x}_{X} \otimes \dyad{y}_{Y} \right) \ . 
    \end{align*}
    Let $\ve,\delta \in [0,1]$. We say the channel $\cE$ is a $(1-\delta)$-correct distributed computation of $f$ on input $p_{XY}$ if 
    \begin{align}
        \Pr_{\rho}[Z \neq \hat{Z}] \leq \delta \ . 
    \end{align}
    We say $\wt{\cE}$ is $\ve$-private-against-outsiders on input $p_{XY}$ if
    \begin{align}\label{eq:priv-against-outsider}
        \text{TD}(p_{X'Y'\hat{Z}E}, p_{X'Y'\hat{Z}} \otimes \rho_{E}) \leq \ve \ . 
    \end{align}
    We say $\wt{\cE}$ is a $(1-\delta)$-correct distributed computation of $f$ (resp.~$\ve$-private-against-outsiders) if the property holds for all input $p_{XY}$.
\end{definition}
%Justification of definitions
We remark \eqref{eq:priv-against-outsider} guarantees approximate privacy against an outsider, because if the adversary holds $E$, their side-information (the register $E$) is approximately independent of the inputs and output of the function. Note that as we define the channel $\cE$ as generating the register $E$, this definition presumes the inputs are initially private from the outsiders. 

If the receiver is also not trusted, then we provide the following definition, which captures both the inevitability that function $f$ may leak information about the inputs to the receiver and that the receiver need not use the expected measurement.
\begin{definition}\label{def:priv-against-receiver}
    Let $f:\cX \times \cY \to \cZ$ be a function. Let $\omega_{A'B'E}$ be a shared state, $\cE_{XA' \to A}$, $\cF_{YB' \to B}$ be encoding channels, and $\cN_{AB \to AB}$ be a noise channel. For input joint distribution $p_{XY}$, define the joint state
    \begin{align*}
        \rho_{X'Y'ZABE} \coloneq \left( \id_{X'Y'Z} \otimes \cN \circ \cE \otimes \cF \right)\left( p_{X'Y'ZXY} \otimes \omega \right) \ , 
    \end{align*}
    where $\sum_{x,y} p_{XY}(x,y) \dyad{x}_{X'} \otimes \dyad{y}_{Y'} \otimes \dyad{f(x,y)}_{Z} \otimes \dyad{x}_{X} \otimes \dyad{y}_{Y}$. We say $(\cE,\cF,\cN,\omega)$ is $\ve$-private-against-receiver on input $p_{XY}$ if
    \begin{align}\label{eq:priv-against-receiver}
         p_{g}(X'Y' \vert ABE)_{\rho} - p_{g}(X'Y' \vert Z)_{\rho} \leq \ve \ . 
    \end{align}
    We say the channel $(\cE,\cF,\cN,\omega)$ is $\ve$-private-against-receiver if it is  $\ve$-private-against-receiver on all input $p_{XY}$.
\end{definition}

We now give the formal variation of Observation \ref{obs:private-distributed-computation}.
\begin{proposition}
    Let $(G,\cdot)$ be a TNC group (Definition \ref{def:TNC}) and $f: \cX \times \cY \to \cZ$ be the function defined by the multiplication table of $(G,\cdot)$. Let $(\cE,\cF,\id \otimes \id,\dyad{\Phi})$ denote the channels and resource state defining the encoding and transmission in Protocol \ref{prot:dense-network-coding-of-some-type}. Let $\wt{\cE}_{XY \to ZE}$ be the channel induced by the receiver applying the honest measurement. Then $(\cE,\cF,\id \otimes \id,\dyad{\Phi})$ is $0$-private-against-receiver, $\wt{\cE}_{XY \to ZE}$ is $0$-private-against-outsiders, and $\wt{\cE}_{XY \to ZE}$ is $1$-correct.
\end{proposition}
\begin{proof}
    We let the encoding channels $\cE,\cF$ be as implicitly defined in Protocol \ref{prot:dense-network-coding-of-some-type} so that in the notation of Definition \ref{def:priv-against-receiver} we are considering $(\cE,\cF,\id \otimes \id, \dyad{\Phi})$. In the case that the honest measurement $\cM^{G}$ (i.e. the one specified in Protocol \ref{prot:dense-network-coding-of-some-type}) is implemented, we can define the total map
    \begin{align}\label{eq:honest-map-for-private-comp}
        \wt{\cE}_{XY \to ZE} \coloneq (\cM^{G} \circ \cE \otimes \cF)(p_{X'Y'ZXY} \otimes \dyad{\Phi}) \ .
    \end{align} 
    The $1$-correctness of $\wt{\cE}$ then follows from Theorem \ref{thm:dense-network-coding}.
    
    We now prove privacy-against-receiver. As the shared state $\ket{\Phi^{+}}_{AB}$ is pure, it's purification is $\ket{\Phi^{+}}_{AB} \otimes \ket{\psi}_{E}$ for some state $\ket{\psi}_{E}$. The noiseless identity channels preserve this, i.e. after the transmission of the encoded state, the state is $\dyad{\Phi_{g \cdot h}}_{AB} \otimes \dyad{\psi}_{E}$ as follows from \eqref{eq:receivers-received-state}. Using that $\vert \cX \vert = \vert \cY \vert = \vert \cZ \vert = \vert G \vert$, it follows the joint state after encoding and transmission is
    \begin{align}
        \rho_{X'Y'ZABE} = \sum_{g,h \in G} p_{XY}(g,h) \dyad{g}_{X'} \otimes \dyad{h}_{Y'} \otimes \dyad{g\cdot h}_{Z} \otimes \dyad{\Phi_{g\cdot h}}_{AB} \otimes \dyad{\psi}_{E} \ . 
    \end{align}
    Note that $V \coloneq \sum_{g \in G} \ket{\Phi_{g}}\dyad{g}$ is an isometry by the orthornormality of the states (Proposition \ref{prop:TNC-coding}) so that
    \begin{align}
        \rho_{X'Y'ZABE} = \sum_{g,h \in G} p_{XY}(g,h) \dyad{g}_{X'} \otimes \dyad{h}_{Y'} \otimes V\dyad{g\cdot h}_{Z} V^{\dagger} \otimes \dyad{\psi}_{E} \ . 
    \end{align}
    It follows
    \begin{align}\label{eq:proof-of-priv-against-rec}
        p_{g}(X'Y'\vert ABE)_{\rho} = p_{g}(X'Y' \vert ZE)_{\rho} = p_{g}(X'Y' \vert Z)_{\rho} \ ,
    \end{align}
    where the first equality uses \eqref{eq:guess-prob-and-min-ent} and the isometric invariance of min-entropy \cite{Tomamichel-Book} and the second equality uses the independence of $E$ and Item 2 of Proposition \ref{prop:guessing-prob-properties}. \eqref{eq:proof-of-priv-against-rec} proves $(\cE,\cF,\id \otimes \id, \dyad{\Phi})$ is $0$-private-against-receiver by Definition \ref{def:priv-against-receiver} and that we were agnostic about the distribution $p_{XY}$.

    Finally, when the honest measurement $\cM^{G}$ is applied,
    \begin{align}
        \rho_{X'Y'\hat{Z}E} = \sum_{g,h \in G} p_{XY}(g,h) \dyad{g}_{X'} \otimes \dyad{h}_{Y'} \otimes \dyad{g\cdot h}_{\hat{Z}} \otimes \dyad{\psi}_{E} \ ,
    \end{align}
    so $\rho_{X'Y'\hat{Z}E} = p_{X'Y'\hat{Z}} \otimes \rho_{E}$ for any input distribution $p_{XY}$, and thus $\wt{\cE}_{XY \to ZE}$ is $0$-private-against-outsiders by Definition \ref{def:priv-against-outsiders}. 
\end{proof}
\begin{remark} The amount of information $f$ leaks about the inputs is dependent on the distribution over the inputs. An example of this is in fact formalized in Proposition \ref{prop:DNC-does-not-transmit-inputs}. Because of this dependence on the input distribution, we do not introduce a notion of how much information about the inputs a function $f$ leaks.
\end{remark}

\subsection{Security of MDI Secret Key Growing Scheme}
We now prove the security of the MDI secret key growing scheme presented in the main text. We introduce the definition of security, then prove a one-shot secure key length bound, then use this to derive the asymptotic rate given in the main text, and finally present further details on the example plotted in the main text.

\subsubsection{Security Definition and Preliminaries} While the MDI quantum key growing scheme is presented in the main text (Protocol \ref{prot:main-text-MDI-Q-Key-Scheme}), here we present a slightly more detailed version that makes assumptions more explicit.
\begin{protocoldesc}[H]
\caption{MDI Quantum Key Growing Scheme}\label{prot:app-MDI-Q-Key-Scheme}
		\textbf{Inputs:} \\
		\hspace{0.5cm}
		\begin{tabular}{ l l}
			$n \in \mathbb{N}$ & Number of rounds  \\
            $\rho^{0}_{A^{n}B^{n}}$ & Initial Quantum State \\
            $r$ & Bits of the error correction (EC) syndrome \\
            $t$ & Bits of EC verification 2-universal hash
		\end{tabular}
		
		\vspace{0.5cm}
		
		\textbf{Protocol:} 
		\begin{enumerate}
            \item \textbf{Growing Phase:}
            \begin{enumerate}[label=\alph*.]
			     \item For $i \in \{1,...,n\}$, Alice draws uniformly random bits $(x_{2i-1},x_{2i}) \leftarrow \{0,1\}^{\times 2}$, applies qubit discrete Weyl operator $W_{x_{2i-1},x_{2i}}$ to the $A_{i}$ system, and sends it to the receiver over an insecure quantum channel.
                \item For $i \in \{1,...,n\}$, Bob draws uniformly random bits $(y_{2i-1},y_{2i}) \leftarrow \{0,1\}^{\times 2}$, applies qubit discrete Weyl operator $W_{y_{2i-1},y_{2i}}$ to the $B_{i}$ system, and sends it to the receiver over an insecure quantum channel.
                \item Either Alice receives $\hat{\mbf{z}} \in \{0,1\}^{\times 2n}$ over an insecure classical channel of pre-defined computational basis or she aborts. She computes the bitwise XOR of what she received and her local randomness, $\mbf{w} := \hat{\mbf{z}} \oplus^{2n}_{n} \mbf{x}$.
            \end{enumerate}
            \item \textbf{Post-Processing Phase:} Using an authenticated classical channel shared by Alice and Bob,
            \begin{enumerate}[label=\alph*.]
                \item \textbf{Error Correction:} Alice and Bob apply some error correcting code so that Bob now holds $\hat{\mbf{w}} \in \{0,1\}^{2n}$.
                \item \textbf{Error Verification:} Alice and Bob apply a function from a 2-universal hash function to their $\mbf{w}$ and $\hat{\mbf{w}}$ respectively. Bob communicates the output to Alice. If $f(\mbf{w}) \neq f(\hat{\mbf{w}})$ , they abort.
                \item \textbf{Privacy Amplification:} Alice chooses $f$ at random from a two-universal hash function family $\cF$ and announces it to Bob. Alice and Bob apply $f$ to $\mbf{w}$ and $\hat{\mbf{w}}$ respectively, resulting in $K_{A}, K_{B}$ respectively, which is their key.
            \end{enumerate}
		\end{enumerate}
	\end{protocoldesc}

We define the security conditions in the following manner, which are identical to the standard security definitions in quantum key distribution (c.f.~\cite{Portmann_2022}).
\begin{definition}\label{def:MDI-Q-Key-Scheme-security} Let $\ve_{s},\ve_{c} \in [0,1]$. Let $\Omega$ denote the event that the protocol does not abort. 
\begin{enumerate}
    \item Protocol \ref{prot:main-text-MDI-Q-Key-Scheme} is $\ve_{s}$-sound on input $\rho^{0}_{A^{n}_{1}B^{n}_{1}}$ if the output of the protocol satisfies
    \begin{align}\label{eq:MDI-grow-soundness}
        \frac{1}{2}\Pr[\Omega]\Vert \rho_{K_{A}K_{B}E'}^{\Omega} - \chi_{K_{A}K_{B}} \otimes \rho^{\Omega}_{E'} \Vert_{1} \leq \ve_{s} \ ,
    \end{align}
    where $\chi_{K_{A}K_{B}} \coloneq \frac{1}{\vert \cK \vert}  \sum_{k \in \cK} \dyad{k} \otimes \dyad{k}$ is the perfectly correlated state, $\rho^{\Omega}$ denotes the state conditioned on not aborting, and $E'$ denotes all the side-information the eavesdropper obtains during the protocol, including starting with the purifying space of $\rho^{0}$.
    \item We say Protocol \ref{prot:main-text-MDI-Q-Key-Scheme} is $\ve_{c}$-complete on $\rho^{0}_{A^{n}B^{n}E_{0}}$ and honest measurement $\cM_{AB \to Z}$ if the probability of aborting is at most $1-\ve_{c}$ when the receiver performs $\cM_{AB \to Z}$ per round.\footnote{Specifying the map $\cM$ allows one to take into account there may be noise in the actual honest implementation.}
\end{enumerate}
\end{definition}
\noindent As these definitions are no different from standard quantum key distribution, we refer the reader to \cite{Portmann_2022} for the justification of these security definitions.

Next, we recall some further standard definitions. Further information on these quantities may be found in \cite{Tomamichel-Book}. For subnormalized states $\rho,\sigma \in \Density_{\leq}(A)$, the purified distance is
\begin{align}
  P(\rho,\sigma) \coloneq \sqrt{1-F_{\ast}(\rho,\sigma)} \ ,   
\end{align}
where 
\begin{align*}
    F_{\ast}(\rho,\sigma)  \coloneq \left( \Tr[\sqrt{\sqrt{\rho}\sigma\sqrt{\rho}}] + \sqrt{(1-\Tr[\rho])(1-\Tr[\sigma])} \right)^{2} \ 
\end{align*}
is the generalized fidelity.

For subnormalized state $\rho_{AB} \in \Density_{\leq}(AB)$ and parameter $\ve \in (0,\Tr[\rho])$, the (purified-distance-smoothed) smooth min-entropy of $\rho_{AB}$ is 
\begin{align}\label{eq:purified-smooth-distance-smoothed-min-entropy}
    H^{\ve}_{\min}(A \vert B)_{\rho} \coloneq \max_{\wt{\rho} \in \Density_{\leq}(AB): P(\wt{\rho},\rho)\leq \ve} H_{\min}(A \vert B)_{\wt{\rho}} \ . 
\end{align}
For our purposes, the reason the smooth min-entropy is useful is it roughly captures the ability to extract private randomness from $\rho_{XE}$ using 2-universal hash functions (a.k.a.~the leftover hashing lemma) as the following captures.
\begin{proposition}(Straightforward Corollary of \cite[Proposition 9]{Tomamichel-2017a})\label{prop:LHL}
    Let $\rho_{XE} \in \Density_{\leq}(XE)$. Let $\ve_{\text{sec}} \in (0,1)$ and $\ve \in (0, \min\{\sqrt{\Tr(\rho)}, \ve_{\text{sec}} \})$. Then, using a family of two-universal hash functions from $\cX \to \{0,1\}^{\ell} \eqqcolon Z$ that are chosen from uniformly, it is the case
    \begin{align}\label{eq:decoupled-sec}
        \frac{1}{2}\Vert \rho_{f(X)FE} - \pi_{Z} \otimes \rho_{FE} \Vert_{1} \leq \ve_{\text{sec}} \ , 
    \end{align}
    so long as
    \begin{align}
        H^{\ve}_{\min}(X|E)_{\rho} - 2\log(\frac{1}{4(\ve_{\text{sec}}-\ve)}) \geq \ell \ .
    \end{align}
\end{proposition}

\subsubsection{One-Shot MDI QKG Security Proof}
We now provide the security proof for one-shot MDI QKG. We remark that other than some of the reasoning of the steps in the protocol and using Lemma \ref{lem:cond-Y-func-entropy}, the proof uses standard tools in proving the security of QKD (see e.g.~\cite{Tomamichel-2017a}).
\begin{lemma}\label{lem:one-shot-bound}
    Consider Protocol \ref{prot:app-MDI-Q-Key-Scheme} with input $\rho^{0}_{A^{n}B^{n}}$. Let $\ve_{\text{sec}} \in (0,1)$, $\ve \in (0,\ve_{\text{sec}})$. One, obtains a $(2^{-t}+\sqrt{\ve_{\text{sec}}})$-sound key so long as the protocol is defined to hash to a length of key $\ell$ where
    \begin{equation}
    \begin{aligned}
        \ell & \leq H^{\ve}_{\min}(X^{2n}_{1}\vert A^{n}_{1}B^{n}_{1}E_{0})_{\rho^{1}} - t - r - 2\log(\frac{1}{4(\ve_{\text{sec}}-\ve)})\ ,
    \end{aligned}
    \end{equation}
    where $E_{0}$ is an arbitrary purifying system of $\rho^{0}_{A^{n}B{n}}$ and $\rho^{1}$ is the state after Alice and Bob apply their encoding to $\rho^{0}$ (see \eqref{eq:rho-1-prot-process}).
\end{lemma}
\begin{proof}
    By the standard arguments for QKD, e.g. \cite[Lemma 1]{Tomamichel-2017a}, to guarantee \eqref{eq:MDI-grow-soundness}, it suffices to guarantee that the probability that the keys do not agree and error correction and error verification fail is bounded by $\ve_{\text{ec}}$ and that \eqref{eq:decoupled-sec} is satisfied when the event that error verification passes holds. That $\ve_{\text{ec}} \leq 2^{-t}$ follows from the choice of error verification via 2-universal hash functions in Protocol \ref{prot:app-MDI-Q-Key-Scheme} and \cite[Theorem 2]{Tomamichel-2017a}. Thus, we focus on guaranteeing \eqref{eq:decoupled-sec}. To this end it suffices to bound the smooth min-entropy of the bitstring to which Alice applies privacy amplification in Protocol \ref{prot:app-MDI-Q-Key-Scheme}, $W^{2n} \coloneq \hat{Z} \oplus^{2n}_{2} X^{2n}_{1}$, conditioned on all information available to Eve the end of Step 2.b. of Protocol \ref{prot:main-text-MDI-Q-Key-Scheme} and joint with error correction and error verification passing. This is because then we may apply Proposition \ref{prop:LHL} to said state.

    For clarity, we begin with an explanation of the global state as a function of the quantum channels applied during Protocol \ref{prot:app-MDI-Q-Key-Scheme}. In Steps 1.a. and 1.b. of Protocol \ref{prot:app-MDI-Q-Key-Scheme}, Alice and Bob each generate local randomness and use this randomness to encode their local systems using qubit discrete Weyl operators. This process can be represented by the map $\cE^{1} = \cE^{\otimes n}$ where $\cE$ is defined via Kraus operators $\{W_{x_{1},x_{2}} \otimes W_{y_{1},y_{2}} \otimes \dyad{x_{1},x_{2}} \otimes \dyad{y_{1},y_{2}}\}_{(x_{1},x_{2}), (y_{1},y_{2}) \in \{0,1\}^{\times 2}}$. This results in the the global encoded state $\rho^{1}$ being
    \begin{align}
        \rho^{1}_{X^{2n}_{1}Y^{2n}_{1}A^{n}_{1}B^{n}_{1}E_{0}} \coloneq \cE^{1}(\pi_{X^{2n}_{1}} \otimes \pi_{Y^{2n}_{1}} \otimes \rho^{0}_{A^{n}_{1}B^{n}_{1}E_{0}}) \label{eq:rho-1-prot-process} \ .
    \end{align}
    Without loss of generality, once the $A^{n}_{1}B^{n}_{1}$ systems are released to the adversary, the adversary performs an isometric channel $\cV_{A^{n}_{1}B^{n}_{1}E_{0} \to Z^{2n}_{1}E}$, so $\rho^{2}_{X^{2n}_{1}Y^{2n}_{1}Z^{2n}_{1}E} := \cV(\rho^{1})$.\footnote{This is without loss of generality as Alice will abort if she does not receive a $\hat{z}$ and any quantum channel that generates a $Z^{2n}$ register can be expressed as such an isometry by letting all other information be held by Eve.} The adversary then sends the $Z^{2n}_{1}$ portion of the resulting state to Alice over a classical channel that is insecure from Alice's perspective. This classical channel completely dephases the $Z^{2n}_{1}$ register in the pre-specified computational basis, so the global state is $\rho^{3}_{X^{2n}_{1}Y^{2n}_{1}\hat{Z}^{2n}_{1}E} = \Delta^{\otimes 2n}_{2}(\rho^{2})$. To make clear the register has changed, we label the outcoming registers by $\hat{Z}_{i}$ rather than $Z_{i}$. Alice then computes $\mbf{w}$ which is represented by a channel $\cF$ so that $\rho^{4}_{X^{2n}_{1}Y^{2n}_{1}\hat{Z}^{2n}_{1}W^{2n}_{1}F_{EC}E} = \cF_{X^{2n}_{1}\hat{Z}^{2n}_{1} \to X^{2n}_{1}W^{2n}_{1}\hat{Z}^{2n}_{1}}(\rho^{3})$. Note this channel does not alter the $X^{2n}_{1}$ and $\hat{Z}^{2n}_{1}$ registers. Finally, Alice and Bob perform error correction and error verification, which results in a transcript register $C$ and a random variable $F_{EC} \in \{\checkmark,\times\}$ which stores if error correction failed (denoted $\times$) or not (denoted $\checkmark$). In total $\rho^{5}_{X^{2n}_{1}Y^{2n}_{1}\hat{Z}^{2n}_{1}W^{2n}_{1}\hat{W}^{2n}_{1}CF_{EC}E}  = \cE^{\text{post-proc}}_{W^{2n}_{1}Y^{2n}_{1} \to W^{2n}_{1}\hat{W}^{2n}_{1}CF_{EC}}(\rho^{4})$. Note that Alice's register $W^{2n}_{1}$ does not change under error correction. 
    
    With this account of the joint state, we can bound the smooth min-entropy of the state conditioned on passing error correction and verification, i.e. the conditional state $\rho^{5}_{X^{2n}_{1}Y^{2n}_{1}\hat{Z}^{2n}_{1}W^{2n}_{1}\hat{W}^{2n}_{1}CE \wedge F_{EC} = \checkmark}$ where we are using the notation for the joint state with an event from \cite{Tomamichel-2017a,tan2024prospectsdeviceindependentquantumkey}. If $\Pr_{\rho}[F_{EC} = \checkmark] \leq \sqrt{\ve_{\sec}}$, then if privacy amplification is applied to $\rho^{5}_{\wedge F_{EC} = \checkmark}$, the resulting state will satisfy \eqref{eq:MDI-grow-soundness}. This is simply because $\Pr[\Omega] = \Pr[F_{EC} = \checkmark]$, so under these assumptions the soundness is bounded by $\Pr[\Omega] < \sqrt{\ve_{\sec}} < 2^{-t} + \sqrt{\ve_{\sec}}$ as promised. Thus, we only need to analyze the case when $\Pr_{\rho}[F_{EC} = \checkmark] \geq \sqrt{\ve_{\sec}}$. Importantly, this means $\sqrt{\Pr_{\rho}[F_{EC} = \checkmark]} \geq \sqrt{\ve_{\sec}} > \ve$, so that we can apply \cite[Lemma 10]{Tomamichel-2017a}. With these considerations handled, we have
    \begin{align}
        H^{\ve}_{\min}(W^{2n}_{1} \vert CE)_{\rho^{5}_{\wedge F_{EC} = \checkmark} }
        &\geq H^{\ve}_{\min}(W^{2n}_{1} \vert CE)_{\rho^{5}} \\ 
        &\geq 
        H^{\ve}_{\min}(W^{2n}_{1} \vert \hat{Z}^{2n}_{1}CE)_{\rho^{5}} \\
        &\geq H^{\ve}_{\min}(W^{2n}_{1} \vert \hat{Z}^{2n}_{1} E)_{\rho^{4}} - \log(C) \\ 
        &= H^{\ve}_{\min}(X^{2n}_{1} \oplus^{2n}_{2} \hat{Z}^{2n}_{1} \vert \hat{Z}^{2n}_{1} E)_{\rho^{4}} - \log(C) \\
        &= H^{\ve}_{\min}(X^{2n}_{1} \vert \hat{Z}^{2n}_{1} E)_{\rho^{4}} - \log(C) \\
        &= H^{\ve}_{\min}(X^{2n}_{1} \vert \hat{Z}^{2n}_{1} E)_{\rho^{3}} - \log(C) \\
        &\geq H^{\ve}_{\min}(X^{2n}_{1} \vert Z^{2n}_{1} E)_{\rho^{2}} - \log(C) \\
        &= H^{\ve}_{\min}(X^{2n}_{1} \vert A^{n}_{1}B^{n}_{1}E_{0})_{\rho^{1}} - \log(C) \ ,
    \end{align}
     where the first inequality is \cite[Lemma 10]{Tomamichel-2017a}, the second is data-processing \cite[Theorem 6.19]{Tomamichel-Book}, the third is a chain rule for classical registers \cite[Lemma 6.18]{Tomamichel-Book}, the first equality is the definition of $\mbf{w}$, which applies as Alice does not change her string under error correction, the second equality is Lemma \ref{lem:cond-Y-func-entropy} as the bitwise XOR is a $Y$-conditionally bijective function with $\hat{Z}$ acting as the $Y$ register, the third equality is because $\cF$ does not alter the registers being evaluated in the min-entropy, the fourth inequality is the data-processing of the classical channel measuring the $Z$ register, and the final equality is invariance under an isometry acting on the conditioning systems. We can then replace the error correction transcript via $\log(C) = r + t$ as $r$ bits are announced in error correction and $t$ bits are announced in the verification hash. Finally, one applies Proposition \ref{prop:LHL} to complete the proof.
\end{proof}

\subsubsection{Asymptotic Rate of MDI QKG} 
With the one-shot rate established, we establish the asymptotic rate. This is the formal version of Theorem \ref{thm:main-text-MDI-QKG} in the main text.
\begin{theorem}\label{thm:MDI-Grow-Rate}
   There exists a sequence of MDI quantum key growing schemes following Protocol \ref{prot:app-MDI-Q-Key-Scheme} where for each $n \in \mbb{N}$ the input is $\rho^{0} = \rho_{AB}^{\otimes n}$, it is $\ve_{s,n}$-sound, and $\ve_{c,n}$-complete such that
    \begin{enumerate}
       \item The probability of aborting on the honest implementation vanishes,  i.e.~$ \lim_{n \to \infty} \ve_{c,n} = 1$
       \item The insecurity of the distilled key vanishes, $\lim_{n \to \infty} \ve_{s,n} = 0$, and
       \item the amount of key extracted per round (the rate) is
       \begin{align*}
        R 
        &\geq H(X \vert ABE)_{(\cE \otimes \id_{E})(\psi)} - f_{\text{ec}}H(W\vert Y)_{q} \ , 
        \end{align*}
        where $f_{\text{ec}} \geq 1$ and $q_{WY}$ is the joint distribution induced by the protocol when the receiver is honest (i.e. applies an expected measurement $\cM_{AB \to Z}$ each round).
    \end{enumerate}
\end{theorem}
\begin{proof}
    First, we show it suffices to work with our choice of purification of $\rho_{AB}$. By assumption $\rho_{A^{n}B^{n}} = \rho_{AB}^{\otimes n}$. There exists a purification $\psi_{ABE}$ of $\rho_{AB}$ where $\vert E \vert = \vert A \vert \vert B \vert$. As purifications are unitarily equivalent and we can embed them in larger spaces, for any choice of purification $\wt{\psi}_{A^{n}B^{n}\wt{E}}$, there is an isometry $V_{E^{n} \to E}$ such that $V\psi^{\otimes n}V^{\dagger} = \wt{\psi}$. Thus, letting $\rho^{1}$ denote \eqref{eq:rho-1-prot-process} when the input is $\psi^{\otimes n}$ and $\wt{\rho}^{1}$ when the input is $\wt{\psi}$, we have
    \begin{align}
        H^{\ve}_{\min}(X^{2n}_{1} \vert A^{n}_{1}B^{n}_{1}E^{n}_{1})_{\rho^{1}}
        = H^{\ve}_{\min}(X^{2n}_{1} \vert A^{n}B^{n}\tilde{E})_{V\rho^{1}V^{\dagger}} 
        = H^{\ve}_{\min}(X^{2n}_{1} \vert A^{n}B^{n}\tilde{E})_{\wt{\rho}^{1}} \ ,
    \end{align} 
    where we used that $\cE^{1}$ and $V \cdot V^{\dagger}$ are CPTP maps that act on different spaces and that the smooth min-entropy is invariant under isometries that just act on the conditioning space. Thus, without loss of generality, we can consider $H^{\ve}_{\min}(X^{2n}_{1} \vert A^{n}_{1}B^{n}_{1}E^{n}_{1})_{\rho^{1}}$ for our choice of $\psi_{ABE}$ that purifies $\rho_{AB}$.

    Next, allow parameters $\ve_{n},\ve_{\text{sec},n}$ to now depend on $n \in \mbb{N}$ as this will be used subsequently. Also recall from the proof of Lemma \ref{lem:one-shot-bound} that $\cE^{1} = \cE^{\otimes n}$, so $\rho^{1} = [(\cE \otimes \id_{E})(\psi)]^{\otimes n}$. Thus,
    \begin{align}   
        H^{\ve_{n}}_{\min}(X^{2n}_{1} \vert A^{n}_{1}B^{n}_{1}E^{n}_{1})_{\rho^{1}}
        &= H^{\ve_{n}}_{\min}(X^{2n}_{1} \vert A^{n}_{1}B^{n}_{1}E^{n}_{1})_{[(\cE \otimes \id_{E})(\psi)]^{\otimes n}} \\ 
        &\geq nH(X^{2}_{1} \vert ABE)_{(\cE \otimes \id_{E})(\psi)}  - 4\log(C_{1})\sqrt{\log(2/\ve_{n}^{2})}  \ ,
    \end{align}
    where the inequality is the quantum asymptotic equipartition property \cite[Theorem 9]{Tomamichel-2009a} and $C_{1} > 0$. Thus, using Lemma \ref{lem:one-shot-bound}, we obtain an achievable rate
    \begin{align}
        R = \lim_{n \to \infty} \frac{\ell_{n}}{n} 
        \geq  H(X^{2}_{1} \vert ABE)_{(\cE \otimes \id_{E})(\psi)} 
        - \lim_{n \to \infty} \frac{z_{n}}{n} \ , 
    \end{align}
    where
    \begin{align}
        z_{n} &\coloneq 2\log(\frac{1}{4(\ve_{\text{sec},n}-\ve_{n})})  + 4\log(C_{1})\sqrt{\log(2/\ve_{n}^{2})} + t_{n} + r_{n} \ , 
    \end{align}
    and $t_{n}$ and $r_{n}$ are the bits of error correction verification and error correction syndrome respectively when the input is $n$-fold.
    
    We now just need to control $\ve_{n}$, $\ve_{\text{sec},n}$, $t_{n}$, and $r_{n}$. First, let $\ve_{n} = \frac{1}{\sqrt{n}}$, $\ve_{\text{sec},n} = \frac{2}{\sqrt{n}}$ so that $\ve_{\sec,n} - \ve_{n} = \frac{1}{\sqrt{n}}$. Then the first term of $z_{n}$ is $2\log(\frac{4}{\sqrt{n}}) = o(n)$ and thus will vanish in the limit when divided by $n$. Second, let $t_{n} = \sqrt{n}$ so that by \cite[Theorem 2]{Tomamichel-2017a} $\ve_{ec} \leq 2^{-t_{n}} = 2^{-\sqrt{n}} \to 0$, but $\frac{t_{n}}{n} \to 0$. Taking what we have so far, this implies $\lim_{n \to \infty} \ve_{s,n} \leq \lim_{n \to \infty} [\ve_{\text{sec},n} + \ve_{\ve_{ec,n}}] = 0$. This establishes the soundness error goes to zero as claimed. Note this did not depend on our choice of the sequence $(r_{n})$.
    
    Finally, we need to show there exist choices for the error correction that achieve the claimed rate and guarantee completeness error goes to zero. Note that in the honest case, the random variable $Z$ is drawn according to $q_{Z} = \cR_{AB \to Z}(\rho_{AB})$ per round. Thus there is joint distribution $q_{WY}$ between Alice's $W_{2i-1}W_{2i} \oplus Z_{2i-1}Z_{2i}$ and Bob's $Y_{2i-1}Y_{2i}$ per round $i$ and it is i.i.d. in total. By the Slepian-Wolf coding theorem \cite{slepian1973noiseless}, for all $f_{\text{ec}} \geq 1$, there exists a sequence of one-way error correcting codes such that $\lim_{n \to \infty} \frac{r_{n}}{n} = f_{\text{ec}}H(W \vert Y)_{q}$ and the probability of failing to successfully correct the errors vanishes. Thus, we can use any such a sequence of one-way error correcting codes. As Alice and Bob only abort in Protocol \ref{prot:main-text-MDI-Q-Key-Scheme} if error correction fails, this implies $\ve_{c,n} \to 0$. Moreover, putting these points together, we have $\lim_{n \to \infty} \frac{z_{n}}{n} = f_{\text{ec}}H(W \vert Y)_{q}$. Combining terms completes the proof.
\end{proof}

\subsubsection{Calculations for Fig.~\ref{fig:KeyGrowingVsKeyDistillation}}
Finally, we explain the derivations relevant for Fig.~\ref{fig:KeyGrowingVsKeyDistillation}. The code that makes the calculations based on the derivations is available at \href{https://github.com/qit-george/QuantumDenseNetworkCoding}{this GitHub repository}.

We first derive the special case of the rate for Werner states.
\begin{corollary}\label{cor:Werner-State-Key-Growing}
    In Theorem \ref{thm:MDI-Grow-Rate}, if $\rho_{AB} = \rho_{\lambda} = (1-\lambda) \dyad{\Phi^{+}}_{AB} + \lambda \pi_{AB} \ , $, i.e. the input is many copies of a Werner state, then the rate bound is given by
    \begin{align}
        R \geq 2 - \frac{1}{4}\sum_{x_{1},x_{2} \in \{0,1\}} D(\sigma^{x_{1},x_{2}} \Vert \ol{\sigma}) - H(W|Y)_{q} \ ,
    \end{align}
    where $\sigma_{ABE} = \frac{1}{4} \sum_{y_{1},y_{2} \in \{0,1\}} (I_{A} \otimes W_{y_{1},y_{2}} I_{E}) \psi_{\lambda}(I_{A} \otimes W_{y_{1},y_{2}}^{\dagger} \otimes I_{E})$, $\sigma_{ABE}^{x_{1},x_{2}} = U_{x,y}\sigma U_{x}^{\dagger}$, $\ol{\sigma} = \frac{1}{4}\sum_{x_{1},x_{2} \in \{0,1\}} \sigma^{x_{1},x_{2}}$, and $\psi_{\lambda}$ is a purification of $\rho_{\lambda}$. In particular, this means when $\lambda = 0$, the rate is $R = 2$.
\end{corollary}
\begin{proof}
 We start with our choice of purification of $\rho_{\lambda}$ to simplify calculations. We concisely express the qubit Bell states as
    \begin{align}\label{eq:bell-states}
        \begin{matrix}
        \ket{B_0} = \ket{\Phi^+} & \ket{B_1} = \ket{\Phi^-}  & \ket{B_{2}} = \ket{\Psi^+} & \ket{B_{3}} = \ket{\Psi^-}
        \end{matrix} \ . 
    \end{align} 
Then a purification of $\rho_{\lambda}$ is given  by
\begin{align}\label{eq:Werner-purification}
    \ket{\psi_{\lambda}}_{ABE} \coloneq \sum_{i} \sqrt{q_{i}} \ket{B_{i}}\ket{i} \ , 
\end{align}
where $q_{0}= 1 - \frac{3\lambda}{4}$ and $q_{1} = q_{2} = q_{3} = \frac{\lambda}{4}$. It follows after the encoding
\begin{align}
    \rho_{X^{2}_{1}Y^{2}_{1}ABE} 
    = \frac{1}{16} \sum_{x^{2}_{1},y^{2}_{1}} \dyad{x^{2}_{1}} \otimes \dyad{y^{2}_{1}} \otimes W_{x_{1},x_{2}} \otimes W_{y_{1},y_{2}} \dyad{\psi_{\lambda}} W_{x_{1},x_{2}}^{\dagger} \otimes W_{y_{1},y_{2}}^{\dagger} \ .
\end{align}
We now just need to determine this state's form so that we may evaluate the conditional entropy on it. For notational simplicity, let $x = x^{2}_{1}$, $y = y^{2}_{1}$, so
\begin{align}
    \rho_{X^{2}_{1}ABE}
    = \frac{1}{4} \sum_{x} \dyad{x} \otimes W_{x_{1},x_{2}} \left[ \sum_{y} \frac{1}{4} W_{y_{1},y_{2}} \dyad{\psi_{\lambda}} W_{y_{1},y_{2}}^{\dagger} \right]W_{x_{1},x_{2}}^{\dagger}
    \coloneq \frac{1}{4} \sum_{x} \dyad{x} \otimes \sigma^{x}_{ABE}  \ .
\end{align}
We then have
\begin{align*}
    H(X^{2}_{1} \vert ABE)_{\rho}
    &= -D(\rho_{X^{2}_{1}ABE} \Vert I_{X^{2}_{1}} \otimes \rho_{ABE}) \\
    &= -\left\{D\left(\sum_{x} \frac{1}{4}\dyad{x} \otimes \rho_{ABE}^{x} \Vert \pi_{X} \otimes \ol{\sigma}\right) - \log(\vert \cX \vert) \right\} \\
    &=  -\left\{ \sum_{x} \frac{1}{4} D(\sigma^{x} \Vert \ol{\sigma}) + D(\pi \Vert \pi) - \log(|\cX|)   \right\} \\
    &=  -\left\{ \frac{1}{4}\sum_{x} D(\sigma^{x} \Vert \ol{\sigma}) - 2   \right\} \\
    &=  2 - \frac{1}{4}\sum_{x} D(\sigma^{x} \Vert \ol{\sigma})
\end{align*}
where the second equality follows from the normalization property of relative entropy \cite{Tomamichel-Book}, the third follows from \cite[Exercise 11.8.8]{Wilde-Book} and that the distribution on $X$ is uniform, and the fourth uses $\vert \cX \vert = 4$. This obtains the corollary. The claimed rate when $\lambda = 0$ follows from the fact a direct calculation will verify, when $\lambda = 0$, $\sigma_{ABE} = \pi_{AB} \otimes \dyad{0}_{E}$ and $\rho_{X^{2}_{1}ABE} = \pi_{X} \otimes \sigma_{ABE}$.
\end{proof}

To calculate the $H(W \vert Y)_{q}$ term for Fig.~\ref{fig:KeyGrowingVsKeyDistillation}, we assume the receiver implements the ideal Bell measurement, $\{\dyad{B_{i}}\}_{i \in \{0,1,2,3\}}$ in the notation of \eqref{eq:bell-states}, but that Alice and Bob are connected to the receiver by depolarizing channels $\cD_{p}(X) = (1-p)X + \frac{\Tr[X]}{2}I$. Thus,
\begin{align*}
    p_{XYZ}
    &= \sum_{x,y,z \in \{0,1\}^{2}} \bra{B_{z}}W_{x_{1},x_{2}} \otimes W_{y_{1},y_{2}}\rho_{\lambda}W_{x_{1},x_{2}}^{\dagger} \otimes W_{y_{1},y_{2}}^{\dagger}\ket{B_{z}} \dyad{x} \otimes \dyad{y} \otimes \dyad{z} \ . 
\end{align*}
One may then compute $p_{WY}$ from this by determining which $(x,z)$ are such that $w = x \oplus z$. This allows one to then calculate $H(W \vert Y)_{q}$. This completes the explanation of the key growing curves in Fig.~\ref{fig:KeyGrowingVsKeyDistillation}.

\paragraph{Key Distillation Rate} The key distillation rate curve computes the Devetak-Winter formula \cite{devetak2005distillation} on the chosen Werner state, which we note is also the asymptotic rate of a sequence of QKD protocols that only accept the specific Werner state in question. In other words, we compute
\begin{align*}
    H(X \vert E)_{\cM_{A \to X}(\rho_{AE})}- H(X \vert Y)_{(\cM_{A \to X} \otimes \cN_{B \to Y})(\rho_{\lambda})} \ ,
\end{align*}
where $\rho_{AE} = \Tr_{B}[\psi_{\lambda}]$, $\psi_{\lambda}$ is the purification in \eqref{eq:Werner-purification}, and the measurements $\cM_{A \to X}$, $\cN_{B \to Y}$ both are just projective measurements in the computational basis $\{\dyad{0},\dyad{1}\}$. To compute $H(X \vert E)$, one computes $\rho_{XE}$, $\rho_{E}$, and then computes $H(X \vert E)_{\rho} = H(XE) - H(E)$ using the eigenvalues of $\rho_{XE}$ and $\rho_{E}$. For simplicity, we numerically compute $\rho_{XE}$, $p_{XY} \coloneq (\cM_{A \to X} \otimes \cN_{B \to Y})(\rho_{\lambda})$, and the corresponding entropic quantities. The code used may be found in \href{https://github.com/qit-george/QuantumDenseNetworkCoding}{this GitHub repository}. This completes the account of how Fig.~\ref{fig:KeyGrowingVsKeyDistillation} is generated.

% \paragraph{How to Explain the Above as a QKD Rate}
% To compute $H(X \vert E)$, one computes $\rho_{XE}$, $\rho_{E}$ and then can compute $H(X \vert E)_{\rho} = H(XE) - H(E)$ using the eigenvalues of $\rho_{XE}$ and $\rho_{E}$. For simplicity, we compute $\rho_{XE}$ and $p_{XY} \coloneq (\cM_{A \to X} \otimes \cN_{B \to Y})(\rho_{\lambda})$ and the corresponding entropic quantities numerically. This completes the account of how Fig.~\ref{fig:KeyGrowingVsKeyDistillation} is generated.

% This can be justified using the following well-known methods that reduce the asymptotic rate of a QKD protocol to this expression. We consider a sequence of entanglement-based protocols indexed by the number of copies of state $n$ where:
% \begin{itemize}
%     \item For each signal, Alice and Bob measure in the $X$, $Y$, and $Z$-bases with probabilities $p_{X}^{n}$, $p_{Y}^{n}$, and $p_{Z}^{n}$ respectively such that $\lim_{n \to \infty} p_{X}^{n} = \lim_{n \to \infty} p_{Y}^{n} = 0$ and $\lim_{n \to \infty} p_{Z}^{n} = 1$. 
%     \item Key is only ever extracted from when Alice and Bob both measured in the $Z$-basis. 
%     \item The protocol aborts unless the observed statistics are what would be generated by the Werner state $\rho_{\lambda}$ according to Born's rule. 
% \end{itemize}
% Under these conditions, parameter estimation reconstructs the state, the protocol aborts unless it is the specified Werner state, and key is extracted every round as $\lim_{n \to \infty} p_{Z}^{n} = 1$\cite{lo2005efficient}. 

\end{document}